\documentclass[english,11pt]{article}
\usepackage[T1]{fontenc}
\usepackage[latin9]{inputenc}
\usepackage{geometry}
\usepackage{fancyhdr}
\usepackage{color}
\usepackage{mathrsfs}
\usepackage{mathtools}
\usepackage{algorithm2e}
\usepackage{amsmath}
\usepackage{amsthm}
\usepackage{amssymb}
\usepackage{stmaryrd}
\usepackage{graphicx}

\makeatletter
\newcommand{\lyxaddress}[1]{
\par {\raggedright #1
\vspace{1.4em}
\noindent\par}
}
\theoremstyle{plain}
\newtheorem{thm}{\protect\theoremname}
  \theoremstyle{definition}
  \newtheorem{defn}[thm]{\protect\definitionname}
  \theoremstyle{plain}
  \newtheorem{lem}[thm]{\protect\lemmaname}
  \theoremstyle{plain}
  \newtheorem{conjecture}[thm]{\protect\conjecturename}

\usepackage{amsmath}
\usepackage{amsthm}
\usepackage{amssymb}
\usepackage{hyperref}
\usepackage{graphics}
\usepackage{algorithmic}
\usepackage{subfig}
\usepackage{url}
\usepackage{bm}

\usepackage{setspace}

\usepackage{slashed}

\usepackage{esint}
\usepackage{stmaryrd}

\makeatletter

\usepackage{amsfonts}

\usepackage{framed,color}

\definecolor{shadecolor}{gray}{0.9}

\usepackage{ulem}

\makeatother

\makeatother

\usepackage{babel}
  \providecommand{\conjecturename}{Conjecture}
  \providecommand{\definitionname}{Definition}
  \providecommand{\lemmaname}{Lemma}
\providecommand{\theoremname}{Theorem}

\begin{document}

\title{Radiation reaction on a Brownian scalar electron \\
in high-intensity fields}

\author{{\Large{}Keita Seto}\thanks{keita.seto@eli-np.ro}}
\maketitle

\lyxaddress{\begin{center}
Extreme Light Infrastructure \textendash{} Nuclear Physics (ELI-NP)
/ \\
Horia Hulubei National Institute for R\&D in Physics and Nuclear Engineering
(IFIN-HH), \\
30 Reactorului St., Bucharest-Magurele, jud. Ilfov, P.O.B. MG-6, RO-077125,
Romania.
\par\end{center}}
\begin{abstract}
Radiation reaction against a relativistic electron is of critical
importance since the experiment to check this ``quantumness'' becomes
possible soon with an extremely high-intensity laser beam. However,
there is a fundamental mathematical quest to apply any laser profiles
to laser focusing and superposition beyond the Furry picture of its
usual method by a plane wave. To give the apparent meaning of $q(\chi)$
the quantumness factor with respect to a radiation process is absent.
Thus for resolving the above questions, we propose stochastic quantization
of the classical radiation reaction model for any laser field profiles,
via the construction of the relativistic Brownian kinematics with
the dynamics of a scalar electron and the Maxwell equation with a
current by a Brownian quanta. This is the first proposal of the coupling
system between a relativistic Brownian quanta and fields in Nelson's
stochastic quantization. Therefore, we can derive the radiation field
by its Maxwell equation, too. This provides us the fact that $q(\chi)$
produced by QED is regarded as $\mathscr{P}(\varOmega_{\tau}^{\mathrm{ave}})$
of an existence probability such that a scalar electron stay on its
average trajectory.
\end{abstract}
\pagestyle{fancy}
\lhead{}\newpage{}\tableofcontents

\section{Introduction\label{Intro}}

In this paper, we investigate {\bf ``Radiation reaction (RR)''} acting
on a scalar electron by stochastic quantization, namely, quantum dynamics
of a radiating quanta with its Brownian and relativistic kinematics
({\bf Fig.\ref{FigRRonBM}}). Then, we clarify the fact that $q(\chi)$
the quantumness of RR is $\mathscr{P}(\varOmega_{\tau}^{\mathrm{ave}})$
an existence probability of a scalar electron given by this Brownian
kinematics. 

RR is expected to be fully investigated experimentally \cite{RA5,Sarri(2014),Cole(2017)}
by collisions of a high-intensity laser \cite{ELI-NP,ELI-NP(2017),GIST(2017)}
and a high-energy electron soon. This mechanism is regarded as the
higher-order correction or the almost same effect of a non-linear
Compton scattering \cite{Brown-Kibble(1964),Nikishov(1964a)),Nikishov(1964b)}
evaluated by the Furry picture \cite{Furry(1951)} in the recent laser-plasma
physics. Its radiation formula including {\bf Q}uantum {\bf E}lectro{\bf d}ynamics
(QED) or scalar QED effects is derived from this non-linear Compton
scattering \cite{A.Sokolov(1986),I.Sokolov(2010),I.Sokolov(2011a),Berestetskii-Lifshitz-Pitaevskii},
namely in QED, 
\begin{equation}
\frac{dW_{\mathrm{Quantum}}}{dt}=q(\chi)\times\frac{dW_{\mathrm{classical}}}{dt}\label{eq: Rad-formula-sQED}
\end{equation}
assisted by its quantumness $q(\chi)$;
\begin{align}
q(\chi) & =\frac{9\sqrt{3}}{8\pi}\int_{0}^{\chi^{-1}}dr\,r\int_{\frac{r}{1-\chi r}}^{\infty}dr'K_{5/3}(r')+\frac{9\sqrt{3}}{8\pi}\int_{0}^{\chi^{-1}}dr\,\frac{\chi^{2}r^{3}}{1-\chi r}K_{2/3}\left(\frac{r}{1-\chi r}\right),\label{eq: q}
\end{align}
\begin{align}
\chi & \coloneqq\frac{3}{2}\frac{\hbar}{m_{0}^{2}c^{3}}\sqrt{-g_{\mu\nu}(-eF_{\mathrm{ex}}^{\mu\alpha}v_{\alpha})(-eF_{\mathrm{ex}}^{\nu\beta}v_{\beta})}\nonumber \\
 & \propto(\mathrm{electron\,energy})\times\sqrt{\mathrm{laser\,intensity}}\,.
\end{align}
The case $q(\chi)=1$ for $\chi\sim0$ is regarded as no quantum correction.
On the other hand if one uses an extremely high-intensity laser such
as the $10\mathrm{PW}$ laser of ELI-NP \cite{RA5,ELI-NP,ELI-NP(2017)},
the quantum correction $q(\chi)=0.3$ appears for the laser intensity
of $10^{22}\mathrm{W}/\mathrm{cm}^{2}$ and an electron energy of
$600\mathrm{MeV}$ \cite{RA5,Seto(2015)} (see {\bf Fig.\ref{q_chi}}).
Therefore, the investigation of $q(\chi)$ w.r.t.\ RR links to the
interest in high-intensity laser science. However, Eq.(\ref{eq: Rad-formula-sQED})
is derived by the Furry picture to employ a laser profile of a plane
wave \cite{Zakowicz(2005),Boca-Florescu(2010)}. Hence, there are
several proposals for the non-plane wave condition of laser focusing
and superposition \cite{Piazza(2014),Piazza(2017a),Piazza(2017b)}. 

Anyway, the effective regime of RR locates at quantized, relativistic
and high-intensity field interactions marked by ``the star'' in {\bf Fig.\ref{figHIFP}}.
This Furry picture is the way from relativistic quantum dynamics.
The one from classical, relativistic and high-intensity regime may
be another candidate. Let us consider this second candidate for any
laser profiles. 
\begin{figure}
\noindent \centering{} \includegraphics[scale=0.25]{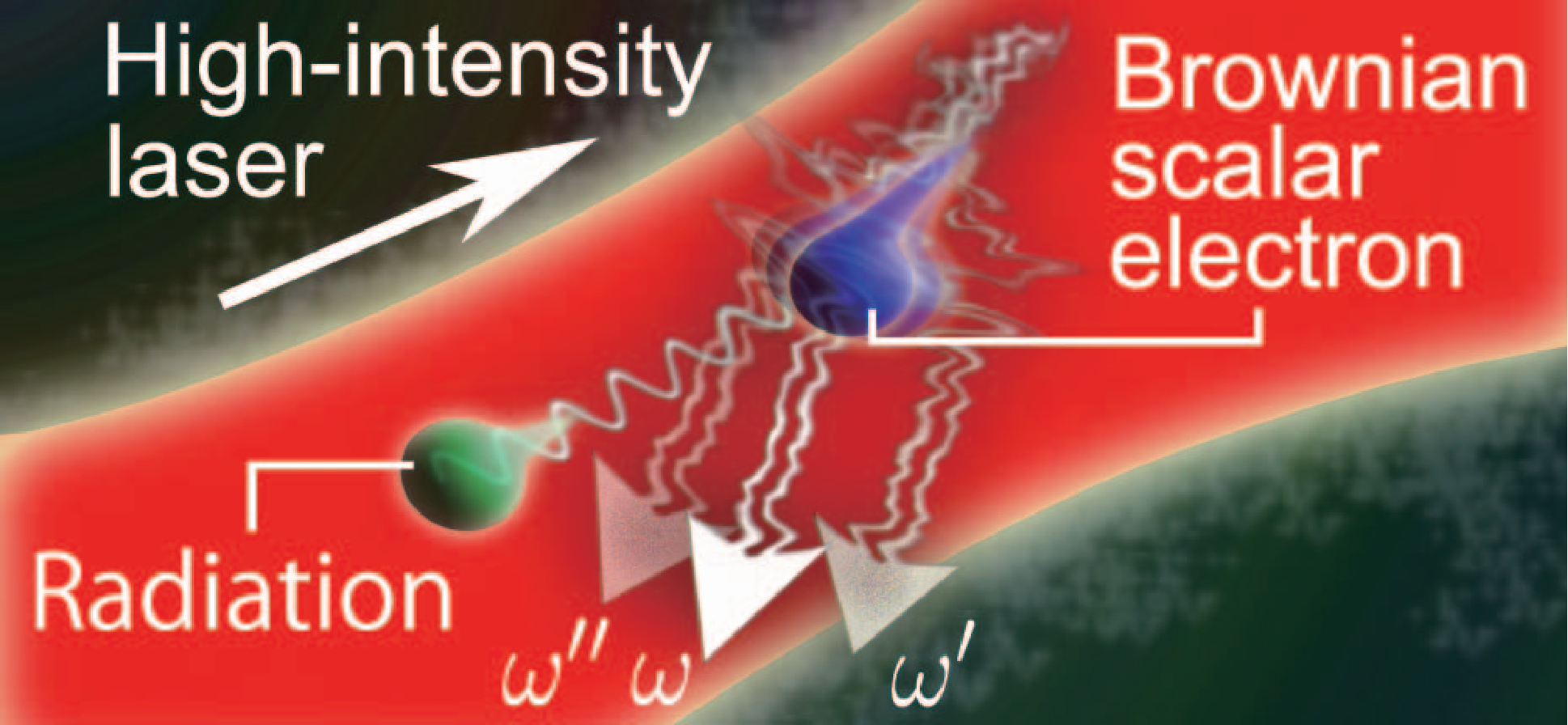}\caption{\label{FigRRonBM}  RR on a Brownian scalar electron. The signatures
of $\omega$, $\omega'$ and $\omega''$ denote sample paths of a
scalar electron due to its randomness. By solving the Maxwell's equation,
RR in quantum dynamics is derived like in classical dynamics. }
 
\end{figure}
\begin{figure}
\noindent \centering{}\includegraphics[scale=1.2]{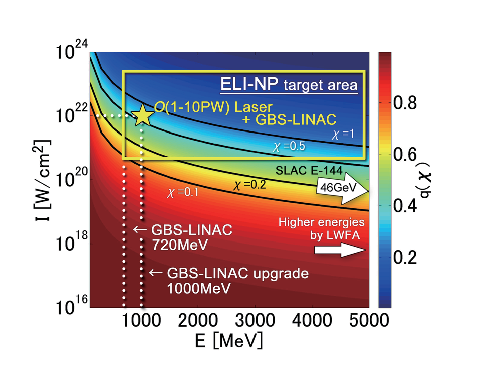}\caption{\label{q_chi}Quantumness of radiation, $q(\chi)$. This is the plot
of Eq.(\ref{eq: q}) w.r.t.\ energies of an electron and laser intensities.
When we choose the combination between a laser intensity of $O(10^{22}\mathrm{W}/\mathrm{cm}^{2})$
and an electron energy of $O(1\mathrm{GeV})$, the factor $q(\chi)$
reaches $0.3$ which is the feasible regime produced by the ELI-NP
facility \cite{RA5}. The SLAC E-144 included the experiments of the
non-linear Compton scatterings by the combination of\textcolor{red}{{}
}$O(10^{18}\mathrm{W}/\mathrm{cm}^{2}+46\mathrm{GeV})$ \cite{SLAC}. }
\end{figure}
\begin{figure}
\centering{}\includegraphics[scale=0.4]{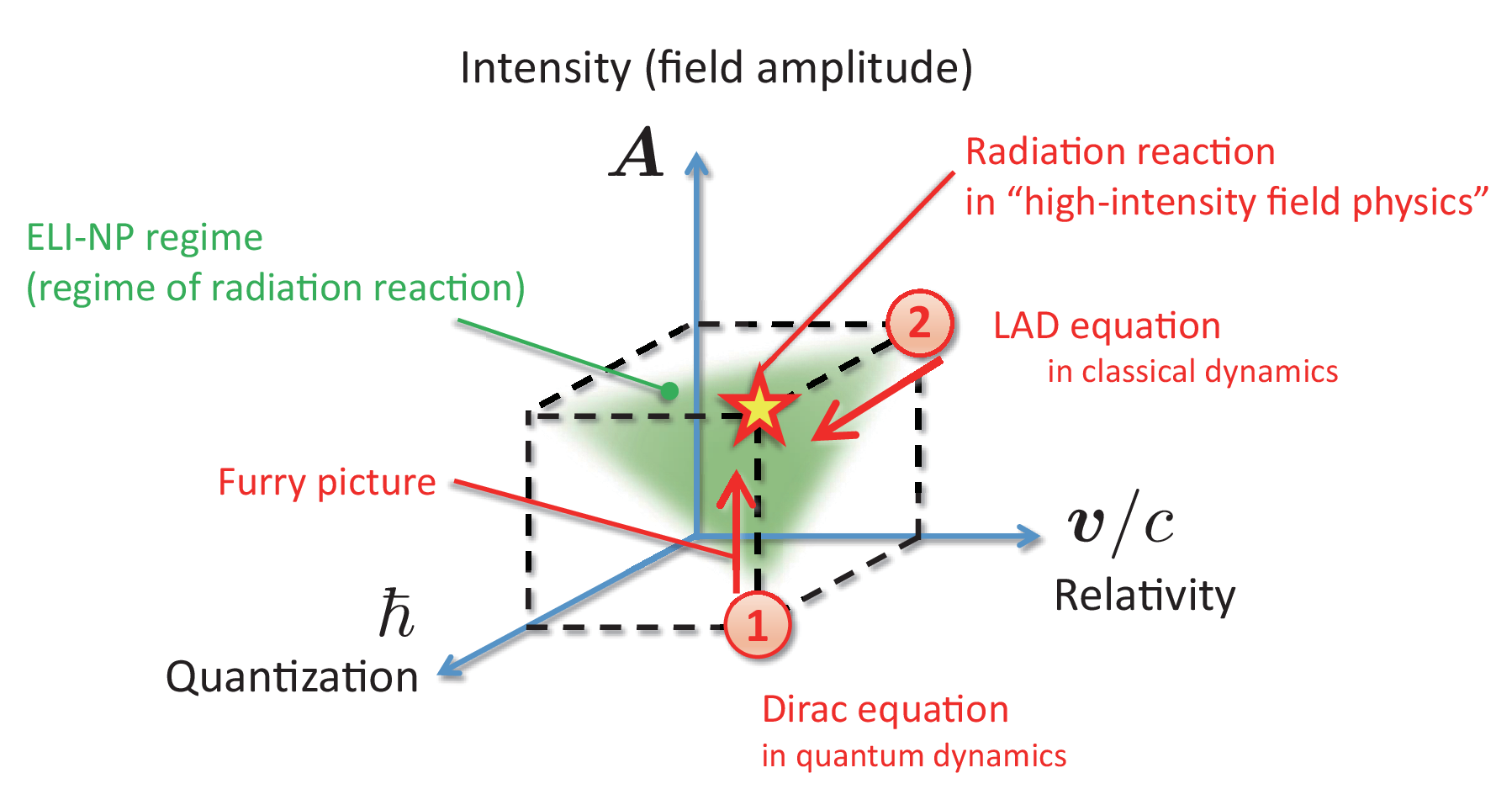}\caption{\label{figHIFP}The physical regime of a scalar electron with RR.
The point at ``the star'' is the regime of high-intensity field physics.
The Furry picture is the first way from relativistic quantum dynamics
to the point of ``the star'' where $q(\chi)$ is not unity. On the
other hand, the quantization after reaching high-intensity ``classical''
dynamics is its second candidate.}
\end{figure}

For the second idea, we should refer classical RR model. RR has been
treated by the Lorentz-Abraham-Dirac (LAD) equation as its standard
model \cite{Dirac(1938)}:
\begin{align}
m_{0}\frac{dv^{\mu}}{d\tau} & =-e(F_{\mathrm{ex}}^{\mu\nu}+F_{\mathrm{LAD}}^{\mu\nu})v_{\nu}\label{eq: LAD eq}
\end{align}
\begin{align}
F_{\mathrm{LAD}}^{\mu\nu}(x) & =-\frac{m_{0}\tau_{0}}{ec^{2}}\left[\frac{d^{3}x^{\mu}}{d\tau^{3}}\cdot\frac{dx^{\nu}}{d\tau}-\frac{d^{3}x^{\nu}}{d\tau^{3}}\cdot\frac{dx^{\mu}}{d\tau}\right]\label{eq: LAD field}
\end{align}
With the metric $g=(+,-,-,-)$ and $\tau_{0}\coloneqq e^{2}/6\pi\varepsilon_{0}m_{0}c^{3}$.
The force $-eF_{\mathrm{LAD}}^{\mu\nu}v_{\nu}$ represents the interaction
of RR acting on a scalar electron. When we can find the quantization
of Eq.(\ref{eq: LAD eq}) with Eq.(\ref{eq: LAD field}), it is not
known how this affects Eq.(\ref{eq: Rad-formula-sQED}), especially,
$q(\chi)$ by any laser field profiles. We are then to study quantum
dynamics of RR by adopting Nelson's stochastic quantization \cite{Nelson(1966a),Nelson(2001_book)}
as the second candidate in {\bf Fig.\ref{figHIFP}}, which can draw
a real trajectory of a quanta by a Brownian motion. In addition, that
achievement regarded as one of the few application of Nelson's stochastic
quantization.

However, its well-defined relativistic version and the Maxwell equation
have been absent. Thus, we resolve the following issues for RR in
this article. For $\{\hat{x}(\tau,\omega)\}_{\tau\in\mathbb{R}}$
a Brownian trajectory, $\langle\mathrm{A}\rangle$ the set of the
relativistic Brownian kinematics and the dynamics for a scalar electron
\begin{align}
d_{\pm}\hat{x}^{\mu}(\tau,\omega) & =\mathcal{V}_{\pm}^{\mu}(\hat{x}(\tau,\omega))d\tau+\lambda\times dW_{\pm}^{\mu}(\tau,\omega)\label{eq: intro-kinematics}
\end{align}
\begin{align}
m_{0}\mathfrak{D_{\tau}}\mathcal{V}^{\mu}(\hat{x}(\tau,\omega)) & =-e\mathcal{\hat{V}}_{\nu}(\hat{x}(\tau,\omega))F^{\mu\nu}(\hat{x}(\tau,\omega))\label{eq: intro-KG}
\end{align}
is regarded as the Klein-{\allowbreak}Gordon ({\bf KG}) equation
with $\mathcal{V}^{\mu}(x)\coloneqq1/m_{0}\times\left[i\hbar\partial{}^{\mu}\ln\phi(x)+eA{}^{\mu}(x)\right]$
given by its wave function $\phi$ and $\{W_{\pm}(\tau,\omega)\}_{\tau\in\mathbb{R}}$
of Wiener processes {[}{\bf Sect.\ref{kinematics_dynamics4D}}{]}.
And also $\langle\mathrm{B}\rangle$ the Maxwell equation with a current
of a Brownian scalar electron {[}{\bf Sect.\ref{Maxwell}}{]}
\begin{align}
\partial_{\mu}[F^{\mu\nu}(x)+\delta f^{\mu\nu}(x)] & =\mu_{0}\times\mathbb{E}\left\llbracket -ec\int_{\mathbb{R}}d\tau\,\mathrm{Re}\left\{ \mathcal{V}^{\nu}(x)\right\} \delta^{4}(x-\hat{x}(\tau,\bullet))\right\rrbracket \label{eq: intro-Maxwell}
\end{align}
is constructed. We propose $\langle\mathrm{C}\rangle$ an action integral
to give the above Eq.(\ref{eq: intro-KG}) and Eq.(\ref{eq: intro-Maxwell})
in this model, too {[}{\bf Sect.\ref{Action}}{]}. In the fact, the
consistent system of Eqs.(\ref{eq: intro-kinematics}-\ref{eq: intro-KG})
and Eq.(\ref{eq: intro-Maxwell}) has been absent after Nelson's first
article \cite{Nelson(1966a)}. Especially, the charge current in Eq.(\ref{eq: intro-Maxwell})
adapting stochastic quantization has not been discovered for a long
time. So, this is the first proposal for the coupling system between
a relativistic Brownian quanta and fields. To describe its interaction
is, therefore, a new result. Namely by solving Eq.(\ref{eq: intro-Maxwell})
in {\bf Sect.\ref{RR}}, we derive $\langle\mathrm{D}\rangle$ RR
in quantum dynamics
\begin{align}
m_{0}\mathfrak{D_{\tau}}\mathcal{V}^{\mu}(\hat{x}(\tau,\omega))) & =-eF_{\mathrm{ex}}^{\mu\nu}(\hat{x}(\tau,\omega)\mathcal{V}_{\nu}(\hat{x}(\tau,\omega))-e\mathfrak{F}^{\mu\nu}(\hat{x}(\tau,\omega))\mathcal{V}_{\nu}(\hat{x}(\tau,\omega))\label{eq: intro-stRR}
\end{align}
\begin{align}
\mathfrak{F}^{\mu\nu}(\hat{x}(\tau,\omega)) & =-\frac{m_{0}\tau_{0}}{ec^{2}}\int_{\varOmega_{(\tau,\omega)}}d\mathscr{P}(\omega')\left[\begin{gathered}\dot{a}^{\mu}(\hat{x}(\tau,\omega'))\cdot\mathrm{Re}\{\mathcal{V}^{\nu}(\hat{x}(\tau,\omega'))\}\\
-\dot{a}^{\nu}(\hat{x}(\tau,\omega'))\cdot\mathrm{Re}\{\mathcal{V}^{\mu}(\hat{x}(\tau,\omega'))\}
\end{gathered}
\right]\label{eq: intro-field}
\end{align}
with $\dot{a}^{\mu}(x)\approx\mathrm{Re}\{\mathfrak{D_{\tau}}^{2}\mathcal{V}^{\mu}(x)\}$,
i.e., the quantization of the LAD equation (\ref{eq: LAD eq}-\ref{eq: LAD field}).
The readers can find the similarity between Eqs.(\ref{eq: LAD eq}-\ref{eq: LAD field})
and Eqs.(\ref{eq: intro-stRR}-\ref{eq: intro-field}) by the comparison.
In the fact, $\lim_{\hbar\rightarrow0}\mathfrak{F}=F_{\mathrm{LAD}}$
is ensured. Thus, Eqs.(\ref{eq: intro-stRR}-\ref{eq: intro-field})
becomes Eqs.(\ref{eq: LAD eq}-\ref{eq: LAD field}) in the classical
limit. This easy comparison is the reason why we study stochastic
quantization for RR. By its Ehrenfest's theorem, $\langle\mathrm{E}\rangle$
a radiation formula
\begin{align}
\frac{dW_{\mathrm{stochastic}}}{dt} & =-m_{0}\tau_{0}\mathscr{P}(\varOmega_{\tau}^{\mathrm{ave}})\frac{d^{2}\langle\hat{x}_{\mu}\rangle_{\tau}}{d\tau^{2}}\cdot\frac{d^{2}\langle\hat{x}^{\mu}\rangle_{\tau}}{d\tau^{2}}\label{eq: intro-Rad-formula}
\end{align}
similar to Eq.(\ref{eq: Rad-formula-sQED}) is imposed. This shows
the fact that $q(\chi)$ is $\mathscr{P}(\varOmega_{\tau}^{\mathrm{ave}})$
a probability which a scalar electron stays at its average position.
Finally, a possibility of its higher-order corrections is discussed. 

Let us note a naive idea of quantum dynamics. The present proposal
does not deny the previous formulations of quantum dynamics. As K.\ Yasue
suggests \cite{Yasue}, quantum dynamics is symbolically illustrated
by
\begin{align*}
(\mathrm{Quantum\,dynamics}) & =(\mathrm{Matrix\,mechanics})\cup(\mathrm{Wave\,mechanics})\cup(\mathrm{Path\,integral})\\
 & \quad\quad\cup(\mathrm{Stochastic\,quantization})\cup(\mathrm{Something\,else}).
\end{align*}
So, each expressions of quantum dynamics are complementary via a wave
function.

Before discussing the relativistic regime, let us summarize Nelson's
model in the non-relativistic regime by a 1D stochastic process (an
(S3)-processes) in the following {\bf Sect.\ref{Nelson-1D}}.

\section{Stochastic kinematics and dynamics in non-relativistic regime by
a 1D stochastic process\label{Nelson-1D}}

\subsection{Kinematics}

\begin{figure}
\begin{centering}
\begin{minipage}[t]{0.5\columnwidth}%
\begin{center}
\includegraphics[scale=0.5]{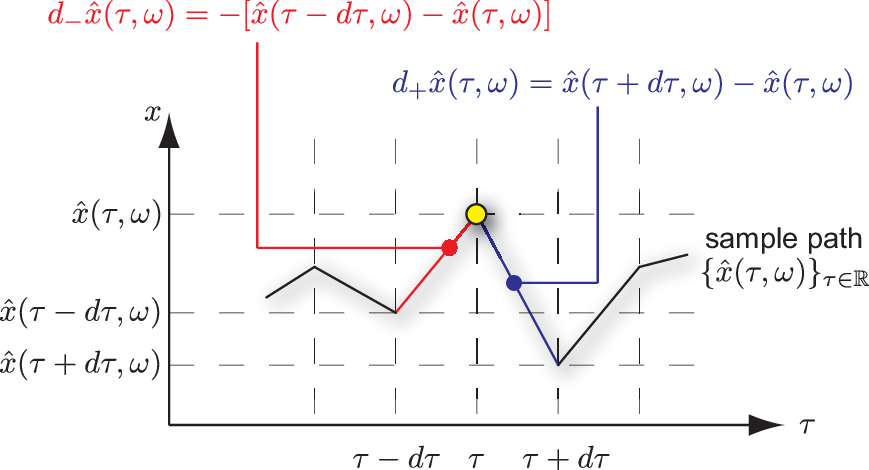}
\par\end{center}
\begin{center}
(A)
\par\end{center}%
\end{minipage}%
\begin{minipage}[t]{0.5\columnwidth}%
\begin{center}
\includegraphics[scale=0.5]{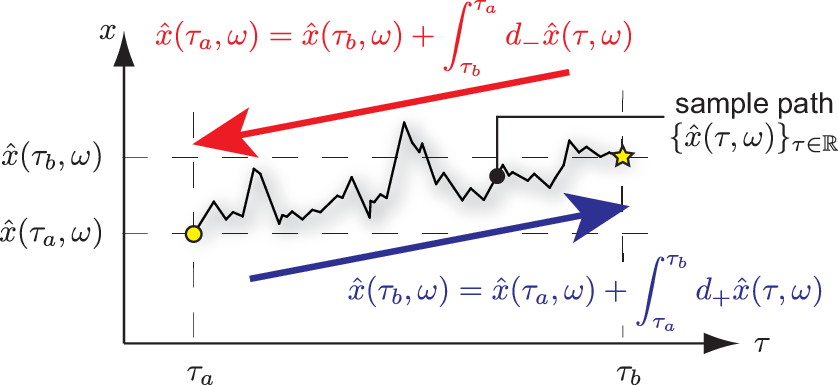}
\par\end{center}
\begin{center}
(B)
\par\end{center}%
\end{minipage}
\par\end{centering}
\caption{\label{Fig_dx} $d_{+}\hat{x}(\tau,\omega)$ and $d_{-}\hat{x}(\tau,\omega)$
of the fractions of a continuous stochastic process $\{\hat{x}(\tau,\omega)\}_{\tau\in\mathbb{R}}$.
Let us draw a sample path $\{\hat{x}(\tau,\omega)\}_{\tau\in\mathbb{R}}$,
then, consider how to generate this by stochastic differential equations.
(A) Due to the non-differentiability of $\{\hat{x}(\tau,\omega)\}_{\tau\in\mathbb{R}}$,
we have to define Eq.(\ref{eq: st-dif}) as the two types of its evolution
at $\hat{x}(\tau,\omega)$; $d_{+}\hat{x}(\tau,\omega)$ for a $\{\mathcal{P}_{\tau}\}$-prog.\ of
a normal diffusion process, and $d_{-}\hat{x}(\tau,\omega)$ for an
$\{\mathcal{F}_{\tau}\}$-prog.\ of an inverse process of a $\{\mathcal{P}_{\tau}\}$-prog.
See also {\bf Fig.\ref{FigWP}} for this. (B) a curve $\{\hat{x}(\tau,\omega)\}_{\tau\in[\tau_{a},\tau_{b}]}$
is imposed by (A). For the initial value $\hat{x}(\tau_{a},\omega)$,
$\hat{x}(\tau_{b},\omega)$ is found by a forward (normal) diffusion
as a $\{\mathcal{P}_{\tau}\}$-prog., i.e., $\hat{x}(\tau_{b},\omega)=\hat{x}(\tau_{a},\omega)+\int_{\tau_{a}}^{\tau_{b}}d_{+}\hat{x}(\tau,\omega)$.
On the other hand, an $\{\mathcal{F}_{\tau}\}$-prog.\ is regarded
as a backward diffusion $\hat{x}(\tau_{a},\omega)=\hat{x}(\tau_{b},\omega)+\int_{\tau_{b}}^{\tau_{a}}d_{-}\hat{x}(\tau,\omega)$
with the terminal value $\hat{x}(\tau_{b},\omega)$. This is the reason
why we have the two expressions of ``$+$'' and ``$-$''. }
\end{figure}
\begin{figure}
\centering{}%
\begin{minipage}[t]{0.5\columnwidth}%
\begin{center}
\includegraphics[scale=0.3]{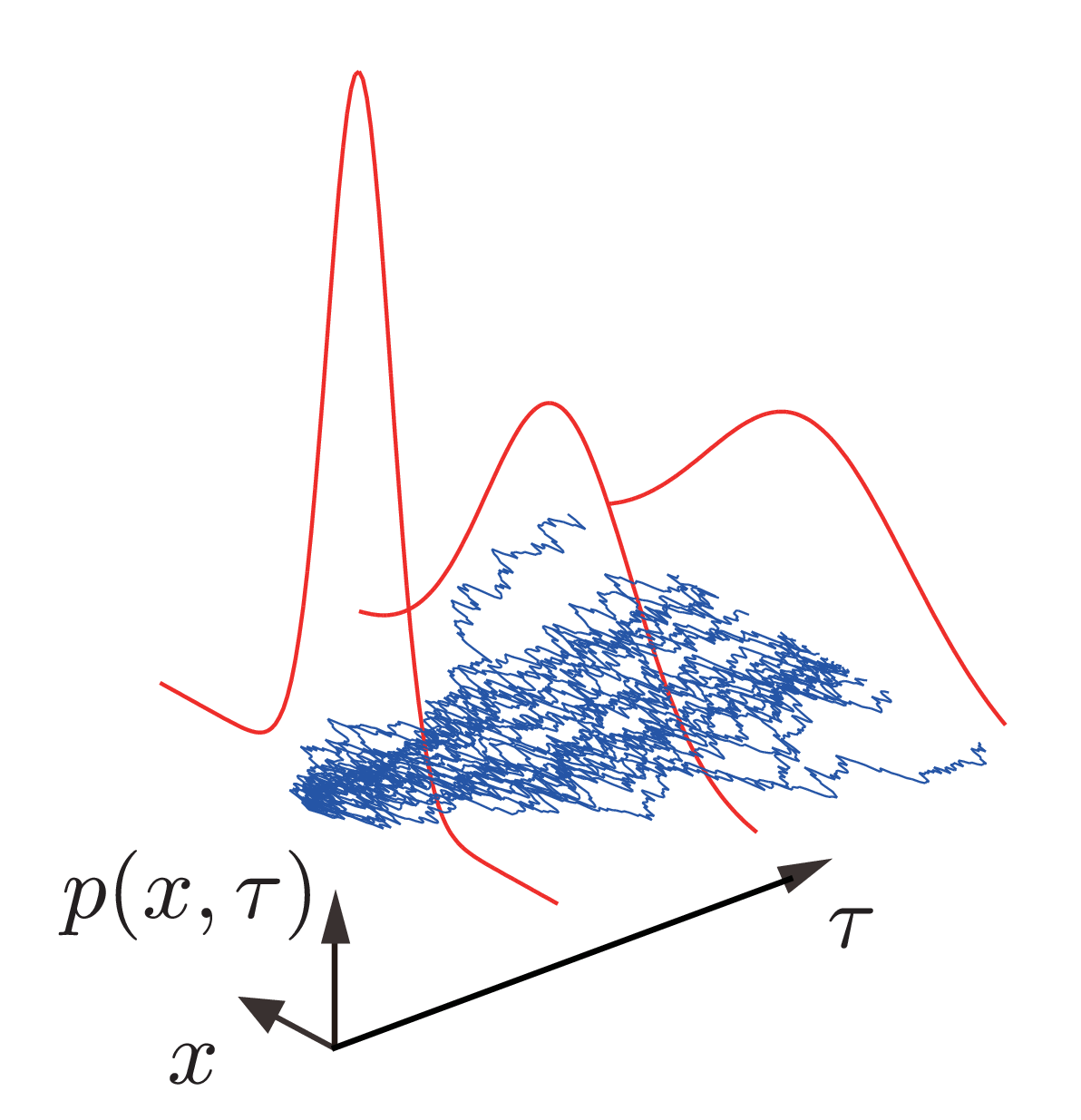}
\par\end{center}
\begin{center}
(a) $\partial_{\tau}p(x,\tau)=1/2\times\partial_{x}^{2}p(x,\tau)$\\
as a $\{\mathcal{P}_{\tau}\}$-WP
\par\end{center}%
\end{minipage}%
\begin{minipage}[t]{0.5\columnwidth}%
\begin{center}
\includegraphics[scale=0.3]{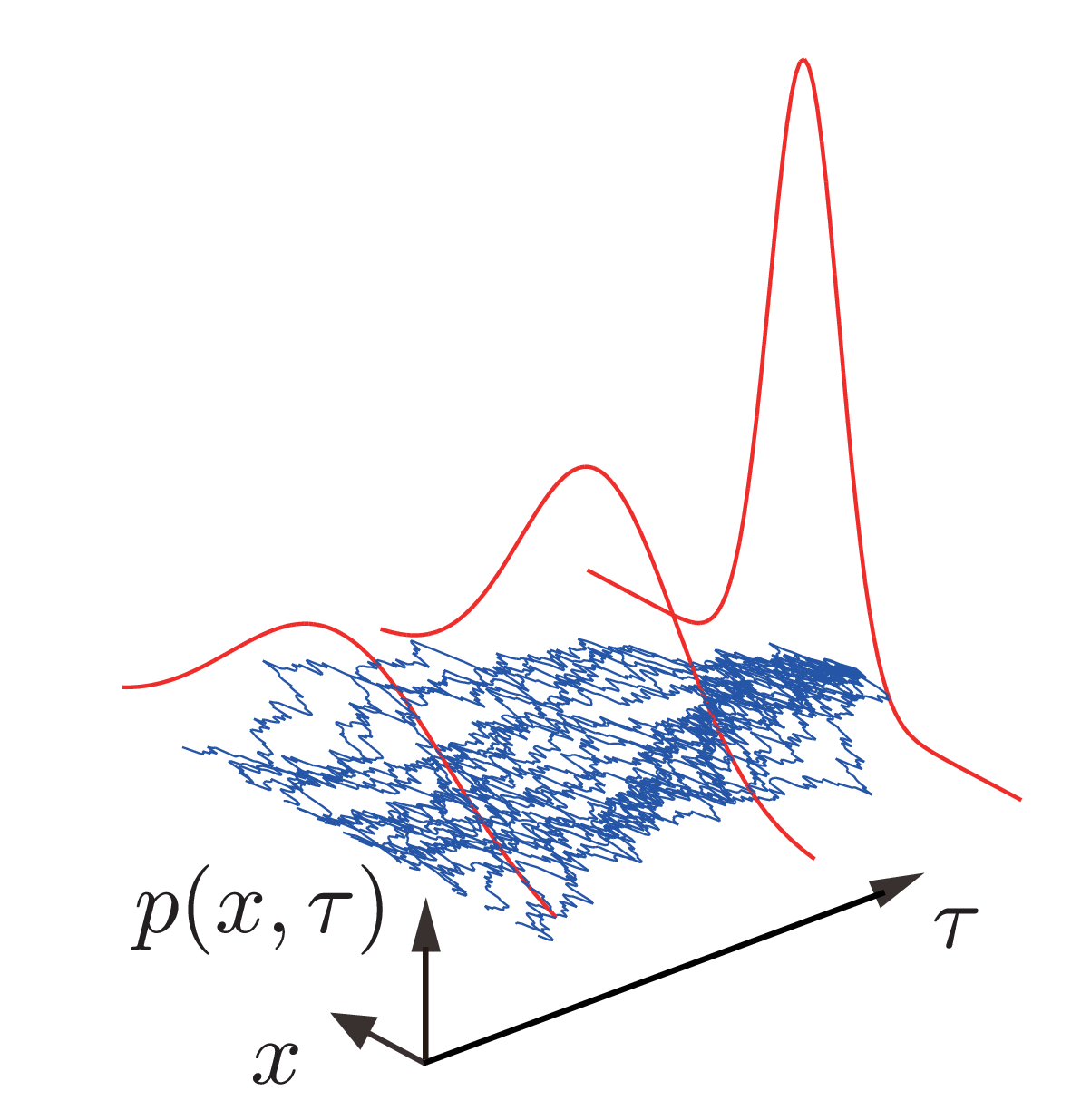}
\par\end{center}
\begin{center}
(b) $\partial_{\tau}p(x,\tau)=-1/2\times\partial_{x}^{2}p(x,\tau)$\\
as an $\{\mathcal{F}_{\tau}\}$-WP
\par\end{center}%
\end{minipage}\caption{\label{FigWP} Trajectories and probabilities of WPs. (a) a $\{\mathcal{P}_{\tau}\}$-WP
and (b) an $\{\mathcal{F}_{\tau}\}$-WP. The blue lines are their
sample paths and their probability densities are drawn by red.}
 
\end{figure}

For $\omega$ the label of sample paths, let $\{\hat{x}(\tau,\omega)\}_{\tau\in\mathbb{R}}$
a trajectory of a quanta with its existence probability of $\mathscr{P}^{1\mathrm{\mathchar`-dim}}(\omega)$
be a 1D Nelson's (S3)-process \cite{Nelson(2001_book)}. 
\begin{align}
\hat{x}(\tau_{b},\omega)-\hat{x}(\tau_{a},\omega) & =\int_{\tau_{a}}^{\tau_{b}}v_{\pm}(\hat{x}(\tau,\omega))d\tau+\lambda\times\int_{\tau_{a}}^{\tau_{b}}dw_{\pm}(\tau,\omega)\label{eq: st-integra}
\end{align}
With $\lambda\coloneqq\sqrt{\hbar/m_{0}}$. Equation (\ref{eq: st-integra})
is also written like 
\begin{align}
\hat{x}(\tau_{b},\omega)-\hat{x}(\tau_{a},\omega) & =\int_{\tau_{a}}^{\tau_{b}}d_{\pm}\hat{x}(\tau,\omega)\label{eq: st-integra2}
\end{align}
by introducing its differential form 
\begin{align}
d_{\pm}\hat{x}(\tau,\omega) & \coloneqq\pm[\hat{x}(\tau\pm d\tau,\omega)-\hat{x}(\tau,\omega)]\nonumber \\
 & =v_{\pm}(\hat{x}(\tau,\omega))d\tau+\lambda\times dw_{\pm}(\tau,\omega).\label{eq: st-dif}
\end{align}
This is a combination of a drift and its randomness governed by $\{w_{\pm}(\tau,\omega)\}_{\tau\in\mathbb{R}}$
of 1D Wiener processes ({\bf WP}; or Brownian motion) such that
\begin{align}
\mathbb{E}\llbracket dw_{\pm}(\tau,\bullet)\rrbracket= & 0,
\end{align}
\begin{align}
\mathbb{E}\llbracket[dw_{\pm}(\tau,\bullet)]^{2}\rrbracket= & d\tau,
\end{align}
\begin{align}
\mathbb{E}\llbracket dw_{+}(\tau,\bullet)\cdot dw_{-}(\tau,\bullet)\rrbracket= & 0,
\end{align}
for $d\tau>0$. Where, $\mathbb{E}\llbracket f(\bullet)\rrbracket\coloneqq\int_{\varOmega}f(\omega)d\mathscr{P}^{1\mathrm{\mathchar`-dim}}(\omega)$
the expectation of $\{f(\omega)\}_{\omega\in\varOmega}$. The two
types of ``$\pm$'' are derived from the randomness of $\{\hat{x}(\tau,\omega)\}_{\tau\in\mathbb{R}}$,
for the time reversibility\footnote{$\{w_{+}(\tau,\omega)\}_{\tau\in\mathbb{R}}$ of a $\{\mathcal{P}_{\tau}\}$-WP
does not have its time reversibility in general. Therefore, we make
a $\{\mathcal{P}_{(-\tau)}\}$-WP and name it an $\{\mathcal{F}_{\tau}\}$-WP,
i.e., $\{w_{-}(\tau,\omega)\}_{\tau\in\mathbb{R}}$.}; let {\bf ``$\boldsymbol{\{\mathcal{P}_{\tau}\}}$-progressive (prog.)''}
denoted by ``$+$'' be a diffusion from $\tau_{a}$ to $\tau_{b}>\tau_{a}$.
And {\bf ``$\boldsymbol{\{\mathcal{F}_{\tau}\}}$-prog.''} by ``$-$''
is an inverse process of ``$\{\mathcal{P}_{\tau}\}$-prog.''\  shown
in {\bf Fig.\ref{Fig_dx}}. Thus, an (S3)-process of Eqs.(\ref{eq: st-integra}-\ref{eq: st-dif})
is $\{\mathcal{P}_{\tau}\}$-prog.\ and $\{\mathcal{F}_{\tau}\}$-prog.
``Progressive'' means that a stochastic process $\{\hat{x}(\tau,\omega)\}_{(\tau,\omega)\in\mathbb{R}\times\varOmega}$
is integrable for $\tau$ and $\omega$, mathematically. For its probability
density $p(x,\tau)\coloneqq d\mathscr{P}^{1\mathrm{\mathchar`-dim}}/dx$,
the difference between a $\{\mathcal{P}_{\tau}\}$-prog.\ and an
$\{\mathcal{F}_{\tau}\}$-prog.\ appears in the forward ($+$) and
backward ($-$) Fokker-Planck ({\bf FP}) equations w.r.t. Eq.(\ref{eq: st-dif})
\cite{Nelson(1966a),Nelson(2001_book),Nelson(1985)}, 
\begin{equation}
\partial_{\tau}p(x,\tau)+\partial_{x}\left[v_{\pm}(x)p(x,\tau)\right]=\pm\frac{\lambda^{2}}{2}\partial_{x}^{2}p(x,\tau).\label{eq: Fokker-Planck-1dim}
\end{equation}
{\bf Figure \ref{FigWP}} shows the difference between $\{\mathcal{P}_{\tau}\}$
and $\{\mathcal{F}_{\tau}\}$-WPs by their trajectories (the blue
lines) and probabilities (the red lines) when $v_{\pm}=0$ in Eq.(\ref{eq: Fokker-Planck-1dim}):
(a) a $\{\mathcal{P}_{\tau}\}$-WP of a forward diffusion and (b)
an $\{\mathcal{F}_{\tau}\}$-WP of a backward diffusion. Equation
(\ref{eq: Fokker-Planck-1dim}) is derived by the following It\^{o}
formula\footnote{The word of ``a.s.'' means ``almost surely,'' namely, the It\^{o}
formula imposed for all $\omega$. } for $\{\hat{x}(\tau,\omega)\}_{\tau\in\mathbb{R}}$ of a $\{\mathcal{P}_{\tau}\}$-prog.\ and
an $\{\mathcal{F}_{\tau}\}$-prog., respectively \cite{Ito(1944),Gradiner(2009)}:
\begin{align}
d_{\pm}f(\hat{x}(\tau,\omega)) & =f'(\hat{x}(\tau,\omega))\cdot d_{\pm}\hat{x}(\tau,\omega)\pm\frac{\lambda^{2}}{2}f''(\hat{x}(\tau,\omega))d\tau\,\,\,\,\mathrm{a.s.}\label{eq: Ito-1D}
\end{align}

\subsection{Dynamics}

What Nelson performed after the above construction of the kinematics
was the derivation of the {Schr\"{o}dinger} equation $[i\hbar\partial_{t}+e\phi(x,t)]\psi(x,t)=-\hbar^{2}/2m_{0}\times\partial_{x}^{2}\psi(x,t)$
by the following with the field $E(x,t)\coloneqq-\partial_{x}\phi(x,t)$:
\begin{align}
m_{0}\left[\begin{gathered}\partial_{t}v(x,t)+v(x,t)\cdot\partial_{x}v(x,t)-u(x,t)\cdot\partial_{x}u(x,t)-\frac{\hbar}{2m_{0}}\partial_{x}^{2}u(x,t)\end{gathered}
\right] & =-eE(x,t)\label{eq: Nelson-schrodinger-1}
\end{align}
\begin{align}
V(x,t) & \coloneqq-i\frac{\hbar}{m_{0}}\partial_{x}\ln\psi(x,t)
\end{align}
\begin{align}
v(x,t) & =\mathrm{Re}\left\{ V(x,t)\right\} \label{eq: v-1D}
\end{align}
\begin{align}
u(x,t) & =-\mathrm{Im}\left\{ V(x,t)\right\} \label{eq: u-1D}
\end{align}
Where, $v\coloneqq(v_{+}+v_{-})/2$ and $u\coloneqq(v_{+}-v_{-})/2$
\cite{Nelson(1966a),Nelson(2001_book)}. By $V\coloneqq v-iu$, its
compact form \cite{Nottale(2011)} appears:
\begin{align}
m_{0}D_{t}V(x,t) & =-eE(x,t)\label{eq: Nelson-schrodinger-2}
\end{align}
\begin{align}
D_{t} & \coloneqq\partial_{t}+V(x,t)\partial_{x}-i\frac{\lambda}{2}\partial_{x}^{2}
\end{align}
Equation (\ref{eq: Nelson-schrodinger-2}) is similar to $m_{0}dv/dt=-eE$
in classical dynamics. Then, $v_{+}$ and $v_{-}$ are reproduced
by $V$ for Eqs.(\ref{eq: st-integra}-\ref{eq: st-dif}). Since this
is coupled with Eq.(\ref{eq: Fokker-Planck-1dim}), he succeeded to
demonstrate the question why $\psi^{*}(x,t)\psi(x,t)$ is regarded
as $p(x,t)$. 

\section{Stochastic kinematics and dynamics of a scalar electron by a 4D stochastic
process\label{kinematics_dynamics4D}}

We give the set of the stochastic kinematics and dynamics of a scalar
electron in this section, i.e, the quantization of $dx^{\mu}=v^{\mu}d\tau$
and $m_{0}dv^{\mu}=-ev_{\nu}F^{\mu\nu}$. It is recommended to read
from {\bf Sect.\ref{Scheme}} to the readers who want to check its
scheme briefly at first.

\subsection{Kinematics}

The following is a natural idea for a scalar electron in \underline{the 4D spacetime}:
We assume an expansion of Eq.(\ref{eq: st-dif}) to the relativistic
kinematics such that 
\begin{equation}
d_{\pm}\hat{x}^{\mu}(\tau,\omega)=\mathcal{V}_{\pm}^{\mu}(\hat{x}(\tau,\omega))d\tau+\lambda\times dW_{\pm}^{\mu}(\tau,\omega)\label{eq: Kinematics_St}
\end{equation}
for $\lambda\coloneqq\sqrt{\hbar/m_{0}}$ and $\mu=0,1,2,3$ with
its existence probability $\mathscr{P}(\omega)$. It is coupled with
$m_{0}\mathfrak{D_{\tau}}\mathcal{V}^{\mu}=-e\mathcal{\hat{V}}_{\nu}F^{\mu\nu}$
equivalent to the KG equation (see it later). As a mimic of the 1D
case, we want to require the relativistic FP equation 
\begin{equation}
\partial_{\tau}p(x,\tau)+\partial_{\mu}[\mathcal{V}_{\pm}^{\mu}(x)p(x,\tau)]=\mp\frac{\lambda^{2}}{2}\partial_{\mu}\partial^{\mu}p(x,\tau)\label{eq: Fokker-Planck}
\end{equation}
w.r.t.\ $p(x,\tau)\coloneqq d\mathscr{P}/d^{4}x$ and the It\^{o}
formula, 
\begin{align}
d_{\pm}f(\hat{x}(\tau,\omega)) & =\partial_{\mu}f(\hat{x}(\tau,\omega))\cdot d_{\pm}\hat{x}^{\mu}(\tau,\omega)\pm\frac{\lambda^{2}}{2}(-g^{\mu\nu})\partial_{\mu}\partial_{\nu}f(\hat{x}(\tau,\omega))d\tau\,\,\,\,\mathrm{a.s.}\label{eq: Ito formula}
\end{align}
We will demonstrate the validity of the above adapting the KG equation
well in the later discussion.

Anyway for $\{\hat{x}(\tau,\omega)\}_{\tau\in\mathbb{R}}$ of a 4D
$\{\mathcal{P}_{\tau}\}$-prog., it imposes the expanded formula of
Eq.(\ref{eq: Ito-1D}); $d_{+}f(\hat{x}(\tau,\omega))=\partial_{\mu}f(\hat{x}(\tau,\omega))\cdot d_{+}\hat{x}^{\mu}(\tau,\omega)+\lambda^{2}/2\times\delta^{\mu\nu}\partial_{\mu}\partial_{\nu}f(\hat{x}(\tau,\omega))d\tau$
a.s. By comparing it with Eq.(\ref{eq: Ito formula}), $\delta^{\mu\nu}\partial_{\mu}\partial_{\nu}f(\hat{x}(\tau,\omega))$
has to be $(-g^{\mu\nu})\partial_{\mu}\partial_{\nu}f(\hat{x}(\tau,\omega))$.
Therefore, let us introduce a $\{\mathscr{P}_{\tau}\}$-prog.\  for
``$+$'' in Eq.(\ref{eq: Kinematics_St})
\[
\underset{\{\mathscr{P}_{\tau}\}\mathchar`-\mathrm{prog.}}{\underbrace{\hat{x}(\tau,\omega)}}\coloneqq(\underset{\{\mathcal{F}_{\tau}\}\mathchar`-\mathrm{prog.}}{\underbrace{\hat{x}^{0}(\tau,\omega)},}\underset{\{\mathcal{P}_{\tau}\}\mathchar`-\mathrm{prog.}}{\underbrace{\hat{x}^{1}(\tau,\omega),\hat{x}^{2}(\tau,\omega),\hat{x}^{3}(\tau,\omega)}})
\]
and an $\{\mathscr{F}_{\tau}\}$-prog.\ by ``$-$'' 
\[
\underset{\{\mathscr{F}_{\tau}\}\mathchar`-\mathrm{prog.}}{\underbrace{\hat{x}(\tau,\omega)}}\coloneqq(\underset{\{\mathcal{P}_{\tau}\}\mathchar`-\mathrm{prog.}}{\underbrace{\hat{x}^{0}(\tau,\omega)},}\underset{\{\mathcal{F}_{\tau}\}\mathchar`-\mathrm{prog.}}{\underbrace{\hat{x}^{1}(\tau,\omega),\hat{x}^{2}(\tau,\omega),\hat{x}^{3}(\tau,\omega)}}).
\]
Where, let $\{W_{+}(\tau,\omega)\}_{\tau\in\mathbb{R}}$ and $\{W_{-}(\tau,\omega)\}_{\tau\in\mathbb{R}}$
of WPs in Eq.(\ref{eq: Kinematics_St}) be $\{\mathscr{P}_{\tau}\}$
and $\{\mathscr{F}_{\tau}\}$-prog., respectively. Hereby, we name
$\{\hat{x}(\tau,\omega)\}_{\tau\in\mathbb{R}}$ of a $\{\mathscr{P}_{\tau}\}$
and $\{\mathscr{F}_{\tau}\}$-prog.\ {\bf``a D-prog.''} \textcolor{red}{}
For $d\tau\geq0$ and $\mathbb{E}\llbracket f(\bullet)\rrbracket\coloneqq\allowbreak\int_{\varOmega}f(\omega)\allowbreak d\mathscr{P}(\omega)$,
the following rules are satisfied as the expansion of the 1D case:
\begin{align}
\mathbb{E}\llbracket dW_{\pm}^{\mu}(\tau,\bullet)\rrbracket= & 0
\end{align}
\begin{align}
\mathbb{E}\llbracket dW_{\pm}^{\mu}(\tau,\bullet)\cdot dW_{\pm}^{\nu}(\tau,\bullet)\rrbracket=\delta^{\mu\nu}\times & d\tau
\end{align}
\begin{align}
\mathbb{E}\llbracket dW_{\pm}^{\mu}(\tau,\bullet)\cdot dW_{\mp}^{\nu}(\tau,\bullet)\rrbracket= & 0
\end{align}
The definition of $d_{-}f(\hat{x}(\tau,\omega))$ by Eq.(\ref{eq: Ito formula})
can be checked by the expansion for an $\{\hat{x}(\tau,\omega)\}_{\tau\in\mathbb{R}}$
of $\{\mathscr{F}_{\tau}\}$-prog.,
\begin{align}
d_{-}f(\hat{x}(\tau,\omega)) & =f(\underset{\{\mathcal{P}_{\tau}\}\mathchar`-\mathrm{\mathrm{prog.}}}{\underbrace{\hat{x}^{0}(\tau+d\tau,\omega)},}\underset{\{\mathcal{F}_{\tau}\}\mathchar`-\mathrm{\mathrm{prog.}}}{\underbrace{\hat{x}^{i=1,2,3}(\tau,\omega)}})-f(\hat{x}^{0}(\tau,\omega),\hat{x}^{i=1,2,3}(\tau-d\tau,\omega))
\end{align}
with the help by Eq.(\ref{eq: Kinematics_St}) (a 4D version of Eq.(\ref{eq: st-dif})
)
\begin{align}
d_{-}\hat{x}(\tau,\omega) & =(\underset{\{\mathcal{P}_{\tau}\}\mathchar`-\mathrm{\mathrm{prog.}}}{\underbrace{\hat{x}^{0}(\tau+d\tau,\omega)},}\underset{\{\mathcal{F}_{\tau}\}\mathchar`-\mathrm{\mathrm{prog.}}}{\underbrace{\hat{x}^{i=1,2,3}(\tau,\omega)}})-(\hat{x}^{0}(\tau,\omega),\hat{x}^{i=1,2,3}(\tau-d\tau,\omega))
\end{align}
to first order of $d\tau$. $d_{+}f(\hat{x}(\tau,\omega))$ is imposed
by $\{\hat{x}(\tau,\omega)\}_{\tau\in\mathbb{R}}$ of a $\{\mathscr{P}_{\tau}\}$-prog.
Let Eq.(\ref{eq: Ito formula}) be the general definition of $d_{\pm}$.
The probability $\mathscr{P}(\omega)$ is calculated by the FP equation
(\ref{eq: Fokker-Planck}) since $p(x,\tau)\coloneqq d\mathscr{P}/d^{4}x$.
Nelson introduced the drift velocities via the so-called mean derivatives
\cite{Nelson(1966a)}. In the present case, it is evaluated by 
\begin{equation}
\mathcal{V}_{\pm}^{\mu}(\hat{x}(\tau,\omega))\coloneqq\mathbb{E}\left\llbracket \left.\frac{d_{\pm}\hat{x}^{\mu}}{d\tau}(\tau,\bullet)\right|\hat{x}(\tau,\omega)\right\rrbracket (\omega).\label{eq: mean-deriv}
\end{equation}
Where, $\mathbb{E}\llbracket f(\hat{x}(\tau,\bullet))|\mathscr{C}\rrbracket(\omega)=\int f(\hat{x}(\tau,\omega))d\mathscr{P}_{\mathscr{C}}(\omega)$
the conditional expectation for $\mathscr{P}_{\mathscr{C}}(X)$ a
conditional probability of $X$ given $\mathscr{C}$. 

Then, we define the complex differential $\hat{d}\coloneqq(d_{+}+d_{-})/2-i(d_{+}-d_{-})/2$
symbolically such that
\begin{align}
\hat{d}f(\hat{x}(\tau,\omega)) & =\partial_{\mu}f(\hat{x}(\tau,\omega))\cdot\hat{d}\hat{x}^{\mu}(\tau,\omega)-\frac{i\lambda^{2}}{2}\partial_{\mu}\partial^{\mu}f(\hat{x}(\tau,\omega))d\tau\,\,\,\,\mathrm{a.s.}\label{eq: C-Ito formula}
\end{align}
and the complex velocity 
\begin{align}
\mathcal{V}^{\mu}(\hat{x}(\tau,\omega)) & \coloneqq\mathbb{E}\left\llbracket \left.\frac{\hat{d}\hat{x}^{\mu}}{d\tau}(\tau,\bullet)\right|\hat{x}(\tau,\omega)\right\rrbracket (\omega)\\
 & =\frac{\mathcal{V}_{+}^{\mu}(\hat{x}(\tau,\omega))+\mathcal{V}_{-}^{\mu}(\hat{x}(\tau,\omega))}{2}-i\frac{\mathcal{V}_{+}^{\mu}(\hat{x}(\tau,\omega))-\mathcal{V}_{-}^{\mu}(\hat{x}(\tau,\omega))}{2}\nonumber 
\end{align}
with its other assumption \cite{Nottale(2011)}: 
\begin{equation}
\mathcal{V}^{\mu}(x)\coloneqq\frac{1}{m_{0}}\times\left[i\hbar\partial{}^{\mu}\ln\phi(x)+eA{}^{\mu}(x)\right]\label{eq: Comp-V}
\end{equation}

Let $||A||_{(g,\mathbb{C}^{4})}^{2}\coloneqq A_{\mu}^{*}A^{\mu}$
($A^{*}$ is the complex conjugate of a vector $A$) on the Minkowski
spacetime, 
\begin{equation}
d\tau\coloneqq\frac{1}{c}\times\sqrt{\mathbb{E}\left\llbracket ||\hat{d}\hat{x}(\tau,\bullet)-\lambda\times\hat{d}W(\tau,\bullet)||_{(g,\mathbb{C}^{4})}^{2}\right\rrbracket }
\end{equation}
of the proper time links to the Lorentz invariant 
\begin{equation}
\mathbb{E}\left\llbracket \mathcal{V}_{\mu}^{*}(\hat{x}(\tau,\bullet))\mathcal{V}^{\mu}(\hat{x}(\tau,\bullet))\right\rrbracket =c^{2}\,.\label{eq: Inv.}
\end{equation}
For satisfying Eq.(\ref{eq: Inv.}), $\phi$ in Eq.(\ref{eq: Comp-V})
has to be a wave function of the KG equation with $i\hbar\mathfrak{D}_{\alpha}\coloneqq i\hbar\partial_{\alpha}+eA_{\alpha}(x)$,
\begin{equation}
(i\hbar\mathfrak{D}_{\alpha})\cdot(i\hbar\mathfrak{D}^{\alpha})\phi(x)-m_{0}^{2}c^{2}\phi(x)=0\label{eq: KG eq}
\end{equation}
as {\bf Ref.}\cite{Zastawniak(1990)} suggests. Namely by
\begin{align}
\mathcal{V}_{\mu}^{*}(x)\mathcal{V}^{\mu}(x) & =\frac{1}{m_{0}^{2}}\times\frac{\mathrm{Re}\{\phi^{*}(x)(i\hbar\mathfrak{D}_{\mu})\cdot(i\hbar\mathfrak{D}^{\mu})\phi(x)\}}{\phi^{*}(x)\phi(x)}+\frac{\hbar^{2}}{2m_{0}^{2}}\times\frac{\partial_{\mu}\partial^{\mu}[\phi(x)\cdot\phi^{*}(x)]}{\phi^{*}(x)\phi(x)},\label{eq: VV-cal}
\end{align}
the first term in RHS is $c^{2}$ and the second term becomes zero
by its expectation after the substitution $x=\hat{x}(\tau,\omega)$.
Let us demonstrate this in the end of the next small section for the
probability density.

\subsection{Equations of probability density}

By using this complex velocity, the FP equation of Eq.(\ref{eq: Fokker-Planck})
derives the equation of continuity for $p$, w.r.t.\ $x\in\hat{x}(\tau,\varOmega)\coloneqq\{\hat{x}(\tau,\omega)\}_{\omega\in\varOmega}$.
\begin{equation}
\partial_{\tau}p(x,\tau)+\partial_{\mu}\left[\mathrm{Re}\{\mathcal{V}^{\mu}(x)\}p(x,\tau)\right]=0\label{eq: eq of continuity}
\end{equation}
Then for a naturally boundary condition $p(x,\pm\infty)=0$, 
\begin{equation}
\partial_{\mu}\left[\mathrm{Re}\{\mathcal{V}^{\mu}(x)\}\int_{\mathbb{R}}d\tau\,p(x,\tau)\right]=0
\end{equation}
is found. Or by using 
\begin{equation}
p(x,\tau)\coloneqq\mathbb{E}\llbracket\delta^{4}(x-\hat{x}(\tau,\bullet))\rrbracket\label{eq: prob-density}
\end{equation}
due to the definition of $\mathbb{E}\llbracket f(\hat{x}(\tau,\bullet))\rrbracket$,
i.e.,
\begin{align}
\mathbb{E}\llbracket f(\hat{x}(\tau,\bullet))\rrbracket & \coloneqq\int_{\varOmega}f(\hat{x}(\tau,\omega')))d\mathscr{P}(\omega')\\
 & =\int_{\mathbb{R}^{4}}f(x')p(x',\tau)d^{4}x',
\end{align}
the following is imposed:
\begin{equation}
\partial_{\mu}\mathbb{E}\left\llbracket -ec\int_{\mathbb{R}}d\tau\,\mathrm{Re}\{\mathcal{V}^{\mu}(x)\}\delta^{4}(x-\hat{x}(\tau,\bullet))\right\rrbracket =0\label{eq: cons-flow}
\end{equation}
This is the conservation law of the current density 
\begin{equation}
j_{\mathrm{stochastic}}^{\mu}(x)\coloneqq\mathbb{E}\left\llbracket -ec\int_{\mathbb{R}}d\tau\,\mathrm{Re}\left\{ \mathcal{V}^{\mu}(x)\right\} \delta^{4}(x-\hat{x}(\tau,\bullet))\right\rrbracket \label{eq:j-stochastic}
\end{equation}
in the 4D spacetime (see this role in {\bf Sect.\ref{Maxwell}}).
The following relation is also found by Eq.(\ref{eq: Fokker-Planck}).
\begin{align}
\mathrm{Im}\{\mathcal{V}^{\mu}(x)\} & =\begin{cases}
\begin{gathered}\frac{\lambda^{2}}{2}\times\partial^{\mu}\ln p(x,\tau)\end{gathered}
, & \begin{gathered}x\in\hat{x}(\tau,\varOmega)\end{gathered}
\\
\begin{gathered}\frac{\lambda^{2}}{2}\times\partial^{\mu}\ln\int_{\mathbb{R}}d\tau\,p(x,\tau)\end{gathered}
, & x\in\bigcup_{\tau\in\mathbb{R}}\hat{x}(\tau,\varOmega)
\end{cases}\label{eq: osmotic pressure1}
\end{align}
Where, $\bigcup_{\tau\in\mathbb{R}}\hat{x}(\tau,\varOmega)=\mathrm{supp}(\int_{\mathbb{R}}d\tau\,p(\circ,\tau))$.
Equation (\ref{eq: osmotic pressure1}) is a mimic of the osmotic
pressure formula \cite{Nelson(1966a),Nelson(2001_book)}. 

Let us come back to the discussion of Eq.(\ref{eq: Inv.}) and Eq.(\ref{eq: VV-cal})
(see also \cite{Zastawniak(1990)}). By $\phi(x)\coloneqq\exp[R(x)/\hbar+iS(x)/\hbar]$
w.r.t.\ $R$ and $S$ of real-valued functions, $\phi^{*}(x)\phi(x)=\exp[2R(x)/\hbar]$.
Where, $\partial^{\mu}R(x)=\mathrm{Im}\{m_{0}\mathcal{V}^{\mu}(x)\}=\hbar/2\times\partial^{\mu}\ln p(x,\tau)$
by Eq.(\ref{eq: Comp-V}) and Eq.(\ref{eq: osmotic pressure1}) for
$x\in\hat{x}(\tau,\varOmega)$. Thus,
\begin{align}
\frac{\hbar^{2}}{2m_{0}^{2}}\times\frac{\partial_{\mu}\partial^{\mu}[\phi(x)\cdot\phi^{*}(x)]}{\phi^{*}(x)\phi(x)}= & \frac{\hbar^{2}}{2m_{0}^{2}}\times\frac{\partial_{\mu}\partial^{\mu}p(x,\tau)}{p(x,\tau)}.
\end{align}
After assuming $\phi$ as the solution of the KG equation (\ref{eq: KG eq}),
\begin{align}
\mathbb{E}\llbracket\mathcal{V}_{\mu}^{*}(\hat{x}(\tau,\bullet))\mathcal{V}^{\mu}(\hat{x}(\tau,\bullet))\rrbracket & =c^{2}+\frac{\hbar^{2}}{2m_{0}^{2}}\times\mathbb{E}\left\llbracket \frac{\partial_{\mu}\partial^{\mu}p(\hat{x}(\tau,\bullet),\tau)}{p(\hat{x}(\tau,\bullet),\tau)}\right\rrbracket \nonumber \\
 & =c^{2}+\frac{\hbar^{2}}{2m_{0}^{2}}\times\int_{\mathbb{R}^{4}}d^{4}x\,\partial_{\mu}\partial^{\mu}p(x,\tau),
\end{align}
Eq.(\ref{eq: Inv.}) is found by $\partial^{\mu}p(x,\tau)|_{x\in\partial\mathbb{R}^{4}}=0$
the boundary condition on $\partial\mathbb{R}^{4}$ the boundary of
$\mathbb{R}^{4}$. 

\subsection{Dynamics}

Let us give the equation of Nottale's style \cite{Nottale(2011)}
for the dynamics of a Brownian scalar electron: 
\begin{equation}
m_{0}\mathfrak{D_{\tau}}\mathcal{V}^{\mu}(\hat{x}(\tau,\omega))=-e\mathcal{\hat{V}}_{\nu}(\hat{x}(\tau,\omega))F^{\mu\nu}(\hat{x}(\tau,\omega))\label{eq: EOM_St}
\end{equation}
\begin{align}
\hat{\mathcal{V}}^{\mu}(x) & \coloneqq\mathcal{V}^{\mu}(x)+i\frac{\lambda^{2}}{2}\times\partial^{\mu}\label{eq: Operator_CompV}
\end{align}
\begin{align}
\mathfrak{D_{\tau}} & \coloneqq\hat{\mathcal{V}}^{\mu}(x)\cdot\partial_{\mu}\label{eq: Operator_CompDtau}
\end{align}
Equations (\ref{eq: EOM_St}-\ref{eq: Operator_CompDtau}) corresponds
to $m_{0}dv^{\mu}/dt=-ev_{\nu}F^{\mu\nu}$ in classical dynamics and
it implies the KG equation \cite{Nottale(2011)} since
\begin{align}
\mathfrak{D_{\tau}}\mathcal{V}^{\mu}+\frac{e}{m_{0}}\mathcal{\hat{V}}_{\nu}F^{\mu\nu} & =\begin{gathered}\frac{1}{2}\partial^{\mu}\left[\frac{(i\hbar\partial_{\nu}+eA_{\nu})(i\hbar\partial^{\nu}+eA{}^{\nu})\phi}{m_{0}^{2}\phi}\right]\end{gathered}
\nonumber \\
 & =0.
\end{align}
We can easily clarify that Eq.(\ref{eq: EOM_St}) satisfies the $U(1)$-gauge
symmetry. Let us study the relation between Eq.(\ref{eq: Inv.}) and
Eq.(\ref{eq: EOM_St}). This relation corresponds to the one between
$v_{\mu}v^{\mu}=c^{2}$ and $d/d\tau(v_{\mu}v^{\mu})=0$ in classical
dynamics. The readers may find another definition of $\mathfrak{D_{\tau}}$:
\begin{equation}
\mathfrak{D}_{\tau}^{\pm}\coloneqq\mathbb{E}\left\llbracket \left.\frac{d_{\pm}}{d\tau}\right|\hat{x}(\tau,\omega)\right\rrbracket 
\end{equation}
\begin{equation}
\mathfrak{D}_{\tau}\coloneqq\mathbb{E}\left\llbracket \left.\frac{\hat{d}}{d\tau}\right|\hat{x}(\tau,\omega)\right\rrbracket =\frac{1-i}{2}\mathfrak{D}_{\tau}^{+}+\frac{1+i}{2}\mathfrak{D}_{\tau}^{-}
\end{equation}
Nelson introduced the partial integral formula for his (S3)-process.
In our case for $\{\hat{x}(\tau,\omega)\}$ of a D-prog., the differential
form of that formula for $\{\alpha^{\mu}\}_{\mu=0,1,2,3}$ and $\{\beta^{\mu}\}_{\mu=0,1,2,3}$
of complex functions\footnote{For the demonstration of Eq.(\ref{eq: Nelson-partial integral0}),
it is enough to be confirm the following formula with Eq.(\ref{eq: Fokker-Planck}):
\begin{multline*}
\mathbb{E}\left\llbracket \mathfrak{D}_{\tau}^{+}\alpha_{\mu}(\hat{x}(\tau,\bullet))\cdot\beta^{\mu}(\hat{x}(\tau,\bullet))+\alpha_{\mu}(\hat{x}(\tau,\bullet))\cdot\mathfrak{D}_{\tau}^{-}\beta^{\mu}(\hat{x}(\tau,\bullet))\right\rrbracket \\
=\mathbb{E}\left\llbracket \mathfrak{D}_{\tau}^{-}\alpha_{\mu}(\hat{x}(\tau,\bullet))\cdot\beta^{\mu}(\hat{x}(\tau,\bullet))+\alpha_{\mu}(\hat{x}(\tau,\bullet))\cdot\mathfrak{D}_{\tau}^{+}\beta^{\mu}(\hat{x}(\tau,\bullet))\right\rrbracket 
\end{multline*}
} is 
\begin{align}
\frac{d}{d\tau}\mathbb{E}\llbracket\alpha_{\mu}(\hat{x}(\tau,\bullet))\beta^{\mu}(\hat{x}(\tau,\bullet))\rrbracket & =\mathbb{E}\left\llbracket \mathfrak{D}_{\tau}^{\pm}\alpha_{\mu}(\hat{x}(\tau,\bullet))\cdot\beta^{\mu}(\hat{x}(\tau,\bullet))+\alpha_{\mu}(\hat{x}(\tau,\bullet))\cdot\mathfrak{D}_{\tau}^{\mp}\beta^{\mu}(\hat{x}(\tau,\bullet))\right\rrbracket .\label{eq: Nelson-partial integral0}
\end{align}
By linear combining the above ``$+$'' and ``$-$'' formulas, 
\begin{align}
\frac{d}{d\tau}\mathbb{E}\llbracket\alpha_{\mu}(\hat{x}(\tau,\bullet))\beta^{\mu}(\hat{x}(\tau,\bullet))\rrbracket & =\mathbb{E}\left\llbracket \mathfrak{D}_{\tau}\alpha_{\mu}(\hat{x}(\tau,\bullet))\cdot\beta^{\mu}(\hat{x}(\tau,\bullet))+\alpha_{\mu}(\hat{x}(\tau,\bullet))\cdot\mathfrak{D}_{\tau}^{*}\beta^{\mu}(\hat{x}(\tau,\bullet))\right\rrbracket \label{eq: Nelson-partial integral1}\\
 & =\mathbb{E}\left\llbracket \mathfrak{D}_{\tau}^{*}\alpha_{\mu}(\hat{x}(\tau,\bullet))\cdot\beta^{\mu}(\hat{x}(\tau,\bullet))+\alpha_{\mu}(\hat{x}(\tau,\bullet))\cdot\mathfrak{D}_{\tau}\beta^{\mu}(\hat{x}(\tau,\bullet))\right\rrbracket .\label{eq: Nelson-partial integral2}
\end{align}
When $\alpha_{\mu}=\mathcal{V}_{\mu}^{*}$ and $\beta^{\mu}=\mathcal{V}^{\mu}$
in Eq.(\ref{eq: Nelson-partial integral2}), 
\begin{align}
\frac{d}{d\tau}\mathbb{E}\left\llbracket \mathcal{V}_{\mu}^{*}(\hat{x}(\tau,\bullet))\mathcal{V}^{\mu}(\hat{x}(\tau,\bullet))\right\rrbracket  & =\mathbb{E}\left\llbracket \mathcal{V}_{\mu}^{*}(\hat{x}(\tau,\bullet))\cdot\mathfrak{D_{\tau}}\mathcal{V}^{\mu}(\hat{x}(\tau,\bullet))+\mathfrak{D_{\tau}}^{*}\mathcal{V}_{\mu}^{*}(\hat{x}(\tau,\bullet))\cdot\mathcal{V}^{\mu}(\hat{x}(\tau,\bullet))\right\rrbracket \nonumber \\
 & =\frac{\lambda^{4}e}{2m_{0}}\times\int_{\mathbb{R}^{4}}d^{4}x\,p(x,\tau)\cdot\partial_{\mu}\partial_{\nu}F^{\mu\nu}(x)\nonumber \\
 & =0
\end{align}
is found. Thus, Eq.(\ref{eq: Inv.}) is supported by Eq.(\ref{eq: EOM_St})
the dynamics of a Brownian quanta, too.

Consider the expectation of Eq.(\ref{eq: EOM_St}), Ehrenfest's theorem
is imposed naturally for $\langle\hat{x}\rangle_{\tau}\coloneqq\mathbb{E}\llbracket\hat{x}(\tau,\bullet)\rrbracket$
and $\delta\hat{x}(\tau,\omega)\coloneqq\hat{x}(\tau,\omega)-\langle\hat{x}\rangle_{\tau}$:
\begin{align}
m_{0}\frac{d^{2}\langle\hat{x}^{\mu}\rangle_{\tau}}{d\tau^{2}} & =\mathbb{E}\left\llbracket -eF^{\mu\nu}(\hat{x}(\tau,\bullet))\mathrm{Re}\{\mathcal{V}_{\nu}(\hat{x}(\tau,\bullet))\}\right\rrbracket \label{eq: Ehrenfest}\\
 & =-eF^{\mu\nu}(\langle\hat{x}\rangle_{\tau})\frac{d\langle\hat{x}_{\nu}\rangle_{\tau'}}{d\tau}+O(\langle\otimes^{2}\delta\hat{x}\rangle_{\tau})\label{eq: Ehrenfest-expansion}
\end{align}
$d\langle\hat{x}\rangle_{\tau}/d\tau-\mathrm{Re}\{\mathcal{V}(\langle\hat{x}\rangle_{\tau})\}=O(\langle\otimes^{2}\delta\hat{x}\rangle_{\tau})$
is employed since $d\langle\hat{x}\rangle_{\tau}/d\tau=\mathbb{E}\llbracket\mathrm{Re}\{\mathcal{V}(\hat{x}(\tau,\bullet)\}\rrbracket$
by Eqs.(\ref{eq: Nelson-partial integral1},\ref{eq: Nelson-partial integral2})
for $\alpha^{\mu}(x)=x^{\mu}$ and $\beta^{\mu}(x)=1$. Hereby, the
requirement of the issue-$\langle\mathrm{A}\rangle$, the system of
relativistic kinematics and dynamics of a scalar electron is completed.

\subsection{Schematic method of stochastic quantization\label{Scheme}}

For obtaining Eq.(\ref{eq: EOM_St}) the equations of quantum dynamics
of a relativistic quanta from classical dynamics, let us conclude
the above discussion schematically. At first, we regard classical
dynamics as the combination of the kinematics
\begin{align}
dx^{\mu}(\tau) & =v^{\mu}(\tau)d\tau
\end{align}
and its dynamics
\begin{align}
m_{0}\frac{dv^{\mu}}{d\tau} & =-ev_{\nu}(\tau)F^{\mu\nu}(x(\tau)).
\end{align}
By considering that this kinematics is deduced from the expectation
of Eq.(\ref{eq: Kinematics_St}), namely, $\mathbb{E}\llbracket d_{\pm}\hat{x}^{\mu}(\tau,\bullet)\rrbracket=v^{\mu}(\tau)d\tau$,
the above classical dynamics is considered as the one by Ehrenfest's
theorem. Hence, the classical kinematics and dynamics are replaced
by
\begin{align}
\mathbb{E}\llbracket d_{\pm}\hat{x}^{\mu}(\tau,\bullet)\rrbracket & =\frac{d\langle\hat{x}\rangle_{\tau}}{d\tau}d\tau,
\end{align}
\begin{equation}
m_{0}\frac{d^{2}\langle\hat{x}^{\mu}\rangle_{\tau}}{d\tau^{2}}=\mathbb{E}\left\llbracket -e\mathrm{Re}\{\mathcal{V}_{\nu}(\hat{x}(\tau,\bullet))\}F^{\mu\nu}(\hat{x}(\tau,\bullet))\right\rrbracket .
\end{equation}
Since $d^{2}\langle\hat{x}^{\mu}\rangle_{\tau}/d\tau^{2}=\mathbb{E}\llbracket\mathrm{Re}\{\mathfrak{D_{\tau}}\mathcal{V}(\hat{x}(\tau,\bullet)\}\rrbracket$,
Eq.(\ref{eq: EOM_St}) the dynamics of a Brownian scalar electron
is found with its complex conjugate. This is the schematic method
of stochastic quantization.

\section{Maxwell equation\label{Maxwell}}

The Maxwell equation (the issue-$\langle\mathrm{B}\rangle$) is also
given by
\begin{align}
\partial_{\mu}[F^{\mu\nu}(x)+\delta f^{\mu\nu}(x)] & =\mu_{0}\times\mathbb{E}\left\llbracket -ec\int_{\mathbb{R}}d\tau'\mathrm{Re}\left\{ \mathcal{V}^{\nu}(x)\right\} \delta^{4}(x-\hat{x}(\tau',\bullet))\right\rrbracket \label{eq: Maxwell_St}
\end{align}
corresponding to
\begin{align}
\partial_{\mu}[F^{\mu\nu}(x)+\delta f^{\mu\nu}(x)] & =\mu_{0}\times\left[-ec\int_{\mathbb{R}}d\tau'\frac{dx^{\nu}}{d\tau}(\tau')\times\delta^{4}(x-x(\tau'))\right]\label{eq: classical-Maxwell}
\end{align}
in classical physics. Where, we will use $\delta f$ as a singular
field attached to a quanta like the Coulomb field. $j_{\mathrm{stochastic}}$
of the current density given by Eq.(\ref{eq:j-stochastic}) (and in
Eq.(\ref{eq: Maxwell_St})) is equal to one of a KG particle, i.e.,
$j_{\mathrm{stochastic}}=j_{\mathrm{\mathrm{K\mathchar`-G}}}$,
\begin{equation}
j_{\mathrm{\mathrm{K\mathchar`-G}}}^{\mu}(x)=-\frac{iec\lambda^{2}}{2}\left[\phi^{*}(x)\mathfrak{D}^{\mu}\phi(x)-\phi(x)(\mathfrak{D}^{*})^{\mu}\phi^{*}(x)\right]\label{eq:j-KG}
\end{equation}
by the assumption
\begin{equation}
\phi^{*}(x)\phi(x)=\int_{\mathbb{R}}d\tau'p(x,\tau').\label{eq: normalization}
\end{equation}
Since $\partial_{\mu}j_{\mathrm{stochastic}}^{\mu}(x)=0$, $\partial_{\mu}j_{\mathrm{\mathrm{K\mathchar`-G}}}^{\mu}(x)=0$
and 
\begin{align}
j_{\mathrm{stochastic}}^{\mu}(x)= & \frac{\int_{\mathbb{R}}d\tau'p(x,\tau')}{\phi^{*}(x)\phi(x)}\times j_{\mathrm{\mathrm{K\mathchar`-G}}}^{\mu}(x),
\end{align}
therefore, $\partial_{\mu}[\int_{\mathbb{R}}d\tau'p(x,\tau')/\phi^{*}(x)\phi(x)]=0$
has to be satisfied. We emphasize that Eq.(\ref{eq: normalization})
is considered as the normalization of $\phi$ a wave function of the
KG equation corresponding to one of the {Schr\"{o}dinger} equation.
The quantum effect of RR appearing in Eq.(\ref{eq: Rad-formula-sQED})
is derived by the definition of $j_{\mathrm{stochastic}}^{\mu}$ obviously.
When we write it such as
\begin{align}
j_{\mathrm{stochastic}}^{\mu}(x) & =-ec\int_{\mathbb{R}}d\tau'\int_{\omega'\in\varOmega}d\mathscr{P}(\omega')\mathrm{Re}\left\{ \mathcal{V}^{\mu}(x)\right\} \delta^{4}(x-\hat{x}(\tau',\omega')),\label{eq: current-st}
\end{align}
its classical limit ($\hbar\rightarrow0$) converge to a path $\omega_{0}$
a smooth trajectory of $\{x(\tau)\}_{\tau\in\mathbb{R}}\coloneqq\{\hat{x}(\tau,\omega_{0})\}_{\tau\in\mathbb{R}}$.
Then, Eq.(\ref{eq: classical-Maxwell}) is deduced from Eq.(\ref{eq: Maxwell_St}),
namely,
\begin{align}
\lim_{\hbar\rightarrow0}j_{\mathrm{stochastic}}^{\mu}(x) & =-ec\int_{\mathbb{R}}d\tau'\frac{dx^{\mu}}{d\tau}(\tau')\delta^{4}(x-x(\tau')).
\end{align}
See {\bf Sect.\ref{classical limit}} for more detail of the classical
limit. In {\bf Ref.}\cite{Seto(2015)}\footnote{We emphasize Eq.(\ref{eq: current-chi}) the modification of the classical
current $-ec\int_{\mathbb{R}}d\tau'dx/d\tau(\tau')\times\delta^{4}(x-x(\tau'))$
by $-e\mapsto-e\times q(\chi)$ the replacement of the charge imposes
Eq.(\ref{eq: Rad-formula-sQED}) of the radiation formula \cite{Seto(2015)}.
$q(\chi)$ is regarded as the representative value of the distribution
$\mathscr{P}$.}, 
\begin{align}
j_{\mathrm{Ref.}\cite{Seto(2015)}}^{\mu}(x) & =-ec\int_{\mathbb{R}}d\tau'q(\chi(\tau'))\frac{dx^{\mu}}{d\tau}(\tau')\delta^{4}(x-x(\tau'))\label{eq: current-chi}
\end{align}
is introduced for the QED correction. We can assume that $q(\chi)$
is associated by the probability measure of $\mathscr{P}$ by the
comparison between Eq.(\ref{eq: current-st}) and Eq.(\ref{eq: current-chi}).

\section{Action integral\label{Action}}

Let us give the systematic way to define Eq.(\ref{eq: EOM_St}) and
Eq.(\ref{eq: Maxwell_St}) by the following action integral (issue-$\langle\mathrm{C}\rangle$):
\begin{align}
\mathfrak{S}_{0}[\hat{x},A] & =\int_{\mathbb{R}}d\tau\mathbb{E}\left\llbracket \frac{m_{0}}{2}\mathcal{V}_{\alpha}^{*}(\hat{x}(\tau,\bullet))\mathcal{V}^{\alpha}(\hat{x}(\tau,\bullet))\right\rrbracket \nonumber \\
 & \quad\quad+\int_{\mathbb{R}}d\tau\mathbb{E}\left\llbracket -eA_{\alpha}(\hat{x}(\tau,\bullet))\mathrm{Re}\{\mathcal{V}^{\alpha}(\hat{x}(\tau,\bullet))\}\right\rrbracket \nonumber \\
 & \quad\quad\quad+\int_{\mathbb{R}^{4}}d^{4}x\frac{1}{4\mu_{0}c}[F(x)+\delta f(x)]^{2}\label{eq: action integral-0}
\end{align}
This is an analogy of the classical action integral:
\begin{align}
S_{\mathrm{classical}}[x,A] & =\int_{\mathbb{R}}d\tau\frac{m_{0}}{2}v_{\alpha}(\tau)v^{\alpha}(\tau)\nonumber \\
 & \quad\quad+\int_{\mathbb{R}}d\tau[-eA_{\alpha}(x(\tau))v^{\alpha}(\tau)]\nonumber \\
 & \quad\quad\quad+\int_{\mathbb{R}^{4}}d^{4}x\frac{1}{4\mu_{0}c}[F(x)+\delta f(x)]^{2}
\end{align}
However to implement the sub-equations (\ref{eq: Comp-V},\ref{eq: Inv.})
for a scalar electron,
\begin{align}
\mathfrak{S}[\hat{x},A,\lambda] & =\mathfrak{S}_{0}[\hat{x},A]+\int_{\mathbb{R}}d\tau\frac{d}{d\tau}\mathbb{E}\left\llbracket W(\hat{x}(\tau,\bullet))\right\rrbracket \nonumber \\
 & \quad\quad+\int_{\mathbb{R}}d\tau\lambda(\tau)\mathbb{E}\left\llbracket \mathcal{V}_{\alpha}^{*}(\hat{x}(\tau,\bullet))\mathcal{V}^{\alpha}(\hat{x}(\tau,\bullet))-c^{2}\right\rrbracket \label{eq: action integral}
\end{align}
is hereby proposed in the present model. Where, $d\mathbb{E}\left\llbracket W(\hat{x}(\tau,\bullet))\right\rrbracket /d\tau$
with $W(x)=\mathrm{Re}\{i\hbar\ln\phi(x)\}$ is an uncertainty of
its Lagrangian, and the final term in RHS represents its holonomic
constraint. The following Euler-Lagrange-Yasue equation is derived
by the variation of Eq.(\ref{eq: action integral}) w.r.t.\ $\hat{x}$
(see also {\bf Ref.}\cite{Yasue}):
\begin{equation}
\frac{\partial L_{{\rm particle}}}{\partial\hat{x}^{\mu}}-\mathfrak{D}_{\tau}^{*}\frac{\partial L_{{\rm particle}}}{\partial\mathcal{V}^{\mu}}-\mathfrak{D}_{\tau}\frac{\partial L_{{\rm particle}}}{\partial\mathcal{V}^{*\mu}}=0\label{eq: Euler-Lagrange}
\end{equation}
\begin{align}
\frac{\partial L_{{\rm particle}}}{\partial\mathcal{V}^{\mu}}+\frac{\partial L_{{\rm particle}}}{\partial\mathcal{V}^{*\mu}} & =0\label{eq: for sub-eq}
\end{align}
\begin{align}
L_{{\rm particle}}(\hat{x},\mathcal{V},\mathcal{V^{\mathrm{*}}}) & \coloneqq\frac{m_{0}}{2}\mathcal{V}_{\alpha}^{*}\mathcal{V}^{\alpha}-eA_{\alpha}(\hat{x})\mathrm{Re}\{\mathcal{V}^{\alpha}\}\nonumber \\
 & \quad\quad+\mathrm{Re}\{\mathcal{V}^{\alpha}\}\cdot\partial_{\alpha}W(\hat{x})+\lambda(\tau)\times(\mathcal{V}_{\alpha}^{*}\mathcal{V}^{\alpha}-c^{2})
\end{align}
In addition, let us also consider the variation of Eq.(\ref{eq: action integral})
w.r.t\ $\lambda$ for its holonomic constraint. Then, the following
equations are the results w.r.t.\ a scalar electron:
\begin{align*}
[m_{0}+2\lambda(\tau)]\times\mathrm{Re}\{\mathfrak{D_{\tau}}\mathcal{V}^{\mu}(\hat{x}(\tau,\omega))\} & =\mathrm{Re}\{-e\mathcal{\hat{V}}_{\nu}(\hat{x}(\tau,\omega))F^{\mu\nu}(\hat{x}(\tau,\omega))\}
\end{align*}
\begin{align*}
[m_{0}+2\lambda(\tau)]\times\mathrm{Re}\{\mathcal{V}^{\mu}(\hat{x}(\tau,\omega))\} & =\mathrm{Re}\{i\hbar\partial^{\mu}\ln\phi(\hat{x}(\tau,\bullet))\}+eA^{\mu}(\hat{x}(\tau,\bullet))
\end{align*}
\begin{align}
\mathbb{E}\left\llbracket \mathcal{V}_{\alpha}^{*}(\hat{x}(\tau,\bullet))\mathcal{V}^{\alpha}(\hat{x}(\tau,\bullet))\right\rrbracket -c^{2} & =0
\end{align}
For satisfying the above three, $\lambda=0$ is required. Thus, we
can find not only Eq.(\ref{eq: EOM_St}), but Eq.(\ref{eq: Comp-V})
the definition of $\mathcal{V}^{\mu}$ and Eq.(\ref{eq: Inv.}) from
the action integral.

The Maxwell equation (\ref{eq: Maxwell_St}) is derived by the variation
of Eq.(\ref{eq: action integral}) for $A$; $\partial_{\mu}[\partial\mathfrak{L}_{\mathrm{field}}/\partial(\partial_{\mu}A_{\nu})]-\partial\mathfrak{L}_{\mathrm{field}}/\partial A_{\nu}=0$
with the Lagrangian density
\begin{align}
\mathfrak{L}_{\mathrm{field}} & =\frac{1}{4\mu_{0}c}[F(x)+\delta f(x)]^{2}+\mathbb{E}\left\llbracket -\int_{\mathbb{R}}d\tau\,eA_{\alpha}(x)\mathrm{Re}\left\{ \mathcal{V}^{\alpha}(x)\right\} \delta^{4}(x-\hat{x}(\tau,\bullet))\right\rrbracket .
\end{align}
Thus, a radiating Brownian quanta is illustrated by Eq.(\ref{eq: action integral})
with Eq.(\ref{eq: Kinematics_St}) the kinematics of a scalar electron
for defining its probability $\mathscr{P}$ for the action integral
$\mathbb{E}\llbracket\int_{\mathbb{R}}d\tau L_{{\rm particle}}(\hat{x},\mathcal{V},\mathcal{V^{\mathrm{*}}})\rrbracket$.

\subsection{Remark: derivation of Eq.(\ref{eq: Euler-Lagrange}) and Eq.(\ref{eq: for sub-eq}):}

For $\mathfrak{D}_{\tau}^{\pm}\hat{x}^{\mu}(\tau,\omega)=\mathcal{V}_{\pm}^{\mu}(\hat{x}(\tau,\omega))$,
$L_{{\rm particle}}(\hat{x},\mathcal{V},\mathcal{V^{\mathrm{*}}})=L_{0}(\hat{x},\mathcal{V}_{+},\mathcal{V}_{-})$,
and the following relations such that
\begin{equation}
\frac{\partial L_{0}}{\partial\hat{x}^{\mu}}=\frac{\partial L_{{\rm particle}}}{\partial\hat{x}^{\mu}},
\end{equation}
\begin{equation}
\frac{\partial L_{0}}{\partial\mathcal{V}_{+}^{\mu}}=\frac{1-i}{2}\frac{\partial L_{{\rm particle}}}{\partial\mathcal{V}^{\mu}}+\frac{1+i}{2}\frac{\partial L_{{\rm particle}}}{\partial\mathcal{V}^{*\mu}},
\end{equation}
\begin{equation}
\frac{\partial L_{0}}{\partial\mathcal{V}_{-}^{\mu}}=\frac{1+i}{2}\frac{\partial L_{{\rm particle}}}{\partial\mathcal{V}^{\mu}}+\frac{1-i}{2}\frac{\partial L_{{\rm particle}}}{\partial\mathcal{V}^{*\mu}},
\end{equation}
Eqs.(\ref{eq: Euler-Lagrange}-\ref{eq: for sub-eq}) are derived
via 
\begin{align}
\delta\int_{\mathbb{R}}d\tau\mathbb{E}\left\llbracket L_{0}(\hat{x},\mathcal{V}_{+},\mathcal{V}_{-})\right\rrbracket  & =\int_{\mathbb{R}}d\tau\mathbb{E}\left\llbracket \frac{\partial L_{0}}{\partial\hat{x}^{\mu}}\delta\hat{x}^{\mu}+\frac{\partial L_{0}}{\partial\mathcal{V}_{+}^{\mu}}\delta\mathcal{V}_{+}^{\mu}+\frac{\partial L_{0}}{\partial\mathcal{V}_{-}^{\mu}}\delta\mathcal{V}_{-}^{\mu}\right\rrbracket \nonumber \\
 & =\int_{\mathbb{R}}d\tau\mathbb{E}\left\llbracket \left(\delta\hat{x}^{\mu}\cdot\frac{\partial}{\partial\hat{x}^{\mu}}+\mathfrak{D}_{\tau}\delta\hat{x}^{\mu}\cdot\frac{\partial}{\partial\mathcal{V}^{\mu}}+\mathfrak{D}_{\tau}^{*}\delta\hat{x}^{\mu}\cdot\frac{\partial}{\partial\mathcal{V}^{*\mu}}\right)L_{{\rm particle}}\right\rrbracket 
\end{align}
with Eq.(\ref{eq: Nelson-partial integral2}) of Nelson's partial
integral formula, namely,
\begin{align}
\delta\int_{\mathbb{R}}d\tau\mathbb{E}\left\llbracket L_{{\rm particle}}(\hat{x},\mathcal{V},\mathcal{V^{\mathrm{*}}})\right\rrbracket  & =\int_{\mathbb{R}}d\tau\mathbb{E}\left\llbracket \delta\hat{x}^{\mu}\left(\frac{\partial}{\partial\hat{x}^{\mu}}-\mathfrak{D}_{\tau}^{*}\frac{\partial}{\partial\mathcal{V}^{\mu}}-\mathfrak{D}_{\tau}\frac{\partial}{\partial\mathcal{V}^{*\mu}}\right)L_{{\rm particle}}\right\rrbracket \nonumber \\
 & \quad\quad+\int_{\mathbb{R}}d\tau\frac{d}{d\tau}\mathbb{E}\left\llbracket \delta\hat{x}^{\mu}\left(\frac{\partial}{\partial\mathcal{V}^{\mu}}+\frac{\partial}{\partial\mathcal{V}^{*\mu}}\right)L_{{\rm particle}}\right\rrbracket .
\end{align}
With respect to any $\{\delta\hat{x}(\tau,\omega)\}_{(\tau,\omega)\in\mathbb{R}\times\varOmega}$,
Eq.(\ref{eq: Euler-Lagrange}) and Eq.(\ref{eq: for sub-eq}) are
fulfilled.

\section{Radiation reaction\label{RR}}

Since the Maxwell equation is given by Eq.(\ref{eq: Maxwell_St}),
let us describe RR on a Brownian scalar electron as the quantization
of the LAD equation (\ref{eq: LAD eq}-\ref{eq: LAD field}) (the
issue-$\langle\mathrm{D}\rangle$). We also discuss its Ehrenfest's
theorem, its classical limit, Landau-Lifshitz's approximation, and
the radiation formula corresponding to the well-known classical model
in this section. 

\subsection{Quantization of the LAD equation}

By recalling Eqs.(\ref{eq: LAD eq}-\ref{eq: LAD field}) in the classical
regime, $F_{\mathrm{LAD}}$ by Eq.(\ref{eq: LAD field}) is a homogeneous
solution of $\partial_{\mu}F^{\mu\nu}=-ec\mu_{0}\int_{\mathbb{R}}d\tau\,v^{\nu}(\tau)\times\delta^{4}(x-x(\tau))$.
The readers can find the derivation of $F_{\mathrm{LAD}}$ in {\bf Ref.}\cite{A.Sokolov(1986),Dirac(1938),Hartemann(2002)}.
We explore $\mathfrak{F}$ the RR field given by Eq.(\ref{eq: Maxwell_St})
corresponding to $F_{\mathrm{LAD}}$. Let us solve Eq.(\ref{eq: Maxwell_St})
as the mimic of the LAD model at $x=\hat{x}(\tau,\omega)$ under the
Lorenz gauge. Where, $\varOmega$ is a set of all  sample paths in
our physics. Instead of Eq.(\ref{eq: Maxwell_St}), consider
\begin{align}
\partial_{\mu}\mathcal{F}_{(\mathring{\pm})}^{\mu\nu}(x)= & \mu_{0}j_{\mathrm{stochastic}}^{\nu}(x)
\end{align}
with $\partial_{\mu}\mathcal{A}_{(\mathring{\pm})}^{\mu}(x)=0$. Where,
$\mathcal{F}_{(\mathring{\pm})}^{\mu\nu}\coloneqq\partial^{\mu}\mathcal{A}_{(\mathring{\pm})}^{\nu}-\partial^{\nu}\mathcal{A}_{(\mathring{\pm})}^{\mu}$
are the retarded ($\mathring{+}$) / advanced ($\mathring{-}$) fields,
namely,
\begin{align}
\partial_{\mu}\partial^{\mu}\mathcal{A}_{(\mathring{\pm})}^{\nu}(x)= & \mu_{0}j_{\mathrm{stochastic}}^{\nu}(x).
\end{align}
 Thus, their potentials are given by
\begin{align}
\mathcal{A}_{(\mathring{\pm})}^{\nu}(x) & =-ec\mu_{0}\int_{\mathbb{R}}d\tau'\int_{\varOmega}d\mathscr{P}(\omega')\mathrm{Re}\left\{ \mathcal{V}^{\nu}(\hat{x}(\tau',\omega'))\right\} G_{\mathrm{(\mathring{\pm})}}(x,\hat{x}(\tau',\omega'))\label{eq: A_full}
\end{align}
or by $\int_{\varOmega}f(\hat{x}(\tau',\omega'))d\mathscr{P}(\omega')=\int_{x'}f(x')p(x',\tau')dx'^{4}$,
\begin{align}
\mathcal{A}_{(\mathring{\pm})}^{\nu}(x) & =-ec\mu_{0}\int_{\mathbb{R}}d\tau'\int_{x'}dx'^{4}\mathrm{Re}\left\{ \mathcal{V}^{\nu}(x')\right\} G_{\mathrm{(\mathring{\pm})}}(x,x')p(x',\tau').\label{eq: A_full-2}
\end{align}
Where, $G_{\mathrm{(\mathring{\pm})}}$ are the retarded/advanced
Green functions defined by $\partial_{\alpha}\partial^{\alpha}G_{\mathrm{(\mathring{\pm})}}(x,x')=\delta^{4}(x-x')$.
Therefore, 
\begin{align}
\mathcal{F}_{(\mathring{\pm})}^{\mu\nu}(x) & =-ec\mu_{0}\int_{\mathbb{R}}d\tau'\int_{\varOmega}d\mathscr{P}(\omega')\nonumber \\
 & \quad\quad\left[\mathrm{Re}\{\mathcal{V}^{\nu}(\hat{x}(\tau',\omega'))\}\cdot\partial^{\mu}-(\mu\leftrightarrow\nu)\right]G_{(\mathring{\pm})}(x,\hat{x}(\tau',\omega')).\label{eq: F-full}
\end{align}
For the calculation of $\mathcal{F}_{(\mathring{\pm})}$, consider
\[
V_{(\tau,\omega)}\coloneqq\mathrm{supp}(\int_{\mathbb{R}}p(\circ,\tau')d\tau')\cap\{\left.x'\right|||x'-\hat{x}(\tau,\omega)||^{2}=0\}
\]
a neighborhood of a quanta on its light cone. For each $\omega'$,
there is the largest $T^{(\omega')}\coloneqq[\tau-\tau_{1}^{(\omega')},\tau+\tau_{2}^{(\omega')}]$
such that $\hat{x}(\tau-\tau_{1}^{(\omega')},\omega')$ and $\hat{x}(\tau+\tau_{2}^{(\omega')},\omega')$
stay in $V_{(\tau,\omega)}$. Thus, $\varOmega_{(\tau,\omega)}\coloneqq\{\omega'|V_{(\tau,\omega)}\cap\{\hat{x}(\tau',\omega')\}_{\tau'\in T^{(\omega')}}\ne\emptyset\}$
is the largest set of feasible paths for Eq.(\ref{eq: F-full}) when
$x=\hat{x}(\tau,\omega)$, 
\begin{align}
\mathcal{F}_{(\mathring{\pm})}^{\mu\nu}(\hat{x}(\tau,\omega)) & =-ec\mu_{0}\int_{\varOmega_{(\tau,\omega)}}d\mathscr{P}(\omega')\int_{T^{(\omega')}}d\tau'\nonumber \\
 & \quad\quad\left.\left[\mathrm{Re}\{\mathcal{V}^{\nu}(\hat{x}(\tau',\omega'))\}\cdot\partial^{\mu}-(\mu\leftrightarrow\nu)\right]G_{(\mathring{\pm})}(x,\hat{x}(\tau',\omega'))\right|_{x=\hat{x}(\tau,\omega)}.\label{eq: F_X_restrict}
\end{align}
Where, Eq.(\ref{eq: F_X_restrict}) restricts the domain of the integral
for $\tau'$ on $T^{(\omega')}$. Let us assume the each duration
in $\{T^{(\omega')}\}_{\omega'\in\varOmega_{(\tau,\omega)}}$ is finite
and enough short since $G_{\mathrm{(\mathring{\pm})}}$ decays as
$(\mathrm{\mathrm{distance}\,of\,two\,points})^{-1}$. Thus, the stochastic-Taylor
expansion of a function $\mathrm{Re}\{f(\hat{x}(\tau',\omega'))\}$
at $\tau$ is employed: 
\begin{align}
\mathrm{Re}\left\{ f(\hat{x}(\tau',\omega'))\right\}  & =\sum_{m=0}^{\infty}\frac{(\tau'-\tau)^{m}}{m!}\mathrm{Re}\{\mathfrak{D}_{\tau}^{m}f(\hat{x}(\tau,\omega'))\}+R(f)\label{eq: Taylor}
\end{align}
Where, $\mathfrak{D}_{\tau}=[\mathcal{V}^{\mu}(x)+i\lambda^{2}/2\times\partial^{\mu}]\cdot\partial_{\mu}$,
$\mathfrak{D}_{\tau}^{0}=\mathrm{identity}$, and $R(f)$ is its reminder,
\begin{align}
R(f) & =\lambda\times\sum_{m=0}^{n}\int_{\tau}^{\tau'}d\tau_{1}\int_{\tau}^{\tau_{1}}d\tau_{2}\cdots\int_{\tau}^{\tau_{m-2}}d\tau_{m-1}\nonumber \\
 & \quad\quad\mathrm{Re}\left\{ \int_{\tau}^{\tau_{m-1}}\hat{d}W^{\alpha}(\tau_{m},\omega')\cdot\partial_{\alpha}\mathfrak{D}_{\tau}^{m}f(\hat{x}(\tau_{m},\omega'))\right\} .
\end{align}
This is produced by the iteration of the formula $f(\hat{x}(\tau',\omega'))=f(\hat{x}(\tau,\omega'))+\int_{\tau}^{\tau'}\hat{d}f(\hat{x}(\sigma,\omega'))$,
the It\^{o} integral of Eq.(\ref{eq: C-Ito formula})\footnote{Since  $f(\hat{x}(\tau_{b},\omega))-f(\hat{x}(\tau_{a},\omega))=\int_{\tau_{a}}^{\tau_{b}}d_{\pm}f(\hat{x}(\tau,\omega))$,
\begin{align*}
f(\hat{x}(\tau_{b},\omega))-f(\hat{x}(\tau_{a},\omega)) & =\frac{1}{2}\left[\int_{\tau_{a}}^{\tau_{b}}d_{+}f(\hat{x}(\tau,\omega))+\int_{\tau_{a}}^{\tau_{b}}d_{-}f(\hat{x}(\tau,\omega))\right]\\
 & \quad\quad-\frac{i}{2}\left[\int_{\tau_{a}}^{\tau_{b}}d_{+}f(\hat{x}(\tau,\omega))-\int_{\tau_{a}}^{\tau_{b}}d_{-}f(\hat{x}(\tau,\omega))\right].
\end{align*}
}, i.e.,
\begin{align}
f(\hat{x}(\tau',\omega')) & =f(\hat{x}(\tau,\omega'))+\int_{\tau}^{\tau'}\mathfrak{D}_{\tau}f(\hat{x}(\tau'',\omega'))d\tau''\nonumber \\
 & \quad\quad+\lambda\times\int_{\tau}^{\tau'}\partial_{\mu}f(\hat{x}(\tau'',\omega'))\cdot\hat{d}W^{\mu}(\tau'',\omega').
\end{align}
$\mathrm{Re}\{\mathfrak{D}_{\tau}^{m}\hat{x}(\tau,\omega')\}$ corresponds
to $d^{m}x(\tau)/d\tau^{m}$ in Eqs.(\ref{eq: LAD eq}-\ref{eq: LAD field})
for $m=0,1,2,\cdots$. Let us introduce
\begin{align}
\varDelta\hat{x}(\tau,\tau',\omega') & \coloneqq-\sum_{m=1}^{\infty}\frac{(\tau'-\tau)^{m}}{m!}\mathrm{Re}\left\{ \mathfrak{D}_{\tau}^{m}\hat{x}(\tau,\omega')\right\} 
\end{align}
such that $\varDelta\hat{x}(\tau,\tau',\omega')-[\hat{x}(\tau,\omega')-\hat{x}(\tau',\omega')]=R(x)$.
Then, we express the Green function by
\begin{align}
\left.G_{(\mathring{\pm})}(x,\hat{x}(\tau',\omega'))\right|_{x=\hat{x}(\tau,\omega')\,\mathrm{for}\,\omega'\in\varOmega_{(\tau,\omega)}} & =\frac{\theta(\mathring{\pm}\varDelta\hat{x}^{0}(\tau,\tau',\omega'))\times\delta(\tau-\tau')}{4\pi\left|\begin{gathered}\varDelta\hat{x}_{\alpha}(\tau,\tau',\omega')\cdot\frac{d\varDelta\hat{x}^{\alpha}(\tau,\tau',\omega')}{d\tau'}\end{gathered}
\right|}+O(R(x))
\end{align}
and its derivative by
\begin{align}
\partial^{\mu}\left.G_{(\mathring{\pm})}(x,\hat{x}(\tau',\omega'))\right|_{x=\hat{x}(\tau,\omega')\,\mathrm{for}\,\omega'\in\varOmega_{(\tau,\omega)}} & =-\frac{\varDelta\hat{x}^{\mu}(\tau,\tau',\omega')}{\begin{gathered}\varDelta\hat{x}_{\alpha}(\tau,\tau',\omega')\cdot\frac{d\varDelta\hat{x}^{\alpha}(\tau,\tau',\omega')}{d\tau'}\end{gathered}
}\nonumber \\
 & \quad\quad\times\frac{d}{d\tau'}G_{(\mathring{\pm})}(\hat{x}(\tau,\omega'),\hat{x}(\tau',\omega'))+O(R(x)).
\end{align}
The calculation of $\mathcal{F}_{(\mathring{\pm})}(\hat{x}(\tau,\omega))$
requires us to formulate $\partial^{\mu}G_{(\mathring{\pm})}\left.(x,\hat{x}(\tau',\omega'))\right|_{x=\hat{x}(\tau,\omega)}$.
We evaluate it by
\begin{align}
\partial^{\mu}\left.G_{(\mathring{\pm})}(x,\hat{x}(\tau',\omega'))\right|_{x=\hat{x}(\tau,\omega)} & =\partial^{\mu}G_{(\mathring{\pm})}\left.(x,\hat{x}(\tau',\omega'))\right|_{x=\hat{x}(\tau,\omega')\,\mathrm{for}\,\omega'\in\varOmega_{(\tau,\omega)}}\nonumber \\
 & \quad\quad+O(\hat{x}(\tau,\omega)-\hat{x}(\tau,\omega'))
\end{align}
with
\begin{align}
\varDelta\hat{x}(\tau,\tau',\omega')-[\hat{x}(\tau,\omega)-\hat{x}(\tau',\omega')] & =O(R(x),\hat{x}(\tau,\omega)-\hat{x}(\tau,\omega')).
\end{align}
$O(R(x),\hat{x}(\tau,\omega)-\hat{x}(\tau,\omega'))$ means the declaration
that we only investigate the paths in $\varOmega_{(\tau,\omega)}\cap\{\omega'|\hat{x}(\tau,\omega')=\hat{x}(\tau,\omega)\}$
at $\tau$. The complex velocity $\mathrm{Re}\{\mathcal{V}^{\nu}(\hat{x}(\tau',\omega'))\}$
is stochastic-Taylor expanded by Eq.(\ref{eq: Taylor}), too. Now,
the RR field $\mathfrak{F}$ is evaluated such that
\begin{align}
\mathfrak{F}(\hat{x}(\tau,\omega))-\frac{\mathcal{F}_{(\mathring{+})}(\hat{x}(\tau,\omega))-\mathcal{F}_{(\mathring{-})}(\hat{x}(\tau,\omega))}{2} & =O(R(x),R(\mathcal{V}),\hat{x}(\tau,\omega)-\hat{x}(\tau,\omega'))
\end{align}
by following the similar calculation given by {\bf Ref.}\cite{A.Sokolov(1986),Dirac(1938),Hartemann(2002)}.
For $F_{\mathrm{ex}}$ the external field(s) such that $\partial_{\mu}F_{\mathrm{ex}}^{\mu\nu}=0$
as a laser field, quantum dynamics corresponding to Eqs.(\ref{eq: LAD eq}-\ref{eq: LAD field})
is hereby imposed w.r.t.\ $F=F_{\mathrm{ex}}+\mathfrak{F}$ in Eqs.(\ref{eq: EOM_St},\ref{eq: Maxwell_St})
by rounding $O(R(x),R(\mathcal{V}),\hat{x}(\tau,\omega)-\hat{x}(\tau,\omega'))$
into $\delta f$ the singularity of the radiation field:
\begin{align}
m_{0}\mathfrak{D_{\tau}}\mathcal{V}^{\mu}(\hat{x}(\tau,\omega))= & -eF_{\mathrm{ex}}^{\mu\nu}(\hat{x}(\tau,\omega))\mathcal{V}_{\nu}(\hat{x}(\tau,\omega))-e\mathfrak{F}^{\mu\nu}(\hat{x}(\tau,\omega))\mathcal{V}_{\nu}(\hat{x}(\tau,\omega))\label{eq: EOM-RR}
\end{align}
\begin{align}
\mathfrak{F}^{\mu\nu}(\hat{x}(\tau,\omega)) & =-\frac{m_{0}\tau_{0}}{ec^{2}}\int_{\varOmega_{(\tau,\omega)}}d\mathscr{P}(\omega')\left[\begin{gathered}\dot{a}^{\mu}(\hat{x}(\tau,\omega'))\cdot\mathrm{Re}\{\mathcal{V}^{\nu}(\hat{x}(\tau,\omega'))\}\\
-\dot{a}^{\nu}(\hat{x}(\tau,\omega'))\cdot\mathrm{Re}\{\mathcal{V}^{\mu}(\hat{x}(\tau,\omega'))\}
\end{gathered}
\right]\label{eq: field-RR}
\end{align}
\begin{align}
\dot{a}^{\mu}(x) & \coloneqq\frac{c^{4}}{[\mathrm{Re}\{\mathcal{V}_{\alpha}(x)\}\cdot\mathrm{Re}\{\mathcal{V}^{\alpha}(x)\}]^{2}}\mathrm{Re}\{\mathfrak{D}_{\tau}^{2}\mathcal{V}(x)\}\nonumber \\
 & \quad\quad\quad\quad\quad\quad\quad\quad-\frac{27}{8}\frac{c^{4}\mathrm{Re}\{\mathcal{V}_{\alpha}(x)\}\cdot\mathrm{Re}\{\mathfrak{D}_{\tau}\mathcal{V}^{\alpha}(x)\}}{[\mathrm{Re}\{\mathcal{V}_{\alpha}(x)\}\cdot\mathrm{Re}\{\mathcal{V}^{\alpha}(x)\}]^{3}}\mathrm{Re}\{\mathfrak{D}_{\tau}\mathcal{V}^{\mu}(x)\}\label{eq: a-dot}
\end{align}
Where, $\mathrm{Re}\left\{ \mathcal{V}(\hat{x}(\tau,\omega))\right\} $
doesn't satisfies $v_{\alpha}v^{\alpha}=c^{2}$ of a common rule in
classical dynamics. The formulation of RR with $v_{\alpha}v^{\alpha}\neq c^{2}$
is found in {\bf Ref.}\cite{Barut-Unal}\footnote{Barut and Unal proposed their model as radiation reaction acting on
a spinning particle and they considered $v_{\alpha}v^{\alpha}\neq c^{2}$
is a natural requirement to express its Zitterbewegung. We can employ
the same calculation in the second page of {\bf Ref.}\cite{Barut-Unal}
to derive Eqs.(\ref{eq: EOM-RR}-\ref{eq: a-dot}) since $\mathrm{Re}\{\mathcal{V}_{\alpha}(x)\}\cdot\mathrm{Re}\{\mathcal{V}^{\alpha}(x)\}\neq0$.}. 

Consider that average trajectory. By restricting
\begin{equation}
\varOmega_{\tau}^{\mathrm{ave}}\coloneqq\left.\varOmega_{(\tau,\omega)}\right|_{\hat{x}(\tau,\omega)=\langle\hat{x}^{\nu}\rangle_{\tau}},
\end{equation}
$\mathscr{P}(\varOmega_{\tau}^{\mathrm{ave}})$ is the probability
which a quanta stays at $\langle\hat{x}\rangle_{\tau}\coloneqq\mathbb{E}\llbracket\hat{x}(\tau,\bullet)\rrbracket$.
Then, the lowest order of Ehrenfest's theorem (Eq.(\ref{eq: Ehrenfest-expansion}))
is below:
\begin{align}
m_{0}\frac{d^{2}\langle\hat{x}^{\mu}\rangle_{\tau}}{d\tau^{2}} & -e\left[F_{\mathrm{ex}}^{\mu\nu}(\langle\hat{x}\rangle_{\tau})+\mathfrak{F}^{\mu\nu}(\langle\hat{x}\rangle_{\tau})\right]\frac{d\langle\hat{x}_{\nu}\rangle_{\tau}}{d\tau}+O(\langle\otimes^{2}\delta\hat{x}\rangle_{\tau})\label{eq: EOM-RR-ave}
\end{align}
\begin{align}
\mathfrak{F}^{\mu\nu}(\langle\hat{x}\rangle_{\tau}) & =-\frac{m_{0}\tau_{0}\mathscr{P}(\varOmega_{\tau}^{\mathrm{ave}})}{ec^{2}}\left[\frac{d^{3}\langle\hat{x}^{\mu}\rangle_{\tau}}{d\tau^{3}}\cdot\frac{d\langle\hat{x}^{\nu}\rangle_{\tau}}{d\tau}-\frac{d^{3}\langle\hat{x}^{\nu}\rangle_{\tau}}{d\tau^{3}}\cdot\frac{d\langle\hat{x}^{\mu}\rangle_{\tau}}{d\tau}\right]\label{eq: field-RR-ave}
\end{align}
A trajectory of $\langle\hat{x}^{\mu}\rangle_{\tau}$ is drawn by
Eqs.(\ref{eq: EOM-RR-ave}-\ref{eq: field-RR-ave}). The non-relativistic
limit of Eqs.(\ref{eq: EOM-RR-ave}-\ref{eq: field-RR-ave}) is found
in {\bf Ref.}\cite{Ozaki-Sasabe}. Where, the following simple relation
is obtained: 
\begin{equation}
\mathfrak{F}(\langle\hat{x}\rangle_{\tau})=\mathscr{P}(\varOmega_{\tau}^{\mathrm{ave}})\times F_{\mathrm{LAD}}(\langle\hat{x}\rangle_{\tau})\label{eq: field-relation1}
\end{equation}

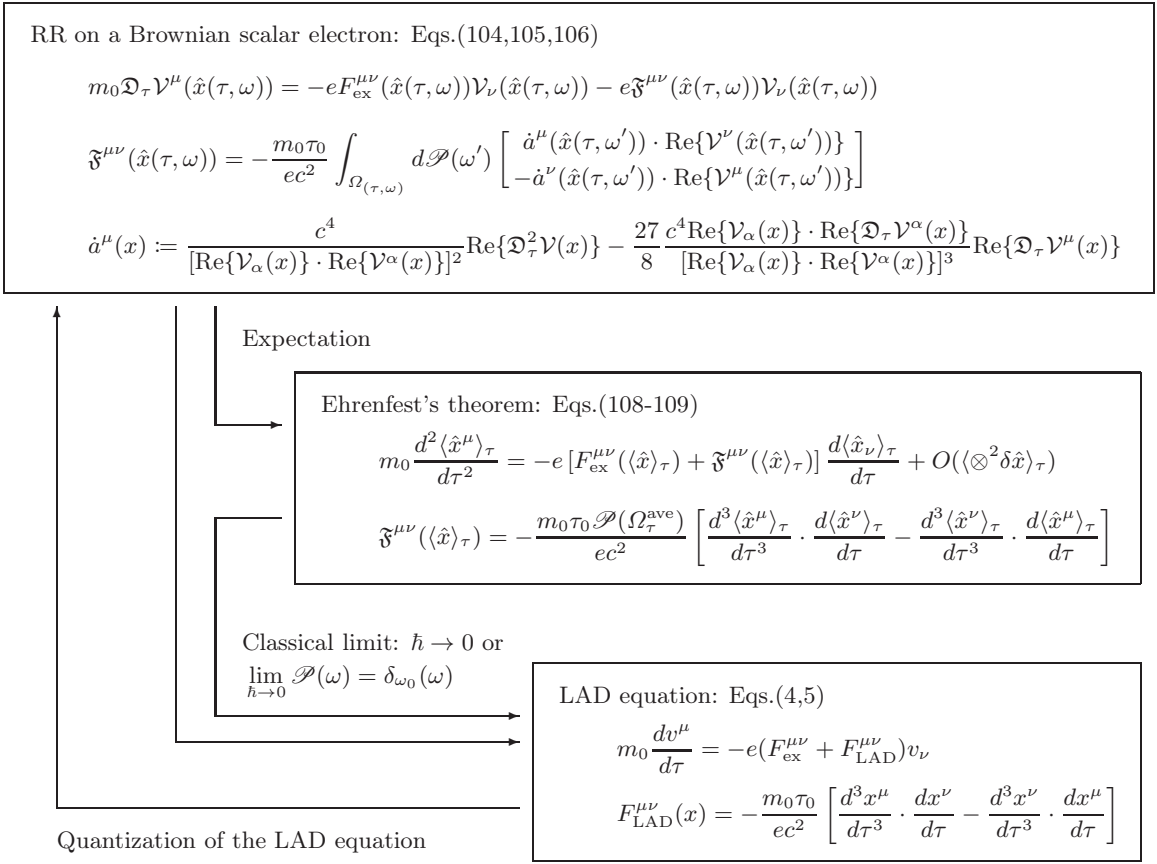
\begin{figure*}[t]
\noindent \begin{centering}
\begin{picture}(500,330)

\put(0,325){\line(1,0){435}}
\put(0,215){\line(1,0){435}}
\put(0,325){\line(0,-1){110}}
\put(435,325){\line(0,-1){110}}

\put(10,310){\fontsize{9pt}{0cm}\selectfont RR on a Brownian scalar electron: Eqs.(\ref{eq: EOM-RR},\ref{eq: field-RR},\ref{eq: a-dot})}

\put(30,290){$ \fontsize{9pt}{0cm}\selectfont \begin{gathered} m_{0}\mathfrak{D_{\tau}}\mathcal{V}^{\mu}(\hat{x}(\tau,\omega))= -eF_{\mathrm{ex}}^{\mu\nu}(\hat{x}(\tau,\omega))\mathcal{V}_{\nu}(\hat{x}(\tau,\omega))\nonumber  -e\mathfrak{F}^{\mu\nu}(\hat{x}(\tau,\omega))\mathcal{V}_{\nu}(\hat{x}(\tau,\omega))  \end{gathered} $}

\put(30,262){$ \fontsize{9pt}{0cm}\selectfont \begin{gathered} \mathfrak{F}^{\mu\nu}(\hat{x}(\tau,\omega)) = -\frac{m_{0}\tau_{0}}{ec^{2}}\int_{\varOmega_{(\tau,\omega)}}d\mathscr{P}(\omega')  \left[ \begin{gathered} \dot{a}^{\mu}(\hat{x}(\tau,\omega'))\cdot\mathrm{Re}\{\mathcal{V}^{\nu}(\hat{x}(\tau,\omega'))\}\\ -\dot{a}^{\nu}(\hat{x}(\tau,\omega'))\cdot\mathrm{Re}\{\mathcal{V}^{\mu}(\hat{x}(\tau,\omega'))\}  \end{gathered} \right] \end{gathered} $}

\put(30,230){$ \fontsize{9pt}{0cm}\selectfont \begin{gathered} \dot{a}^{\mu}(x)\coloneqq  \frac{c^{4}}{[\mathrm{Re}\{\mathcal{V}_{\alpha}(x)\}\cdot\mathrm{Re}\{\mathcal{V}^{\alpha}(x)\}]^{2}}\mathrm{Re}\{\mathfrak{D}_{\tau}^{2}\mathcal{V}(x)\} -\frac{27}{8}\frac{c^{4}\mathrm{Re}\{\mathcal{V}_{\alpha}(x)\}\cdot\mathrm{Re}\{\mathfrak{D}_{\tau}\mathcal{V}^{\alpha}(x)\}}{[\mathrm{Re}\{\mathcal{V}_{\alpha}(x)\}\cdot\mathrm{Re}\{\mathcal{V}^{\alpha}(x)\}]^{3}}\mathrm{Re}\{\mathfrak{D}_{\tau}\mathcal{V}^{\mu}(x)\}  \end{gathered} $}

\put(110,185){\line(1,0){320}}
\put(110,105){\line(1,0){320}}
\put(110,185){\line(0,-1){80}}
\put(430,185){\line(0,-1){80}}

\put(120,170){\fontsize{9pt}{0cm}\selectfont Ehrenfest's theorem: Eqs.(\ref{eq: EOM-RR-ave}-\ref{eq: field-RR-ave})}

\put(140,150){$ \fontsize{9pt}{0cm}\selectfont \begin{gathered} m_{0}\frac{d^{2}\langle\hat{x}^{\mu}\rangle_{\tau}}{d\tau^{2}}=-e\left[F_{\mathrm{ex}}^{\mu\nu}(\langle\hat{x}\rangle_{\tau})+\mathfrak{F}^{\mu\nu}(\langle\hat{x}\rangle_{\tau})\right]\frac{d\langle\hat{x}_{\nu}\rangle_{\tau}}{d\tau} +O(\langle\otimes^{2}\delta\hat{x}\rangle_{\tau}) \end{gathered} $}

\put(140,120){$ \fontsize{9pt}{0cm}\selectfont \begin{gathered} \mathfrak{F}^{\mu\nu}(\langle\hat{x}\rangle_{\tau})=-\frac{m_{0}\tau_{0}\mathscr{P}(\varOmega_{\tau}^{\mathrm{ave}})}{ec^{2}}  \left[\frac{d^{3}\langle\hat{x}^{\mu}\rangle_{\tau}}{d\tau^{3}}\cdot\frac{d\langle\hat{x}^{\nu}\rangle_{\tau}}{d\tau}-\frac{d^{3}\langle\hat{x}^{\nu}\rangle_{\tau}}{d\tau^{3}}\cdot\frac{d\langle\hat{x}^{\mu}\rangle_{\tau}}{d\tau}\right] \end{gathered} $}

\put(200,75){\line(1,0){230}}
\put(200,0){\line(1,0){230}}
\put(200,75){\line(0,-1){75}}
\put(430,75){\line(0,-1){75}}

\put(210,60){\fontsize{9pt}{0cm}\selectfont LAD equation: Eqs.(\ref{eq: LAD eq},\ref{eq: LAD field})}

\put(230,40){$ \fontsize{9pt}{0cm}\selectfont \begin{gathered} m_{0}\frac{dv^{\mu}}{d\tau}  =-e(F_{\mathrm{ex}}^{\mu\nu}+F_{\mathrm{LAD}}^{\mu\nu})v_{\nu} \end{gathered} $}

\put(230,15){$ \fontsize{9pt}{0cm}\selectfont \begin{gathered} F_{\mathrm{LAD}}^{\mu\nu}(x)  =  -\frac{m_{0}\tau_{0}}{ec^{2}}\left[\frac{d^{3}x^{\mu}}{d\tau^{3}}\cdot\frac{dx^{\nu}}{d\tau}-\frac{d^{3}x^{\nu}}{d\tau^{3}}\cdot\frac{dx^{\mu}}{d\tau}\right] \end{gathered} $}

\put(90,195){\fontsize{9pt}{0cm}\selectfont Expectation}
\put(80,210){\line(0,-1){45}}
\put(80,165){\vector(1,0){25}}

\put(80,130){\line(1,0){25}}
\put(80,130){\line(0,-1){75}}
\put(80,55){\vector(1,0){115}}

\put(90,80){\fontsize{9pt}{0cm}\selectfont Classical limit: $ \hbar \rightarrow 0 $ or}
\put(90,65){$ \fontsize{9pt}{0cm}\selectfont \begin{gathered} \lim_{\hbar\rightarrow0}\mathscr{P}(\omega)=\delta_{\omega_{0}}(\omega) \end{gathered} $}

\put(65,210){\line(0,-1){165}}
\put(65,45){\vector(1,0){130}}

\put(20,5){\fontsize{9pt}{0cm}\selectfont Quantization of the LAD equation}
\put(20,20){\vector(0,1){190}}
\put(20,20){\line(1,0){175}}

\end{picture}
\par\end{centering}
\caption{\label{summary-RR}The relation of the RR models. The present model
for a Brownian scalar electron is the set of Eqs.(\ref{eq: EOM-RR},\ref{eq: field-RR},\ref{eq: a-dot}).
The set of Eqs.(\ref{eq: EOM-RR-ave}-\ref{eq: field-RR-ave}) is
imposed by Ehrenfest's theorem w.r.t.\ Eqs.(\ref{eq: EOM-RR},\ref{eq: field-RR},\ref{eq: a-dot}).
The classical limit of Eqs.(\ref{eq: EOM-RR},\ref{eq: field-RR},\ref{eq: a-dot})
or Eqs.(\ref{eq: EOM-RR-ave}-\ref{eq: field-RR-ave}) provides the
LAD equation (\ref{eq: LAD eq},\ref{eq: LAD field}). Conversely,
the direction from the LAD equation to the set of Eqs.(\ref{eq: EOM-RR},\ref{eq: field-RR},\ref{eq: a-dot})
represents its quantization.}
\end{figure*}

\subsection{Classical limit\label{classical limit}}

What will happen with Eq.(\ref{eq: Kinematics_St}) and Eqs.(\ref{eq: EOM-RR}-\ref{eq: a-dot})
in   $\hbar\rightarrow0$? The randomness of Eq.(\ref{eq: Kinematics_St})
is neglected since $\lim_{\hbar\rightarrow0}d_{\pm}\hat{x}(\tau,\omega)=\lim_{\hbar\rightarrow0}\mathcal{V}_{\pm}^{\mu}(\hat{x}(\tau,\omega))d\tau$
with $\lim_{\hbar\rightarrow0}[\mathcal{V}_{+}^{\mu}(\hat{x}(\tau,\omega))-\mathcal{V}_{-}^{\mu}(\hat{x}(\tau,\omega))]=0$,
then, its trajectory becomes a smooth and differentiable function.
Hence, the all sample paths are identified as $\lim_{\hbar\rightarrow0}\langle\hat{x}\rangle_{\tau}$.
Let us express the classical limit in this model by employing the
following equivalency relation
\begin{equation}
\omega\sim\omega'\Leftrightarrow\forall\tau,\lim_{\hbar\rightarrow0}\left[\hat{x}(\tau,\omega)-\hat{x}(\tau,\omega')\right]=0
\end{equation}
and its equivalency class
\begin{equation}
[\omega]=\left\{ \left.\omega'\in\varOmega\right|\omega\sim\omega'\right\} .
\end{equation}
Then there is $\omega_{0}$ the representative of $[\omega]$, $\{\hat{x}(\tau,\omega_{0})\}_{\tau\in\mathbb{R}}$
is the smooth trajectory given in the classical limit. Namely, we
understand it as all of $\omega'\in[\omega]$ converge to $\omega_{0}$
in $\hbar\rightarrow0$. Thus, the probability becomes the Dirac measure
$\delta_{\omega_{0}}$:
\begin{equation}
\lim_{\hbar\rightarrow0}\mathscr{P}(\omega)=\delta_{\omega_{0}}(\omega)
\end{equation}
Where, $\int_{\varOmega}f(\omega)d\delta_{\omega_{0}}(\omega)=f(\omega_{0})$.
Alternatively, we can express the same thing by Eq.(\ref{eq: Fokker-Planck})
with $\hbar\rightarrow0$, 
\begin{align}
\partial_{\tau}p(x,\tau)+\partial_{\mu}[\mathcal{V}_{\pm}^{\mu}(x)p(x,\tau)] & =0
\end{align}
the equation of continuity with $\mathcal{V}_{+}(x)=\mathcal{V}_{-}(x)$.
In this case, an initial profile $\delta^{4}(x-x_{\mathrm{initial}})$
propagates by keeping its profile of the delta distribution. This
is represented by $\delta_{\omega_{0}}(\omega)$ with labels of sample
paths. Of course, $\langle\hat{x}\rangle_{\tau}$ should be included
in the class of $[\omega]$. Therefore by Eq.(\ref{eq: field-relation1}), 

\begin{align}
\lim_{\hbar\rightarrow0}\mathfrak{F}(\hat{x}(\tau,\omega)) & =\lim_{\hbar\rightarrow0}\mathfrak{F}(\langle\hat{x}\rangle_{\tau})\nonumber \\
 & =F_{\mathrm{LAD}}(x(\tau)).\label{eq: field-relation2}
\end{align}
Thus, the LAD equation (\ref{eq: LAD eq}-\ref{eq: LAD field}) is
derived from Eqs.(\ref{eq: EOM-RR}-\ref{eq: a-dot}) by its classical
limit. {\bf Figure \ref{summary-RR}} shows the relation between each
models which we discussed.

\subsection{Landau-Lifshitz's approximation}

The instability (run-away) of Eqs.(\ref{eq: EOM-RR}-\ref{eq: a-dot})
is expected like one of the LAD equation (\ref{eq: LAD eq}-\ref{eq: LAD field})
\cite{Dirac(1938)} since it includes the high-order derivative of
$\mathrm{Re}\{\mathfrak{D}_{\tau}^{2}\mathcal{V}(x)\}$. The Landau-Lifshitz
({\bf LL}) approximation, the perturbation w.r.t.\ $\tau_{0}$ is
normally applied to Eq.(\ref{eq: LAD eq}-\ref{eq: LAD field}) for
avoiding this complexity \cite{Landau-Lifshitz}. The version for
Eqs.(\ref{eq: EOM-RR}-\ref{eq: a-dot}) is realized by the following:
\begin{align}
\dot{a}^{\mu}(x)= & \frac{c^{4}\times\mathrm{Re}\{\ddot{\mathcal{V}}_{\mathrm{approx.}}^{\mu}(x)\}}{[\mathrm{Re}\{\mathcal{V}_{\alpha}(x)\}\cdot\mathrm{Re}\{\mathcal{V}^{\alpha}(x)\}]^{2}}+O(\tau_{0})\label{eq: a-dot-LL}
\end{align}
\begin{align}
\mathrm{Re}\{\ddot{\mathcal{V}}_{\mathrm{approx.}}^{\mu}\} & \coloneqq-\frac{e}{m_{0}}\mathrm{Re}\{\mathcal{V}_{\alpha}\}\cdot\mathrm{Re}\{\mathcal{V}^{\beta}\}\partial_{\beta}F_{\mathrm{ex}}^{\mu\alpha}+\frac{e}{m_{0}}\mathrm{Im}\{\mathcal{V}_{\alpha}\}\cdot\mathrm{Im}\{\mathcal{V}^{\beta}\}\partial_{\beta}F_{\mathrm{ex}}^{\mu\alpha}\nonumber \\
 & \quad\quad+\frac{e^{2}}{m_{0}^{2}}g_{\alpha\beta}F_{\mathrm{ex}}^{\mu\alpha}F_{\mathrm{ex}}^{\beta\gamma}\mathrm{Re}\{\mathcal{V}_{\gamma}\}\label{eq: a-dot-LL2}
\end{align}
Where, $\mathrm{Re}\{\mathcal{V}_{\alpha}(x)\}\cdot\mathrm{Re}\{\mathfrak{D}_{\tau}\mathcal{V}^{\alpha}(x)\}=O(\tau_{0})$
for $\partial_{\mu}F_{\mathrm{ex}}^{\mu\nu}=0$. Ehrenfest's theorem
of Eqs.(\ref{eq: EOM-RR}-\ref{eq: field-RR}) with Eqs.(\ref{eq: a-dot-LL}-\ref{eq: a-dot-LL2})
becomes
\begin{align}
m_{0}\frac{d^{2}\langle\hat{x}^{\mu}\rangle_{\tau}}{d\tau^{2}} & =-eF_{\mathrm{ex}}^{\mu\nu}(\langle\hat{x}\rangle_{\tau})\frac{d\langle\hat{x}_{\nu}\rangle_{\tau}}{d\tau}-e\mathscr{P}(\varOmega_{\tau}^{\mathrm{ave}})F_{\mathrm{LL}}^{\mu\nu}(\langle\hat{x}\rangle_{\tau})\frac{d\langle\hat{x}_{\nu}\rangle_{\tau}}{d\tau}+O(\tau_{0}^{2},\langle\otimes^{2}\delta\hat{x}\rangle_{\tau})\label{eq: LL_eq_ave}
\end{align}
since $\mathrm{Im}\{\mathcal{V}(\langle\hat{x}\rangle_{\tau})\}=O(\langle\otimes^{2}\delta\hat{x}\rangle_{\tau})$,
and 
\begin{align}
F_{\mathrm{LL}}^{\mu\nu}(x(\tau)) & \coloneqq\tau_{0}\partial_{\alpha}F_{\mathrm{ex}}^{\mu\nu}(x(\tau))\cdot\frac{dx^{\alpha}}{d\tau}(\tau)\nonumber \\
 & \quad\quad-\frac{e\tau_{0}}{m_{0}c^{2}}\left(\delta_{\theta}^{\mu}\delta_{\epsilon}^{\nu}-\delta_{\theta}^{\nu}\delta_{\epsilon}^{\mu}\right)g_{\alpha\beta}F_{\mathrm{ex}}^{\theta\alpha}(x(\tau))F_{\mathrm{ex}}^{\beta\gamma}(x(\tau))\frac{dx_{\gamma}}{d\tau}(\tau)\frac{dx^{\epsilon}}{d\tau}(\tau).
\end{align}
When $\mathscr{P}(\varOmega_{\tau}^{\mathrm{ave}})=1$, Eq.(\ref{eq: LL_eq_ave})
becomes the LL equation $m_{0}dv^{\mu}/dt=-e(F_{\mathrm{ex}}^{\mu\nu}+F_{\mathrm{LL}}^{\mu\nu})v_{\nu}$
perfectly \cite{Landau-Lifshitz}. Of course it means the classical
limit of Eq.(\ref{eq: LL_eq_ave}). 

\subsection{Radiation formula}

By Eqs.(\ref{eq: EOM-RR-ave}-\ref{eq: field-RR-ave}) for $\{\langle\hat{x}\rangle_{\tau}\}_{\tau\in\mathbb{R}}$,
the radiation formula in the present model (the issue-$\langle\mathrm{E}\rangle$)
is found:
\begin{align}
m_{0}\frac{d^{2}\langle\hat{x}^{\mu}\rangle_{\tau}}{d\tau^{2}} & \approx-eF_{\mathrm{ex}}^{\mu\nu}(\langle\hat{x}\rangle_{\tau})\frac{d\langle\hat{x}_{\nu}\rangle_{\tau}}{d\tau}-\frac{dW_{\mathrm{stochastic}}}{dt}\frac{d\langle\hat{x}^{\mu}\rangle_{\tau}}{d\tau}
\end{align}
\begin{equation}
\frac{dW_{\mathrm{stochastic}}}{dt}\coloneqq-m_{0}\tau_{0}\mathscr{P}(\varOmega_{\tau}^{\mathrm{ave}})\frac{d^{2}\langle\hat{x}_{\mu}\rangle_{\tau}}{d\tau^{2}}\cdot\frac{d^{2}\langle\hat{x}^{\mu}\rangle_{\tau}}{d\tau^{2}}\label{eq: st-radiation-formula}
\end{equation}
It should be compared with Eq.(\ref{eq: Rad-formula-sQED}) and Larmor's
formula
\begin{align}
\frac{dW_{\mathrm{classical}}}{dt} & =-m_{0}\tau_{0}\frac{dv_{\mu}}{d\tau}\cdot\frac{dv^{\mu}}{d\tau}.
\end{align}
When an external field is a plane wave, this has to converge to one
by {\bf Ref.} \cite{A.Sokolov(1986)}, 
\begin{align}
\mathscr{P}(\varOmega_{\tau}^{\mathrm{ave}}) & =q_{\mathrm{sQED}}(\chi)\nonumber \\
 & =\frac{9\sqrt{3}}{8\pi}\int_{0}^{\chi^{-1}}dr\,r\int_{\frac{r}{1-\chi r}}^{\infty}dr'\,K_{5/3}(r').
\end{align}
For $\mathscr{P}(\varOmega_{\tau}^{\mathrm{ave}})=q_{\mathrm{sQED}}(\chi)+\delta\mathscr{P}(\varOmega_{\tau}^{\mathrm{ave}})$
in the case of non-plane wave fields, the radiation spectrum becomes
\begin{equation}
\frac{d^{2}W_{\mathrm{stochastic}}}{dtd(\hbar\omega)}=\left[\frac{dq_{\mathrm{sQED}}(\chi)}{d(\hbar\omega)}+\frac{d\delta\mathscr{P}(\varOmega_{\tau}^{\mathrm{ave}})}{d(\hbar\omega)}\right]\frac{dW_{\mathrm{classical}}}{dt},
\end{equation}
thus, the observation of $d\delta\mathscr{P}(\varOmega_{\tau}^{\mathrm{ave}})/d(\hbar\omega)$
provides us an unknown correction in non-linear QED. By solving Eq.(\ref{eq: Fokker-Planck}),
$\delta\mathscr{P}(\varOmega_{\tau}^{\mathrm{ave}})$ is found even
in the case of general field profiles like laser focusing and superposition.

\section{Conclusion}

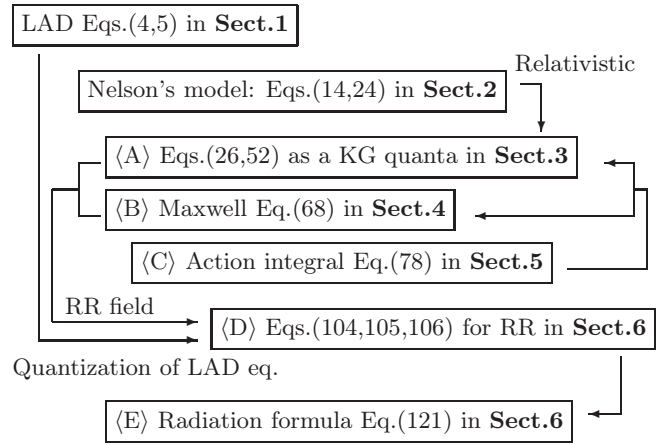
\begin{figure}
\begin{centering}
\begin{picture}(250,180)

\put(25,135){\fbox{\fontsize{9pt}{0cm}\selectfont Nelson's model: Eqs.(\ref{eq: st-dif},\ref{eq: Nelson-schrodinger-2}) in {\bf Sect.\ref{Nelson-1D}}}}

\put(200,140){\vector(0,-1){18}}
\put(200,140){\line(-1,0){8}}
\put(190,145){\fontsize{9pt}{0cm}\selectfont Relativistic}

\put(35,110){\fbox{\fontsize{9pt}{0cm}\selectfont $\langle\mathrm{A}\rangle$ Eqs.(\ref{eq: Kinematics_St},\ref{eq: EOM_St}) as a KG quanta in {\bf Sect.\ref{kinematics_dynamics4D}}}}

\put(20,53){\fontsize{9pt}{0cm}\selectfont RR field}

\put(35,90){\fbox{\fontsize{9pt}{0cm}\selectfont $\langle\mathrm{B}\rangle$ Maxwell Eq.(\ref{eq: Maxwell_St}) in {\bf Sect.\ref{Maxwell}}}}

\put(25,90){\line(0,1){20}}
\put(25,110){\line(1,0){8}}
\put(25,90){\line(1,0){8}}
\put(15,100){\line(1,0){10}}
\put(15,50){\line(0,1){50}}
\put(15,50){\vector(1,0){55}}

\put(45,70){\fbox{\fontsize{9pt}{0cm}\selectfont $\langle\mathrm{C}\rangle$ Action integral Eq.(\ref{eq: action integral}) in {\bf Sect.\ref{Action}}}}

\put(235,110){\line(0,-1){20}}
\put(235,110){\vector(-1,0){11}}
\put(235,90){\vector(-1,0){60}}
\put(240,100){\line(-1,0){5}}
\put(240,70){\line(0,1){30}}
\put(240,70){\line(-1,0){30}}

\put(0,160){\fbox{\fontsize{9pt}{0cm}\selectfont LAD Eqs.(\ref{eq: LAD eq},\ref{eq: LAD field}) in {\bf Sect.\ref{Intro}}}}
\put(10,43){\line(0,1){108}}
\put(10,43){\vector(1,0){60}}

\put(75,45){\fbox{\fontsize{9pt}{0cm}\selectfont $\langle\mathrm{D}\rangle$ Eqs.(\ref{eq: EOM-RR},\ref{eq: field-RR},\ref{eq: a-dot}) for RR in {\bf Sect.\ref{RR}}}}
\put(0,30){\fontsize{9pt}{0cm}\selectfont Quantization of LAD eq.}

\put(35,10){\fbox{\fontsize{9pt}{0cm}\selectfont $\langle\mathrm{E}\rangle$ Radiation formula Eq.(\ref{eq: st-radiation-formula}) in {\bf Sect.\ref{RR}}}}
\put(230,15){\line(0,1){22}}
\put(230,15){\vector(-1,0){12}}

\end{picture}
\par\end{centering}
\caption{\label{summary}Key equations in this article. Our purpose was the
quantization of Eqs.(\ref{eq: LAD eq},\ref{eq: LAD field}) the LAD
equation {[}{\bf Sect.\ref{Intro}}{]}. For obtaining Eqs.(\ref{eq: EOM-RR},\ref{eq: field-RR},\ref{eq: a-dot})
of its quantized form, we proposed $\langle\mathrm{A}\rangle$ Eqs.(\ref{eq: Kinematics_St},\ref{eq: EOM_St})
the relativistic version of Nelson's stochastic quantization for a
scalar electron {[}{\bf Sect.\ref{kinematics_dynamics4D}}{]} and
$\langle\mathrm{B}\rangle$ Eq.(\ref{eq: Maxwell_St}) the Maxwell
equation {[}{\bf Sect.\ref{Maxwell}}{]}. We also gave $\langle\mathrm{C}\rangle$
the action integral for Eq.(\ref{eq: EOM_St}) and Eqs.(\ref{eq: Maxwell_St})
by {\bf Sect.\ref{Action}}. Then, $\langle\mathrm{D}\rangle$ RR
in relativistic quantum dynamics was obtained by solving the Maxwell
equation in {\bf Sect.\ref{RR}}. The quantumness of radiation was
found in $\langle\mathrm{E}\rangle$ a radiation formula by Eq.(\ref{eq: st-radiation-formula}).}
\end{figure}
The topics discussed in this paper are illustrated by {\bf Fig.\ref{summary}}.
We derived the set of Eqs.(\ref{eq: EOM-RR}-\ref{eq: a-dot}), which
is the quantized equation of the LAD equation (\ref{eq: LAD eq}-\ref{eq: LAD field})
via the construction of Issue-$\langle\mathrm{A}\rangle$ Eq.(\ref{eq: Kinematics_St})
of the relativistic kinematics with Eq.(\ref{eq: EOM_St}) of the
dynamics for a Brownian scalar electron proposed in {\bf Sect.\ref{kinematics_dynamics4D}},
and Issue-$\langle\mathrm{B}\rangle$ the Maxwell equation (\ref{eq: Maxwell_St})
in {\bf Sect.\ref{Maxwell}}. Issue-$\langle\mathrm{C}\rangle$ of
the action integral was introduced by Eq.(\ref{eq: action integral})
in {\bf Sect.\ref{Action}}. The consistent system w.r.t.\ Issues-$\langle\mathrm{A\mathchar`-C}\rangle$
was first proposed by the present article. Hence, the description
of RR by solving Eq.(\ref{eq: Maxwell_St}) is that first example
of stochastic quantization. Again, Issue-$\langle\mathrm{D}\rangle$
the stochastic quantization of the LAD equation was given by Eqs.(\ref{eq: EOM-RR}-\ref{eq: a-dot})
in {\bf Sect.\ref{RR}}. We also confirmed its classical limit becomes
the LAD equation. The LL approximation was introduced by Eqs.(\ref{eq: EOM-RR},\ref{eq: field-RR},\ref{eq: a-dot-LL},\ref{eq: a-dot-LL2}).
The readers can understand the fact that we did not employ any restriction
of the external laser fields except the Lorenz gauge in this article.
We obtained Issue-$\langle\mathrm{E}\rangle$ the radiation formula
of Eq.(\ref{eq: st-radiation-formula}) in {\bf Sect.\ref{RR}} by
Ehrenfest's theorem of Eqs.(\ref{eq: EOM-RR-ave},\ref{eq: field-RR-ave}).
By the comparison between Eq.(\ref{eq: Rad-formula-sQED}) and Eq.(\ref{eq: st-radiation-formula}),
we found $q(\chi)$ is included in the existence probability $\mathscr{P}(\varOmega_{\tau}^{\mathrm{ave}})$
which a Brownian quanta stays at its average position $\langle\hat{x}\rangle_{\tau}$.
Hence, the observation of quantumness $q(\chi)$ indicates the existence
probability of a radiating charged quanta. The calculation of Eq.(\ref{eq: Fokker-Planck})
is the requirement to derive $\mathscr{P}(\varOmega_{\tau}^{\mathrm{ave}})$.
Now, a trajectory of a radiating scalar electron can be drawn by
a stochastic process. The precise analysis of $\delta f$ the field
singularity as the Coulomb field in Eq.(\ref{eq: Maxwell_St}) should
be performed. The existence of $d\delta\mathscr{P}(\varOmega_{\tau}^{\mathrm{ave}})/d(\hbar\omega)$
suggests a new correction beyond the Furry picture which may be found
in high-intensity laser experiments. 

\section*{Acknowledgements\addcontentsline{toc}{section}{Acknowledgements}}

KS acknowledges Prof. Kazuo A. Tanaka for useful discussion and the
support from the Extreme Light Infrastructure Nuclear Physics (ELI-NP)
Phase II, a project co-financed by the Romanian Government and the
European Union through the European Regional Development Fund - the
Competitiveness Operational Programme (1/07.07.2016, COP, ID 1334).

\section*{\addcontentsline{toc}{section}{References}}

\newpage{}

\part*{Supplemental material\addcontentsline{toc}{part}{Supplemental material}}

\appendix

\section{Mathematical supports\label{APP-A}}

Since we have not discussed the detail of stochastic processes, let
us see the precise construction from the 1D WP to the 4D kinematics
in this appendix. Thus, our purpose in this Appendix \ref{APP-A}
is to give Eq.(\ref{eq: Kinematics_St}), Eq.(\ref{eq: Ito formula})
and Eq.(\ref{eq: C-Ito formula}) mathematically.

\subsection{Mathematical spaces}

At first, we define the Minkowski spacetime and a probability space
as measure spaces. Where, $\mathscr{B}(I)$ denotes a Borel $\sigma$-algebra
of a topological space $I$. 
\begin{defn}[Minkowski spacetime]
\begin{leftbar}Let $\mathbb{A}^{4}(\mathbb{V_{\mathrm{M}}^{\mathrm{4}}},g)$
be a 4D metric affine space w.r.t.\ a 4D standard vector space $\mathbb{V_{\mathrm{M}}^{\mathrm{4}}}$
and its metric $g$ on $\mathbb{V_{\mathrm{M}}^{\mathrm{4}}}$. By
defining a measure space 
\[
(\mathbb{A}^{4}(\mathbb{V_{\mathrm{M}}^{\mathrm{4}}},g),\mathscr{B}(\mathbb{A}^{4}(\mathbb{V_{\mathrm{M}}^{\mathrm{4}}},g)),\mu),
\]
we regard this as the Minkowski spacetime when $g=(+,-,-,-)$. The
measure $\mu:\mathbb{A}^{4}(\mathbb{V_{\mathrm{M}}^{\mathrm{4}}},g)\rightarrow\mathbb{R}$
is defined as $d\mu(x)=dx^{4}$.\end{leftbar}
\end{defn}

A metric affine space is an abstract mathematical space without its
origin and its coordinate. By the coordinate mapping $\varphi(x)\coloneqq(x^{0},x^{1},x^{2},x^{3})$
for all $x\in\mathbb{V_{\mathrm{M}}^{\mathrm{4}}}$ and $x\in\mathbb{A}^{4}(\mathbb{V_{\mathrm{M}}^{\mathrm{4}}},g)$
with its origin, the readers can consider that $\mathbb{A}^{4}(\mathbb{V_{\mathrm{M}}^{\mathrm{4}}},g)$
and $\mathbb{V_{\mathrm{M}}^{\mathrm{4}}}$ is identified as $\mathbb{R}^{4}$
with the metric $g$. 
\begin{defn}[Probability space]
\begin{leftbar}For a certain abstract non-empty set $\varOmega$,
$\mathcal{F}$ a $\sigma$-algebra of $\varOmega$, and $\mathscr{P}$
a probability measure on the measurable space $(\varOmega,\mathcal{F})$
such that $\mathscr{P}(\varOmega)=1$,
\[
(\varOmega,\mathcal{F},\mathscr{P})
\]
is a probability space. Especially, we use this for 4D stochastic
processes, $\left(\mathit{\Omega}^{1\mathchar`-\mathrm{dim}},\mathcal{F}^{1\mathchar`-\mathrm{dim}},\mathscr{P}^{1\mathchar`-\mathrm{dim}}\right)$
for 1D stochastic processes.\end{leftbar}
\end{defn}

Then, a 4D continuous stochastic process 
\begin{align*}
\hat{x}(\circ,\bullet) & \coloneqq\{\hat{x}(\tau,\omega)\}_{(\tau,\omega)\in\mathbb{R}\times\varOmega}\\
 & =\{\hat{x}(\tau,\omega)\in\mathbb{A}^{4}(\mathbb{V_{\mathrm{M}}^{\mathrm{4}}},g)|\tau\in\mathbb{R},\omega\in\varOmega\}
\end{align*}
is regarded as a $\mathscr{B}(\mathbb{R})\times\mathcal{F}/\mathscr{B}(\mathbb{A}^{4}(\mathbb{V_{\mathrm{M}}^{\mathrm{4}}},g))$-measurable
mapping. Where for two measurable spaces $(X,\mathcal{X})$ and $(Y,\mathcal{Y})$,
an $\mathcal{X}/\mathcal{Y}$-measurable mapping $f$ is a mapping
$f:X\rightarrow Y$ such that $f^{-1}(A)\coloneqq\{x\in X|f(x)\in A\}\subset\mathcal{X}$
for all $A\in\mathcal{Y}$.

\subsubsection{Increasing family }

Let us consider an usual 1D Wiener process
\[
w_{+}(\circ,\bullet)\coloneqq\{w_{+}(t,\omega)\in\mathbb{R}|t\in[0,\infty),\omega\in\mathit{\Omega}^{1\mathchar`-\mathrm{dim}}\}
\]
on $\left(\mathit{\Omega}^{1\mathchar`-\mathrm{dim}},\mathcal{F}^{1\mathchar`-\mathrm{dim}},\mathscr{P}^{1\mathchar`-\mathrm{dim}}\right)$:
\begin{defn}[Wiener process $w_{+}(\circ,\bullet)$]
\begin{leftbar}When a 1D stochastic process $w_{+}(\circ,\bullet)$
satisfies below, it is a 1D {\bf Wiener process} ({\bf WP}) or a
1D {\bf Brownian motion}. 

(1) $w_{+}(0,\omega)=0$ a.s., 

(2) $t\mapsto w_{+}(t,\omega)$ is continuous,

(3) For all times $0=t_{0}<t_{1}<\cdots<t_{t_{n}}$ ($n\in\mathbb{Z}$),
the increments $\{w_{+}(t_{i},\omega)-w_{+}(t_{i-1},\omega)\}_{i=1}^{n}$
are independent and each of them follows the normal distribution $\{N(0,t_{i}-t_{i-1})\}_{i=1}^{n}$.\end{leftbar}
\end{defn}

Instead of the above rule (1), we can choose $w_{+}(0,\omega)=x$
a.s., too. For $\mathbb{E}\llbracket f(w_{+}(t,\bullet))\rrbracket\coloneqq\int_{\omega\in\mathit{\Omega}^{1\mathchar`-\mathrm{dim}}}\allowbreak f(w_{+}(t,\omega))\allowbreak d\mathscr{P}^{1\mathchar`-\mathrm{dim}}(\omega)$,
the following basic result is found:
\begin{lem}
\begin{leftbar}For $\left(\mathit{\Omega}^{1\mathchar`-\mathrm{dim}},\mathcal{F}^{1\mathchar`-\mathrm{dim}},\mathscr{P}^{1\mathchar`-\mathrm{dim}}\right)$,
a 1D WP $w_{+}(\circ,\bullet)$ satisfies below for all $t$:
\begin{equation}
\mathbb{E}\llbracket w_{+}(t+\delta t,\bullet)-w_{+}(t,\bullet)\rrbracket=0
\end{equation}
\begin{equation}
\underset{\delta t\rightarrow0+}{\lim}\mathbb{E}\left\llbracket \frac{[w_{+}(t+\delta t,\bullet)-w_{+}(t,\bullet)]^{2}}{\delta t}\right\rrbracket =1
\end{equation}
\end{leftbar}
\end{lem}

For the discussion of stochastic processes, the definition of an increasing
family (a filtration) is important.
\begin{defn}[Increasing family $\{\mathcal{P}_{t}\}$]
\begin{leftbar}For a probability space $\left(\mathit{\Omega}^{1\mathchar`-\mathrm{dim}},\mathcal{F}^{1\mathchar`-\mathrm{dim}},\mathscr{P}^{1\mathchar`-\mathrm{dim}}\right)$,
$\{\mathcal{P}_{t}\}_{t\in\mathbb{R}}$ is an increasing family of
sub-$\sigma$-algebras on $\varOmega$ such that $-\infty<s\leq t\Longrightarrow\mathcal{P}_{s}\subset\mathcal{P}_{t}\subset\mathcal{F}^{1\mathchar`-\mathrm{dim}}$.
\end{leftbar}
\end{defn}

$\mathcal{P}_{s}\subset\mathcal{P}_{t\geq s}$ denotes an increment
of branches of sample paths. This is the characteristics of a forward
(normal) diffusion process. 
\begin{defn}[$\{\mathcal{P}_{t}\}$-adapted]
\begin{leftbar}When a stochastic process is $\mathcal{P}_{t}/\mathscr{B}(X)$-measurable
for each $t$ ($X=\mathbb{R}$ in the present case), we call it a
$\{\mathcal{P}_{t}\}$-adapted.\end{leftbar}
\end{defn}

\begin{defn}[$\{\mathcal{P}_{t}\}$-WP]
\begin{leftbar}\label{P-Wiener}When \textbf{$w_{+}(\circ,\bullet)$}
satisfies the following, it is a $\{\mathcal{P}_{t}\}$-WP.

(1) $w_{+}(\circ,\bullet)$ is $\{\mathcal{P}_{t}\}$-adapted.

(2) $w_{+}(t,\omega)-w_{+}(s,\omega)$ and $\mathcal{P}_{s}$ are
independent for $0\leq\forall s\le t$.

\end{leftbar}
\end{defn}

For $(\varOmega,\mathcal{F},\mathscr{P})$ and $T\subset\mathbb{R}$,
consider $\mathbb{L}_{T}^{p}(X)$ a family of $\mathscr{B}(T)\times\mathcal{F}/\mathscr{B}(X)$-measurable
mappings ($X$ is an $N$-dimensional topological space) such that
\begin{align*}
\mathbb{L}_{T}^{p}(X) & \coloneqq\left\{ \hat{f}(\circ,\bullet):T\times\varOmega\rightarrow X\left|\begin{gathered}t\rightarrow\hat{f}(t,\omega)\mathrm{\,is\,continuous},\,\sum_{i=1}^{N}\int_{T}|\varphi^{i}\circ\hat{f}(t,\omega)|^{p}dt<\infty\,\mathrm{a.s.}\end{gathered}
\right.\right\} 
\end{align*}
with $\{\varphi^{i}\}_{i=1}^{N}$ the coordinate mapping $\varphi^{i}:X\rightarrow\mathbb{R}$.
Then, its ``adapted'' class is 
\begin{align*}
\mathcal{L}_{\mathrm{loc}}^{p}(\{\mathcal{P}_{t}\};X) & \coloneqq\left\{ \hat{f}(\circ,\bullet)\in\mathbb{L}_{[t_{1},t_{2}]}^{p}(X)\left|\begin{gathered}\forall t_{1}\leq\forall t_{2}\in\mathbb{R},\,\hat{f}(\circ,\bullet)\mathrm{\,is\,}\{\mathcal{P}_{t}\}\mathchar`-\mathrm{adapted.}\end{gathered}
\right.\right\} 
\end{align*}
\begin{defn}[$\{\mathcal{P}_{t}\}$-prog.]
\begin{leftbar}If a stochastic process $\hat{f}(\circ,\bullet)$
is $\mathscr{B}([s,t])\times\mathcal{P}_{t}/\mathscr{B}(X)$-measurable
for each $t\in\mathbb{R}$ and $s\leq t$, $\hat{f}(\circ,\bullet)$
is called $\{\mathcal{P}_{t}\}$-progressively measurable, $\{\mathcal{P}_{t}\}$-progressive
or {\bf $\boldsymbol{\{\mathcal{P}_{t}\}}$-prog.} \end{leftbar}
\end{defn}

\begin{thm}
\begin{leftbar}A stochastic process $\hat{f}(\circ,\bullet)$ is
$\{\mathcal{P}_{t}\}$-prog.\ when $\hat{f}(\circ,\bullet)$ is continuous
and $\{\mathcal{P}_{t}\}$-adapted. Its converse is also satisfied.
\end{leftbar}
\end{thm}

\begin{proof}
Consider the domain of $[s,t]$ for each $s\leq t$. If $\hat{f}(\circ,\bullet):R\times\varOmega\rightarrow\mathbb{R}$
is continuous and $\{\mathcal{P}_{t}\}$-adapted, it is described
by a simple function such as
\begin{align*}
\hat{f}(\tau,\omega) & =\xi_{0}(\omega)\times\boldsymbol{1}_{\{t_{0}\}}(\tau)+\sum_{j=1}^{n}\xi_{j}(\omega)\times\boldsymbol{1}_{(t_{j-1},t_{j}]}(\tau)
\end{align*}
for $s=t_{0}\leq t_{1}\leq\cdots\leq t_{n}=t$ and each of $\{\mathcal{P}_{t_{m=0,1,\cdots,n}}\}$-adapted
$\xi_{m}(\bullet)$,
\begin{align*}
\{(\tau,\omega)\in[s,t]\times\varOmega|f(\tau,\omega)\leq\forall a\}\\
=\{t_{0}\}\times\{\omega\in\varOmega|\xi_{0}(\omega)\leq a\} & \cup\bigcup_{j=1}^{n}(t_{j-1},t_{j}]\times\{\omega\in\varOmega|\xi_{j}(\omega)\leq a\}
\end{align*}
with $n\rightarrow\infty$. Where, 
\[
\boldsymbol{1}_{A}(\tau)=\begin{cases}
1 & \tau\in A\\
0 & \tau\notin A
\end{cases}
\]
with $\forall A\subset\mathbb{R}$. Thus, $\hat{f}(\circ,\bullet)$
is $\mathscr{B}([s,t])\times\mathcal{P}_{t}/\mathscr{B}(X)$-measurable,
namely, it is a $\{\mathcal{P}_{t}\}$-prog. When $f(\circ,\bullet)$
is $\{\mathcal{P}_{t}\}$-prog., it is $\{\mathcal{P}_{t}\}$-adapted
since $f(t,\bullet)$ is $\mathcal{P}_{t}/\mathscr{B}(X)$-measurable.
\end{proof}
\begin{defn}[$\{\mathcal{P}_{t}\}$-martingale]
\begin{leftbar}A stochastic process $m(\circ,\bullet)$ is a $\{\mathcal{P}_{t}\}$-martingale
when $m(\circ,\bullet)$ satisfies below:

(1) $m(t,\bullet)$ is integrable, i.e., $\mathbb{E}\llbracket m(t,\bullet)\rrbracket<\infty$,

(2) $m(\circ,\bullet)$ is $\{\mathcal{P}_{t}\}$-adapted,

(3) For $-\infty<\forall s\leq t$, $\mathbb{E}\llbracket m(t,\bullet)|\mathcal{P}_{s}\rrbracket=\hat{m}(s,\bullet)$
a.s.\end{leftbar}
\end{defn}

Therefore, a $\{\mathcal{P}_{t}\}$-WP $w_{+}(\circ,\bullet)$ is
a $\{\mathcal{P}_{t}\}$-martingale. Then, $\int_{s}^{t}f(t',\omega)dw_{+}(t',\omega)$
is defined by means of an It\^{o} integral of $f\in\mathcal{L}_{\mathrm{loc}}^{2}(\{\mathcal{P}_{t}\};\mathbb{R})$
\cite{Ito(1944)}. Since a $\{\mathcal{P}_{t}\}$-prog.\ is expressed
by a simple function 
\begin{align*}
\hat{f}(t,\omega) & =\xi_{0}(\omega)\times\boldsymbol{1}_{\{t_{0}\}}(t)+\sum_{j=1}^{n}\xi_{j}(\omega)\times\boldsymbol{1}_{(t_{j-1},t_{j}]}(t)
\end{align*}
for $s=t_{0}\leq t_{1}\leq\cdots\leq t_{n}=T$, let us define
\begin{align*}
S_{n}^{+}(s,t,\omega) & \coloneqq\sum_{j=0}^{n-1}\xi_{j}(\omega)\times\left[w_{+}(t\wedge t_{j+1},\omega)-w_{+}(t\wedge t_{j},\omega)\right]
\end{align*}
with $s\land t\coloneqq\max\{s,t\}$. If there is $S^{+}(s,t,\omega)$
such that
\begin{align*}
\lim_{n\rightarrow\infty}\mathbb{E}\left\llbracket [S_{n}^{+}(s,t,\bullet)-S^{+}(s,t,\bullet)]^{2}\right\rrbracket  & =0,
\end{align*}
this $S(s,t,\omega)$ is regarded as $\int_{s}^{t}f(t',\omega)dw_{+}(t',\omega)$.
Therefore, $\int_{s}^{t}f(t',\omega)dw_{+}(t',\omega)$ is a $\{\mathcal{P}_{t}\}$-martingale.
Then, the following is found easily.
\begin{lem}
\begin{leftbar}For  $f(\circ,\bullet),g(\circ,\bullet)\in\mathcal{L}_{\mathrm{loc}}^{2}(\{\mathcal{P}_{t};\mathbb{R}\})$,
the following relation is imposed \cite{Ito(1944)}:
\begin{align}
\mathbb{E}\left\llbracket \int_{0}^{\tau}f(t',\bullet)dw_{+}(t',\bullet)\cdot\int_{0}^{\tau}g(t'',\bullet)dw_{+}(t'',\bullet)\right\rrbracket  & =\mathbb{E}\left\llbracket \int_{0}^{\tau}f(t',\bullet)\cdot g(t',\bullet)dt'\right\rrbracket 
\end{align}
 \end{leftbar}

Finally, the following well-known theorem of the so-called It\^{o}
formula is found:
\end{lem}

\begin{thm}[$\{\mathcal{P}_{t}\}$-It\^{o} formula]
\begin{leftbar}For $a(\circ,\bullet)\in\mathcal{L}_{\mathrm{loc}}^{1}(\{\mathcal{P}_{t}\};\mathbb{R})$
and $b(\circ,\bullet)\in\mathcal{L}_{\mathrm{loc}}^{2}(\{\mathcal{P}_{t}\};\mathbb{R})$,
consider a 1D stochastic process $\hat{X}(\circ,\omega)$ of a $\{\mathcal{P}_{t}\}$-prog.\ given
by
\begin{align}
\hat{X}(t,\omega) & =\hat{X}(0,\omega)+\int_{0}^{t}a(t',\omega)dt'+\int_{0}^{t}b(t',\omega)dw_{+}(t',\omega)
\end{align}
w.r.t.\ its initial condition $\hat{X}(0,\omega)$. Then for a function
$f\in C^{2}(\mathbb{R}^{N})$, the It\^{o} formula
\begin{align}
f(\hat{X}(t,\omega)) & =f(\hat{X}(0,\omega))+\int_{0}^{t}f'(\hat{X}(t',\omega))a(t',\omega)dt'\nonumber \\
 & \quad\quad+\int_{0}^{t}f'(\hat{X}(t',\omega))b(t',\omega)dw_{+}(t',\omega)+\frac{1}{2}\int_{0}^{t}f''(\hat{X}(t',\omega))[b(t',\omega)]^{2}dt'\label{eq: Ito-formula-1}
\end{align}
 is almost surely satisfied \cite{Ito(1944)}. By introducing
\begin{align*}
\int_{0}^{t}d_{+}\hat{X}(t',\omega) & \coloneqq\hat{X}(t,\omega)-\hat{X}(0,\omega),\\
\int_{0}^{t}d_{+}f(\hat{X}(t',\omega)) & \coloneqq f(\hat{X}(t,\omega))-f(\hat{X}(0,\omega)),
\end{align*}
Eq.(\ref{eq: Ito-formula-1}) is also expressed as
\begin{align}
d_{+}f(\hat{X}(t,\omega)) & =f'(\hat{X}(t',\omega))d_{+}\hat{X}(t,\omega)+\frac{1}{2}f''(\hat{X}(t',\omega))[d_{+}\hat{X}(t,\omega)]^{2}\,\,\,\,\mathrm{a.s.}\label{eq: Ito-formula-2}
\end{align}
\end{leftbar}
\end{thm}

\subsection{Decreasing family}

We also introduced a backward diffusion. This is called as a decreasing
family.
\begin{defn}[Decreasing family $\{\mathcal{F}_{t}\}$]
\begin{leftbar}For a probability space $\left(\mathit{\Omega}^{1\mathchar`-\mathrm{dim}},\mathcal{F}^{1\mathchar`-\mathrm{dim}},\mathscr{P}^{1\mathchar`-\mathrm{dim}}\right)$,
$\{\mathcal{F}_{t}\}_{t\in\mathbb{R}}$ is a decreasing family of
sub-$\sigma$-algebras on $\varOmega$ such that $t\leq s<\infty\Longrightarrow\mathcal{F}^{1\mathchar`-\mathrm{dim}}\supset\mathcal{F}_{t}\supset\mathcal{F}_{s}$.
\end{leftbar}
\end{defn}

This is regarded as an inverse process of a $\{\mathcal{P}_{t}\}$-prog.
Let us suggest its WP.
\begin{defn}[$\{\mathcal{F}_{t}\}$-WP]
\begin{leftbar}Let us define a monotonically decreasing function
$f_{\downarrow}:\mathbb{R}\rightarrow\mathbb{R}$ such that $df_{\downarrow}/dt=-1$.
$w_{-}(\circ,\bullet)$ is an $\{\mathcal{F}_{t}\}$-WP when $\{w_{-}(f_{\downarrow}(t),\omega)\}_{t\in\mathbb{R}}$
is a $\{\mathcal{P}_{t}\}$-WP.\end{leftbar}
\end{defn}

\
\begin{defn}[$\{\mathcal{F}_{t}\}$-martingale]
\begin{leftbar}A stochastic process $m(\circ,\bullet)$ is a $\{\mathcal{F}_{t}\}$-martingale
when $m(\circ,\bullet)$ satisfies below:

(1) $m(t,\bullet)$ is integrable, i.e., $\mathbb{E}\llbracket m(t,\bullet)\rrbracket<\infty$,

(2) $m(\circ,\bullet)$ is $\{\mathcal{F}_{t}\}$-adapted,

(3) For $t\leq\forall s<\infty$, $\mathbb{E}\llbracket m(t,\bullet)|\mathcal{F}_{s}\rrbracket=\hat{m}(s,\bullet)$
a.s.\end{leftbar}
\end{defn}

In order to the above definition, an $\{\mathcal{F}_{t}\}$-WP is
an $\{\mathcal{F}_{t}\}$-martingale. That is confirmed by the following
lemma.
\begin{lem}
\begin{leftbar}$w_{-}(\circ,\bullet)$ of an $\{\mathcal{F}_{t}\}$-WP
satisfies below:
\begin{equation}
\mathbb{E}\llbracket w_{-}(t,\bullet)-w_{-}(t-\delta t,\bullet)\rrbracket=0
\end{equation}
\begin{equation}
\underset{\delta t\rightarrow0+}{\lim}\mathbb{E}\left\llbracket \frac{[w_{-}(t,\bullet)-w_{-}(t-\delta t,\bullet)]^{2}}{\delta t}\right\rrbracket =1
\end{equation}
\end{leftbar}
\end{lem}

Then, the following function family is introduced:
\begin{align*}
\mathcal{L}_{\mathrm{loc}}^{p}(\{\mathcal{F}_{t}\};X) & \coloneqq\left\{ \hat{f}(\circ,\bullet)\in\mathbb{L}_{[t_{1},t_{2}]}^{p}(X)\left|\begin{gathered}\forall t_{1}\leq\forall t_{2}\in\mathbb{R},\,\hat{f}(\circ,\bullet)\mathrm{\,is\,}\{\mathcal{F}_{t}\}\mathchar`-\mathrm{adapted.}\end{gathered}
\right.\right\} 
\end{align*}
\begin{defn}[$\{\mathcal{F}_{t}\}$-prog\textcolor{black}{.}]
\begin{leftbar}If a stochastic process $\hat{f}(\circ,\bullet)$
is $\mathscr{B}([s,t])\times\mathcal{F}_{t}/\mathscr{B}(X)$-measurable
for all $s\leq t\in\mathbb{R}$, let us call that $\hat{f}(\circ,\bullet)$
is $\{\mathcal{F}_{t}\}$-prog.\end{leftbar}
\end{defn}

For an $\{\mathcal{F}_{t}\}$-prog.\ $\hat{f}(\circ,\bullet)$, $s=t_{0}\leq t_{1}\leq\cdots\leq t_{n}=T$
and each of $\{\mathcal{F}_{t_{m=0,1,\cdots,n}}\}$-adapted $\xi_{m}(\bullet)$,
\begin{align*}
\hat{f}(\tau,\omega) & =\xi_{0}(\omega)\times\boldsymbol{1}_{\{t_{0}\}}(\tau)+\sum_{j=1}^{n}\xi_{j}(\omega)\times\boldsymbol{1}_{(t_{j-1},t_{j}]}(\tau).
\end{align*}
Let us introduce the summation below,
\begin{align*}
S_{n}^{-}(s,t,\omega) & \coloneqq\sum_{j=0}^{n-1}\xi_{j+1}(\omega)\times\left[w_{-}(t\wedge t_{j+1},\omega)-w_{-}(t\wedge t_{j},\omega)\right].
\end{align*}
When there is $S^{-}(s,t,\omega)$ such that
\begin{align*}
\lim_{n\rightarrow\infty}\mathbb{E}\left\llbracket [S_{n}^{-}(s,t,\bullet)-S^{-}(s,t,\bullet)]^{2}\right\rrbracket  & =0,
\end{align*}
this $S^{-}(s,t,\omega)$ is expressed by $\int_{s}^{t}f(t',\omega)dw_{-}(t',\omega)$. 
\begin{thm}[$\{\mathcal{F}_{t}\}$-It\^{o} formula]
\begin{leftbar}For $A(\circ,\bullet)\in\mathcal{L}_{\mathrm{loc}}^{1}(\{\mathcal{F}_{t}\};\mathbb{R})$
and $B(\circ,\bullet)\in\mathcal{L}_{\mathrm{loc}}^{2}(\{\mathcal{F}_{t}\};\mathbb{R})$,
let $\hat{X}(\circ,\omega)$ of an $\{\mathcal{F}_{t}\}$-prog.\ be
given by
\begin{align}
d_{-}\hat{X}(t,\omega) & =A(t,\omega)dt+B(t,\omega)dw_{-}(t,\omega)
\end{align}
with
\begin{align}
\int_{a}^{b}d_{-}\hat{X}(t,\omega) & \coloneqq\hat{X}(b,\omega)-\hat{X}(a,\omega).
\end{align}
Then by $\int_{a}^{b}d_{-}f(\hat{X}(t',\omega))\coloneqq f(\hat{X}(b,\omega))-f(\hat{X}(a,\omega))$
for $f\in C^{2}(\mathbb{R}^{N})$, its It\^{o} formula becomes
\begin{align}
d_{-}f(\hat{X}(t,\omega)) & =f'(\hat{X}(t,\omega))d_{-}\hat{X}(t,\omega)-\frac{1}{2}f''(\hat{X}(t,\omega))[d_{-}\hat{X}(t,\omega)]^{2}\,\,\,\,\mathrm{a.s.}\label{eq: Ito-formula-3}
\end{align}
 \end{leftbar}
\end{thm}

Consider the derivation of (\ref{eq: Ito-formula-3}) by using (\ref{eq: Ito-formula-2}).
The decreasing family $\{\mathcal{F}_{t}\}_{t\in\mathbb{R}}$ relates
to an increasing family $\{\mathcal{P}_{t'}\}_{t'\in\mathbb{R}}$
by a monotonically decreasing function $f_{\downarrow}:\mathbb{R}\rightarrow\mathbb{R}$
such as $df_{\downarrow}/dt=-1$, $\{\mathcal{P}_{f_{\downarrow}(t)}\}_{t\in\mathbb{R}}$
becomes an decreasing family since $t\le s\Rightarrow\mathcal{P}_{f_{\downarrow}(t)}\supset\mathcal{P}_{f_{\downarrow}(s)}$.
Thus, there is $f_{\downarrow}$ such that $\{\mathcal{P}_{f_{\downarrow}(t')}\}_{t'\in\mathbb{R}}\equiv\{\mathcal{F}_{t}\}_{t\in\mathbb{R}}$.
For $\hat{X}(\circ,\bullet)$ of an $\{\mathcal{F}_{t}\}$-prog.,
there is $\hat{X}'(\circ,\bullet)$ of a $\{\mathcal{P}_{t'}\}$-prog.\ satisfying
$\hat{X}'(t',\omega')=\hat{X}(t,\omega)$ with $t'\coloneqq f_{\downarrow}(t)=t$
at a fixed $t$. For $\forall a\geq0$, let us set $f_{\downarrow}(t-a)=t+a$.
Since $\int_{f_{\downarrow}(t+a)}^{f_{\downarrow}(t-a)}d_{+}\hat{X}'(\mathring{t},\omega')=\hat{X}'(f_{\downarrow}(t-a),\omega')-\hat{X}'(f_{\downarrow}(t+a),\omega')=\hat{X}(t-a,\omega)-\hat{X}(t+a,\omega)=-\int_{t-a}^{t+a}d_{-}\hat{X}(\mathring{t},\omega)$,
$d_{+}\hat{X}'(t',\omega')=-d_{-}\hat{X}(t,\omega)$ is derived. $d_{+}f(\hat{X}'(t',\omega'))=-d_{-}f(\hat{X}(t,\omega))$
is also imposed for $\forall f\in C^{2}(\mathbb{R}^{N})$. Let us
apply those relations to (\ref{eq: Ito-formula-2}), namely,
\begin{align*}
d_{+}f(\hat{X}'(t',\omega')) & =f'(\hat{X}'(t',\omega'))d_{+}\hat{X}'(t',\omega')+\frac{1}{2}f''(\hat{X}'(t',\omega'))[d_{+}\hat{X'}(t',\omega')]^{2}\,\,\,\,\mathrm{a.s.}
\end{align*}
 Finally, Eq.(\ref{eq: Ito-formula-3}) is found by the replacement
from $\hat{X}'(\circ,\bullet)$ to $\hat{X}(\circ,\bullet)$. Let
$\omega'=\omega$ when $\{\hat{X}'(t',\omega')\}_{t'\in\mathbb{R}}=\{\hat{X}(t,\omega)\}_{t\in\mathbb{R}}$.
In this case, $\hat{X}(\circ,\bullet)$ is $\{\mathcal{P}_{t}\}$
and $\{\mathcal{F}_{t}\}$-prog.

\subsection{Forward-backward composition}

Let us introduce the simplest example of a composition of stochastic
processes for the relativistic kinematics. Consider the set of 1D
stochastic processes, $\hat{X}(\circ,\bullet)$ and $\hat{Y}(\circ,\bullet)$
of a $\{\mathcal{P}_{t}\}$-prog.\ and an $\{\mathcal{F}_{t}\}$-prog.,
respectively:

\begin{equation}
\begin{cases}
\begin{gathered}d_{+}\hat{X}(t,\omega)=a_{+}(\hat{X}(t,\omega))dt+\theta\times dw_{+}(t,\omega)\end{gathered}
\\
\begin{gathered}d_{-}\hat{Y}(t,\omega')=a_{-}(\hat{Y}(t,\omega'))dt+\theta\times dw_{-}(t,\omega')\end{gathered}
\end{cases}\label{eq: set of SDE eqs 1}
\end{equation}
By a 2D vector $\hat{\gamma}(t,\omega)\coloneqq(\hat{X}(t,\omega),\hat{Y}(t,\omega))$,
Eq.(\ref{eq: set of SDE eqs 1}) becomes
\begin{align}
\left[\begin{array}{c}
\hat{X}(b,\omega)-\hat{X}(a,\omega)\\
\hat{Y}(b,\omega)-\hat{Y}(a,\omega)
\end{array}\right] & =\int_{a}^{b}\left[\begin{array}{c}
a_{+}(\hat{X}(t,\omega),\hat{Y}(t,\omega))\\
a_{-}(\hat{X}(t,\omega),\hat{Y}(t,\omega))
\end{array}\right]dt+\theta\times\int_{a}^{b}\left[\begin{array}{c}
dw_{+}(t,\omega)\\
dw_{-}(t,\omega)
\end{array}\right]\label{eq: set of SDE eqs 2}
\end{align}
with $a_{+},a_{-}:\mathbb{R}^{2}\rightarrow\mathbb{R}$ . This is
forward-diffused in $\hat{X}$-direction and backward-diffused in
$\hat{Y}$-direction. Consider $(\mathcal{P}_{t},\mathcal{F}_{t})$
the set of the sub-$\sigma$-algebras for each $t$. This is regarded
as a new sub-$\sigma$-algebra of a 2D stochastic process denoted
by $\mathcal{M}_{t}\coloneqq\mathcal{P}_{t}\otimes\mathcal{F}_{t}$
and its family $\{\mathcal{M}_{t}\}_{t\in\mathbb{R}}$. 
\begin{defn}[Forward-backward composition]
\begin{leftbar} Let $\{\mathcal{M}_{t}\}_{t\in\mathbb{R}}$ be a
family of a sub-$\sigma$-algebras for a 2D stochastic process $\hat{\gamma}(\circ,\omega)\coloneqq\{(\hat{X}(t,\omega),\hat{Y}(t,\omega))|t\in\mathbb{R}\}$
on $\left(\mathit{\Omega},\mathcal{F},\mathscr{P}\right)$. Where,
an $\mathcal{M}_{t}/\mathscr{B}(\mathbb{R}^{2})$-measurable $\hat{\gamma}(t,\bullet)$
is $\mathcal{P}_{t}\otimes\mathcal{F}_{t}/\mathscr{B}(\mathbb{R}^{2})$-measurable
when $\hat{X}(t,\bullet)$ is $\mathcal{P}_{t}/\mathscr{B}(\mathbb{R})$-measurable
and $\hat{Y}(t,\bullet)$ is $\mathcal{F}_{t}/\mathscr{B}(\mathbb{R})$-measurable.
Hence, $\mathcal{M}_{t}\coloneqq\mathcal{P}_{t}\otimes\mathcal{F}_{t}$.
When $\{\hat{\gamma}(t,\bullet)\}_{t\in\mathbb{R}}$ is $\mathcal{P}_{t}\otimes\mathcal{F}_{t}/\mathscr{B}(\mathbb{R}^{2})$-measurable
for all $t$, $\hat{\gamma}(\circ,\bullet)$ is $\{\mathcal{P}_{t}\otimes\mathcal{F}_{t}\}$-adapted.
If $\hat{\gamma}(\circ,\bullet)$ is $\mathscr{B}([s,t])\times\mathcal{P}_{t}\otimes\mathcal{F}_{t}/\mathscr{B}(\mathbb{R}^{2})$-measurable
for all $s\leq t\in\mathbb{R}$, it is $\{\mathcal{P}_{t}\otimes\mathcal{F}_{t}\}$-prog.\end{leftbar}
\end{defn}

For $a(\hat{\gamma}(\circ,\bullet))\coloneqq(a_{+}(\hat{\gamma}(\circ,\bullet),a_{-}(\hat{\gamma}(\circ,\bullet)))$
of a $\{\mathcal{P}_{t}\otimes\mathcal{F}_{t}\}$-prog., let us confirm
the following It\^{o} formula. 
\begin{multline}
f(\hat{\gamma}(b,\omega))-f(\hat{\gamma}(a,\omega))\\
=\int_{a}^{b}\frac{\partial f}{\partial x}(\hat{\gamma}(t,\omega))d_{+}\hat{X}(\hat{\gamma}(t,\omega))+\int_{a}^{b}\frac{\partial f}{\partial y}(\hat{\gamma}(t,\omega))d_{-}\hat{Y}(\hat{\gamma}(t,\omega))\\
+\frac{\theta^{2}}{2}\times\int_{a}^{b}\frac{\partial^{2}f}{\partial x\partial x}(\hat{\gamma}(t,\omega))dt-\frac{\theta^{2}}{2}\times\int_{a}^{b}\frac{\partial^{2}f}{\partial y\partial y}(\hat{\gamma}(t,\omega))dt\,\,\,\,\mathrm{a.s.}
\end{multline}
The appearance of $\int_{a}^{b}(\partial_{x}^{2}-\partial_{y}^{2})f(\hat{\gamma}(t,\omega))dt$
is useful to make the relativistic kinematics. The readers can understand
the above formula by the following:
\begin{multline*}
f(\hat{X}(t+\delta t,\omega),\hat{Y}(t+\delta t,\omega))-f(\hat{X}(t,\omega),\hat{Y}(t,\omega))\\
=\left[f(\hat{X}(t+\delta t,\omega),\hat{Y}(t+\delta t,\omega))-f(\hat{X}(t,\omega),\hat{Y}(t+\delta t,\omega))\right]_{\hat{Y}(t+\delta t,\omega)\mathrm{\,is\,fixed}}\\
+\left[f(\hat{X}(t,\omega),\hat{Y}(t+\delta t,\omega)-f(\hat{X}(t,\omega),\hat{Y}(t,\omega))\right]_{\hat{X}(\tau,\omega)\mathrm{\,is\,fixed}}
\end{multline*}

\subsection{Basic construction of 4D processes}

Consider the following types of 4D processes on $(\mathbb{A}^{4}(\mathbb{V_{\mathrm{M}}^{\mathrm{4}}},g),\mathscr{B}(\mathbb{A}^{4}(\mathbb{V_{\mathrm{M}}^{\mathrm{4}}},g)),\mu)$
of the Minkowski spacetime. 
\[
\underset{\{\mathscr{P}_{\tau}\}\mathchar`-\mathrm{prog.}}{\underbrace{\varphi\circ\hat{x}(\tau,\omega)}}\coloneqq(\underset{\{\mathcal{F}_{\tau}\}\mathchar`-\mathrm{prog.}}{\underbrace{\hat{x}^{0}(\tau,\omega)},}\underset{\{\mathcal{P}_{\tau}\}\mathchar`-\mathrm{prog.}}{\underbrace{\hat{x}^{1}(\tau,\omega),\hat{x}^{2}(\tau,\omega),\hat{x}^{3}(\tau,\omega)}})
\]

\[
\underset{\{\mathscr{F}_{\tau}\}\mathchar`-\mathrm{prog.}}{\underbrace{\varphi\circ\hat{x}(\tau,\omega)}}\coloneqq(\underset{\{\mathcal{P}_{\tau}\}\mathchar`-\mathrm{prog.}}{\underbrace{\hat{x}^{0}(\tau,\omega)},}\underset{\{\mathcal{F}_{\tau}\}\mathchar`-\mathrm{prog.}}{\underbrace{\hat{x}^{1}(\tau,\omega),\hat{x}^{2}(\tau,\omega),\hat{x}^{3}(\tau,\omega)}})
\]
\begin{defn}[$\{\mathscr{P}_{\tau}\}$ and $\{\mathscr{F}_{\tau}\}$]
\begin{leftbar}For a probability space $\left(\mathit{\Omega}^{1\mathchar`-\mathrm{dim}},\mathcal{F}^{1\mathchar`-\mathrm{dim}},\mathscr{P}^{1\mathchar`-\mathrm{dim}}\right)$,
let $\{\mathcal{P}_{\tau}\}_{\tau\in\mathbb{R}}$ and $\{\mathcal{F}_{\tau}\}_{\tau\in\mathbb{R}}$
be increasing and decreasing families of 1D continuous processes,
respectively. Then for $\varOmega\coloneqq\times^{4}\mathit{\Omega}^{1\mathchar`-\mathrm{dim}}$,
$\mathcal{F}\coloneqq\otimes^{4}\mathcal{F}^{1\mathchar`-\mathrm{dim}}$
and $\mathscr{P}\coloneqq\times^{4}\mathscr{P}^{1\mathchar`-\mathrm{dim}}$,
there are 4D stochastic processes of $\{\mathscr{P}_{\tau}\}_{\tau\in\mathbb{R}}$
and $\{\mathscr{F}_{\tau}\}_{\tau\in\mathbb{R}}$ on $(\varOmega,\mathcal{F},\mathscr{P})$;
$\{\mathscr{P}_{\tau}\}$ is a family of $\mathscr{P}_{\tau\in\mathbb{R}}\coloneqq\mathcal{F}_{\tau}\otimes\mathcal{P}_{\tau}\otimes\mathcal{P}_{\tau}\otimes\mathcal{P}_{\tau}$
and $\{\mathscr{F}_{\tau}\}$ of $\mathscr{F}_{\tau\in\mathbb{R}}\coloneqq\mathcal{P}_{\tau}\otimes\mathcal{F}_{\tau}\otimes\mathcal{F}_{\tau}\otimes\mathcal{F}_{\tau}$.\end{leftbar}
\end{defn}

The WPs are updated on $(\varOmega,\mathcal{F},\mathscr{P})$:
\begin{defn}[$\{\mathscr{P}_{\tau}\}$-WPs and $\{\mathscr{F}_{\tau}\}$-WPs]
\begin{leftbar}Let $W_{+}(\circ,\bullet)$ and $W_{-}(\circ,\bullet)$
be $\{\mathscr{P}_{\tau}\}$ and $\{\mathscr{F}_{\tau}\}$-WPs as
the 4D stochastic processes:
\[
\underset{\{\mathscr{P}_{\tau}\}\mathchar`-\mathrm{WP}}{\underbrace{\varphi\circ W_{-}(\tau,\omega)}}\coloneqq(\underset{\{\mathcal{F}_{\tau}\}\mathchar`-\mathrm{WP}}{\underbrace{w_{-}(\tau,\omega)},}\underset{\{\mathcal{P}_{\tau}\}\mathchar`-\mathrm{WPs}}{\underbrace{w_{+}^{1}(\tau,\omega),w_{+}^{2}(\tau,\omega),w_{+}^{3}(\tau,\omega)}})
\]

\[
\underset{\{\mathscr{F}_{\tau}\}\mathchar`-\mathrm{WP}}{\underbrace{\varphi\circ W_{-}(\tau,\omega)}}\coloneqq(\underset{\{\mathcal{P}_{\tau}\}\mathchar`-\mathrm{WP}}{\underbrace{w_{+}(\tau,\omega)},}\underset{\{\mathcal{F}_{\tau}\}\mathchar`-\mathrm{WPs}}{\underbrace{w_{-}^{1}(\tau,\omega),w_{-}^{2}(\tau,\omega),w_{-}^{3}(\tau,\omega)}}).
\]
\end{leftbar}
\end{defn}

\begin{lem}
\begin{leftbar}$W_{+}(\circ,\bullet)$ and $W_{-}(\circ,\bullet)$
satisfy below with $W_{\pm}^{\mu}(\tau,\omega)\coloneqq\varphi\circ W_{\pm}(\tau,\omega)$
for each $\mu,\nu=0,1,2,3$:
\begin{align}
\mathbb{E}\left\llbracket \int_{\tau}^{\tau+\delta\tau}dW_{\pm}^{\mu}(\tau',\bullet)\right\rrbracket = & 0
\end{align}
\begin{align}
\mathbb{E}\left\llbracket \int_{\tau}^{\tau+\delta\tau}dW_{\pm}^{\mu}(\tau',\bullet)\times\int_{\tau}^{\tau+\delta\tau}dW_{\pm}^{\nu}(\tau'',\bullet)\right\rrbracket  & =\delta^{\mu\nu}\times\delta\tau
\end{align}
\end{leftbar}
\end{lem}

\begin{thm}
\begin{leftbar}The It\^{o} formula of a $C^{2}$-function $f:\mathbb{A}^{4}(\mathbb{V_{\mathrm{M}}^{\mathrm{4}}},g)\rightarrow\mathbb{C}$
w.r.t.\ $W_{\pm}(\circ,\omega)$ is
\begin{align}
f(W_{\pm}(\tau_{b},\omega))-f(W_{\pm}(\tau_{a},\omega)) & =\int_{\tau_{a}}^{\tau_{b}}\partial_{\mu}f(W_{\pm}(\tau,\omega))dW_{\pm}^{\mu}(\tau,\omega)\nonumber \\
 & \quad\quad\mp\frac{\lambda^{2}}{2}\int_{\tau_{a}}^{\tau_{b}}\partial_{\mu}\partial^{\mu}f(W_{\pm}(\tau,\omega))d\tau\,\,\,\,\mathrm{a.s.}
\end{align}
\end{leftbar}
\end{thm}

Where, we emphasize that $\partial_{\mu}\partial^{\mu}f(W_{\pm}(\tau,\omega))$
appears in the It\^{o} formula. 

\subsection{D-prog.\ $\hat{x}(\circ,\bullet)$}

Then the stochastic differential equation
\[
d_{\pm}\hat{x}^{\mu}(\tau,\omega)=\mathcal{V}_{\pm}^{\mu}(\hat{x}(\tau,\omega))d\tau+\lambda\times dW_{\pm}^{\mu}(\tau,\omega)
\]
is defined as a relativistic kinematics of a stochastic scalar electron.
Let us follow the construction by Nelson \cite{Nelson(2001_book)},
however, it has to be on $(\mathbb{A}^{4}(\mathbb{V_{\mathrm{M}}^{\mathrm{4}}},g),\mathscr{B}(\mathbb{A}^{4}(\mathbb{V_{\mathrm{M}}^{\mathrm{4}}},g)),\mu)$.
\begin{defn}[(R0)-process]
\begin{leftbar}For $\left(\mathit{\Omega},\mathcal{F},\mathscr{P}\right)$,
a $\mathscr{B}(\mathbb{R})\times\mathcal{F}/\mathscr{B}(\mathbb{A}^{4}(\mathbb{V_{\mathrm{M}}^{\mathrm{4}}},g))$-measurable
$\hat{x}(\circ,\bullet)$ is a 4D (R0)-process if each of $\{\varphi_{\mathbb{A}^{4}(\mathbb{V_{\mathrm{M}}^{\mathrm{4}}},g)}^{\mu}\circ\hat{x}(\tau,\bullet)\}_{\mu=0,1,2,3}$
belongs to $L^{1}(\varOmega,\mathscr{P})$ for all $\tau$ and the
mapping $\tau\mapsto\hat{x}(\tau,\omega)$ is almost surely continuous.\end{leftbar}
\end{defn}

By employing $\mathbb{L}_{T}^{p}(E)$  as a family of $\mathscr{B}(\mathbb{R})\times\mathcal{F}/\mathscr{B}(E)$-measurable
mappings for a topological space $E$, let us introduce $\mathcal{L}_{\mathrm{loc}}^{p}(\{\mathscr{P}_{\tau}\};E)$
and $\mathcal{L}_{\mathrm{loc}}^{p}(\{\mathscr{F}_{\tau}\};E)$:
\begin{align*}
\mathbb{L}_{T}^{p}(E) & \coloneqq\left\{ \hat{X}(\circ,\bullet):\mathbb{R}\times\varOmega\rightarrow E\left|\begin{gathered}\tau\rightarrow\hat{X}(\tau,\omega)\mathrm{\,is\,continuous},\\
\sum_{A}\int_{T\subset\mathbb{R}}|\varphi_{E}^{A}\circ\hat{X}(\tau,\omega)|^{p}d\tau<\infty\,\mathrm{a.s.}
\end{gathered}
\right.\right\} 
\end{align*}
\begin{align*}
\mathcal{L}_{\mathrm{loc}}^{p}(\{\mathscr{P}_{\tau}\};E) & \coloneqq\left\{ \hat{X}(\circ,\bullet)\in\mathbb{L}_{[\tau_{1},\tau_{2}]}^{p}(E)\left|\begin{gathered}\forall\tau_{1}\leq\forall\tau_{2}\in\mathbb{R},\,\begin{gathered}\hat{X}(\circ,\bullet)\mathrm{\,is\,}\{\mathscr{P}_{\tau}\}\mathchar`-\mathrm{adapted}.\end{gathered}
\end{gathered}
\right.\right\} 
\end{align*}
\begin{align*}
\mathcal{L}_{\mathrm{loc}}^{p}(\{\mathscr{F}_{\tau}\};E) & \coloneqq\left\{ \hat{X}(\circ,\bullet)\in\mathbb{L}_{[\tau_{1},\tau_{2}]}^{p}(E),\left|\begin{gathered}\forall\tau_{1}\leq\forall\tau_{2}\in\mathbb{R},\,\hat{X}(\circ,\bullet)\mathrm{\,is\,}\{\mathscr{F}_{\tau}\}\mathchar`-\mathrm{adapted.}\end{gathered}
\right.\right\} 
\end{align*}
\begin{defn}[$\{\mathcal{\mathscr{P}}_{\tau}\}$-prog.\ and $\{\mathcal{\mathscr{F}}_{\tau}\}$-prog.]
\begin{leftbar}For all $\sigma\leq\tau\in\mathbb{R}$, let $\{\mathcal{\mathscr{P}}_{\tau}\}$-prog.\ on
$E$ be a $\mathscr{B}([\sigma,\tau])\times\mathcal{\mathscr{P}}_{\tau}/\mathscr{B}(E)$-measurable
process and $\{\mathcal{\mathscr{F}}_{\tau}\}$-prog.\ on $E$ be
a $\mathscr{B}([\sigma,\tau])\times\mathcal{\mathscr{F}}_{\tau}/\mathscr{B}(E)$-measurable
process.\end{leftbar}
\end{defn}

For the later discussion, $\mathring{\epsilon}$ is defined as $\mathring{\epsilon}=1$
when a component $\hat{x}^{\mu}(\circ,\bullet)=\varphi_{\mathbb{A}^{4}(\mathbb{V_{\mathrm{M}}^{\mathrm{4}}},g)}^{\mu}\circ\hat{x}(\circ,\bullet)$
as a 1D stochastic process is $\{\mathcal{P}_{\tau}\}$-adapted and
$\mathring{\epsilon}=-1$ if $\hat{x}^{\mu}(\circ,\bullet)$ is $\{\mathcal{F}_{\tau}\}$-adapted.
Namely for $\{V^{\mu}(\circ,\omega)\}_{\mu=0,1,2,3}$, let us regard
the above as follows:
\[
\mathbb{E}\llbracket V((\tau,\omega))|\mathcal{\mathscr{P}}_{\tau}\rrbracket\coloneqq\left\{ \begin{array}{c}
\begin{gathered}\mathbb{E}\llbracket V^{0}((\circ,\omega))|\mathcal{F}_{\tau}\rrbracket(\omega)\end{gathered}
\\
\begin{gathered}\mathbb{E}\llbracket V^{i=1,2,3}((\circ,\omega))|\mathcal{P}_{\tau}\rrbracket(\omega)\end{gathered}
\end{array}\right.
\]
\[
\mathbb{E}\llbracket V((\tau,\omega))|\mathcal{\mathscr{F}}_{\tau}\rrbracket\coloneqq\left\{ \begin{array}{c}
\begin{gathered}\mathbb{E}\llbracket V^{0}((\circ,\omega))|\mathcal{P}_{\tau}\rrbracket(\omega)\end{gathered}
\\
\begin{gathered}\mathbb{E}\llbracket V^{i=1,2,3}((\circ,\omega))|\mathcal{F}_{\tau}\rrbracket(\omega)\end{gathered}
\end{array}\right.
\]
The mean derivatives (see Eq.(\ref{eq: mean-deriv})) are mathematically
introduced by those ideas:
\begin{defn}[(R1)-process]
\begin{leftbar}If $\hat{x}(\circ,\bullet)$ is an (R0)-process and
a following $\mathcal{V}_{+}(\hat{x}(\circ,\bullet))\in\mathcal{L}_{\mathrm{loc}}^{1}(\{\mathscr{P}_{\tau}\};\mathbb{V}_{\mathrm{M}}^{\mathrm{4}})$
exists, $\hat{x}(\circ,\bullet)$ is an (R1)-process.
\begin{alignat}{1}
\mathcal{V}_{+}(\hat{x}(\tau,\omega)) & \coloneqq\underset{\delta t\rightarrow0+}{\lim}\mathbb{E}\left\llbracket \left.\frac{\hat{x}(\tau+\mathring{\epsilon}\times\delta\tau,\bullet)-\hat{x}(\tau,\bullet)}{\mathring{\epsilon}\delta\tau}\right|\mathcal{\mathscr{P}}_{\tau}\right\rrbracket (\omega)
\end{alignat}
\end{leftbar}
\end{defn}

The each components of $\mathcal{V}_{+}(\hat{x}(\tau,\omega))$ are
interpreted as below:
\begin{align*}
\mathcal{V}_{+}^{0}(\hat{x}(\tau,\omega)) & =\underset{\delta t\rightarrow0+}{\lim}\mathbb{E}\left\llbracket \left.\frac{\hat{x}^{0}(\tau,\bullet)-\hat{x}^{0}(\tau-\delta\tau,\bullet)}{\delta\tau}\right|\mathcal{F}_{\tau}\right\rrbracket (\omega)
\end{align*}
\begin{align*}
\mathcal{V}_{+}^{i=1,2,3}(\hat{x}(\tau,\omega)) & =\underset{\delta t\rightarrow0+}{\lim}\mathbb{E}\left\llbracket \left.\frac{\hat{x}^{i}(\tau+\delta\tau,\bullet)-\hat{x}^{i}(\tau,\bullet)}{\delta\tau}\right|\mathcal{P}_{\tau}\right\rrbracket (\omega)
\end{align*}
\begin{defn}[(S1)-process]
\begin{leftbar}If $\hat{x}(\circ,\bullet)$ is an (R1)-process and
a following $\mathcal{V}_{-}(\hat{x}(\circ,\bullet))\in\mathcal{L}_{\mathrm{loc}}^{1}(\{\mathscr{F}_{\tau}\};\mathbb{V}_{\mathrm{M}}^{\mathrm{4}})$
exists, $\hat{x}(\circ,\bullet)$ is named an (S1)-process. 
\begin{align}
\mathcal{V}_{-}(\hat{x}(\tau,\omega)) & \coloneqq\underset{\delta t\rightarrow0+}{\lim}\mathbb{E}\left\llbracket \left.\frac{\hat{x}(\tau+\mathring{\epsilon}\times\delta\tau,\bullet)-\hat{x}(\tau,\bullet)}{\mathring{\epsilon}\delta\tau}\right|\mathscr{F_{\tau}}\right\rrbracket (\omega)
\end{align}
\end{leftbar}
\end{defn}

The definition of an (S1)-process declare that this is $\{\mathscr{P}_{\tau}\}$-prog.\ and
$\{\mathscr{F}_{\tau}\}$-prog. Therefore, an (S1)-process provides
us the form of the stochastic integral on $\tau_{a}\leq\tau_{b}$:
\begin{align}
\hat{x}^{\mu}(\tau_{b},\omega)-\hat{x}^{\mu}(\tau_{a},\omega) & =\int_{\tau_{a}}^{\tau_{b}}d\tau'\mathcal{V}_{+}^{\mu}(\hat{x}(\tau',\omega))+\int_{\tau_{a}}^{\tau_{b}}dy_{+}^{\mu}(\tau',\omega)\label{eq: S1-a}\\
 & =\int_{\tau_{a}}^{\tau_{b}}d\tau'\mathcal{V}_{-}^{\mu}(\hat{x}(\tau',\omega))+\int_{\tau_{a}}^{\tau_{b}}dy_{-}^{\mu}(\tau',\omega)\label{eq: S1-b}
\end{align}
Where, $y_{+}(\circ,\bullet)$ and $y_{-}(\circ,\bullet)$ of martingales
satisfy below:
\begin{defn}[(R2)-process]
\begin{leftbar}When $\hat{x}(\circ,\bullet)$ is an (R1)-process
and let $y_{+}(\circ,\bullet)$ be its $\{\mathscr{P}_{\tau}\}$-martingale
part such that  $y_{+}(\tau+\mathring{\epsilon}\times\delta\tau,\bullet)-y_{+}(\tau,\bullet)\in\mathcal{L}_{\mathrm{loc}}^{2}(\{\mathscr{P}_{\tau}\};\mathbb{V}_{\mathrm{M}}^{\mathrm{4}})$,
then, $\hat{x}(\circ,\bullet)$ is named an (R2)-process if 
\begin{equation}
\mathbb{E}\left\llbracket \left.y_{+}(\tau+\mathring{\epsilon}\times\delta\tau,\bullet)-y_{+}(\tau,\bullet)\right|\mathcal{\mathscr{P}}_{\tau}\right\rrbracket (\omega)=0
\end{equation}
and a following $\sigma_{+}^{2}(\tau,\bullet)\in\mathcal{L}_{\mathrm{loc}}^{1}(\{\mathscr{P}_{\tau}\};\mathbb{V}_{\mathrm{M}}^{\mathrm{4}}\otimes\mathbb{V}_{\mathrm{M}}^{\mathrm{4}})$
exists such as $\tau\mapsto\sigma_{+}^{2}(\tau,\omega)$ is continuous:
\begin{equation}
\sigma_{+}^{2}(\tau,\omega)\coloneqq\underset{\delta t\rightarrow0+}{\lim}\mathbb{E}\left\llbracket \left.\frac{[y_{+}(\tau+\mathring{\epsilon}\times\delta\tau,\bullet)-y_{+}(\tau,\bullet)]\otimes[y_{+}(\tau+\mathring{\epsilon}\times\delta\tau,\bullet)-y_{+}(\tau,\bullet)]}{\delta\tau}\right|\mathcal{\mathscr{P}}_{\tau}\right\rrbracket (\omega)
\end{equation}
\end{leftbar}
\end{defn}

\begin{defn}[(S2)-process]
\begin{leftbar}When $\hat{x}(\circ,\bullet)$ is an (R2) and (S1)-process
and let $y_{-}(\circ,\bullet)$ be its $\{\mathscr{F}_{\tau}\}$-martingale
part $y_{-}(\tau+\mathring{\epsilon}\times\delta\tau,\bullet)-y_{-}(\tau,\bullet)\in\mathcal{L}_{\mathrm{loc}}^{2}(\{\mathscr{F}_{\tau}\};\mathbb{V}_{\mathrm{M}}^{\mathrm{4}})$,
then, $\hat{x}(\circ,\bullet)$ is called an (S2)-process if 
\begin{equation}
\mathbb{E}\left\llbracket \left.y_{-}(\tau+\mathring{\epsilon}\times\delta\tau,\bullet)-y_{-}(\tau,\bullet)\right|\mathcal{\mathscr{F}}_{\tau}\right\rrbracket (\omega)=0\,,
\end{equation}
and a following $\sigma_{-}^{2}(\tau,\bullet)\in\mathcal{L}_{\mathrm{loc}}^{1}(\{\mathscr{F}_{\tau}\};\mathbb{V}_{\mathrm{M}}^{\mathrm{4}}\otimes\mathbb{V}_{\mathrm{M}}^{\mathrm{4}})$
exists such that $\tau\mapsto\sigma_{-}^{2}(\tau,\omega)$ is continuous:
\begin{equation}
\sigma_{-}^{2}(\tau,\omega)\coloneqq\underset{\delta t\rightarrow0+}{\lim}\mathbb{E}\left\llbracket \left.\frac{[y_{-}(\tau+\mathring{\epsilon}\times\delta\tau,\bullet)-y_{-}(\tau,\bullet)]\otimes[y_{-}(\tau+\mathring{\epsilon}\times\delta\tau,\bullet)-y_{-}(\tau,\bullet)]}{\delta\tau}\right|\mathcal{\mathscr{F}}_{\tau}\right\rrbracket (\omega)
\end{equation}
\end{leftbar}
\end{defn}

\begin{defn}[(R3)-process]
\begin{leftbar}If $\hat{x}(\circ,\bullet)$ is an (R2)-process and
$\det\sigma_{+}^{2}(\tau,\omega)>0$ is almost surely satisfied for
all $\tau\in\mathbb{R}$, then, $\hat{x}(\circ,\bullet)$ is named
an (R3)-process.\end{leftbar}
\end{defn}

\begin{defn}[(S3)-process]
\begin{leftbar}If $\hat{x}(\circ,\bullet)$ is an (R3) and (S2)-process,
and $\det\sigma_{-}^{2}(\tau,\omega)>0$ is almost surely satisfied
for all $\tau\in\mathbb{R}$, then, $\hat{x}(\circ,\bullet)$ is called
an(S3)-process.\end{leftbar}
\end{defn}

Where, $y_{\pm}(\tau,\omega)\coloneqq\lambda\times W_{\pm}(\tau,\omega)$
for $\lambda>0$ satisfies the above definition of the (S3) process,
i.e.,
\begin{align}
\left.\det\sigma_{\pm}(\tau,\omega)\right|_{y_{\pm}(\tau,\omega)=\lambda\times W_{\pm}(\tau,\omega)} & =4\times\lambda^{2}>0.
\end{align}

\begin{defn}[D-prog.\  $\hat{x}(\circ,\bullet)$]
\begin{leftbar}\label{D-progressive}An (S3)-process $\hat{x}(\circ,\bullet)$
on $\left(\mathit{\Omega},\mathcal{F},\mathscr{P}\right)$ is named
``D-prog.'' if $y_{\pm}(\circ,\bullet)\coloneqq\lambda\times W_{\pm}(\circ,\bullet)$
w.r.t.\ $\lambda>0$. $\hat{x}(\circ,\bullet)$ of a D-prog.\  is
given by the following stochastic differential equation:
\begin{equation}
\boxed{\ensuremath{d\hat{x}^{\mu}(\tau,\omega)=\mathcal{V}_{\pm}^{\mu}(\hat{x}(\tau,\omega))d\tau+\lambda\times dW_{\pm}^{\mu}(\tau,\omega)}}\label{eq:D-prog}
\end{equation}
\end{leftbar} 
\end{defn}

For $\tau_{a}\leq\tau_{b}$, $\hat{x}(\circ,\bullet)$ of D-prog.\ by
Eq.(\ref{eq:D-prog}) is regarded as the symbolic expression w.r.t.\ the
following integral: 
\begin{align}
\hat{x}^{\mu}(\tau_{b},\omega)-\hat{x}^{\mu}(\tau_{a},\omega) & =\int_{\tau_{a}}^{\tau_{b}}d\tau'\mathcal{V}_{+}^{\mu}(\hat{x}(\tau',\omega))+\int_{\tau_{a}}^{\tau_{b}}dW_{+}^{\mu}(\tau',\omega)\label{eq: Ito-path1}\\
 & =\int_{\tau_{a}}^{\tau_{b}}d\tau'\mathcal{V}_{-}^{\mu}(\hat{x}(\tau',\omega))+\int_{\tau_{a}}^{\tau_{b}}dW_{-}^{\mu}(\tau',\omega)\label{eq: Ito-path2}
\end{align}
Let $d_{\pm}\hat{x}^{\mu}(\tau,\omega)$ be defined by $\int_{\tau}^{\tau+\epsilon\mathring{\epsilon}\times\delta\tau}d_{\epsilon}\hat{x}^{\mu}(\tau',\omega)\coloneqq\hat{x}^{\mu}(\tau+\epsilon\mathring{\epsilon}\times\delta\tau,\omega)-\hat{x}^{\mu}(\tau,\omega)$
with $\epsilon=\pm$ for $d_{\pm}$. Since $\hat{x}(\tau,\bullet)$
is $\mathcal{\mathscr{P}}_{\tau}/\mathscr{B}(\mathbb{A}^{4}(\mathbb{V_{\mathrm{M}}^{\mathrm{4}}},g))$
and $\mathcal{\mathscr{F}}_{\tau}/\mathscr{B}(\mathbb{A}^{4}(\mathbb{V_{\mathrm{M}}^{\mathrm{4}}},g))$-measurable
for all $\tau$, the following is imposed:
\begin{thm}
\begin{leftbar}\label{D-progressive_2}A D-progressive $\hat{x}(\circ,\bullet)$
is continuous and $\{\mathcal{\mathscr{P}}_{\tau}\cap\mathcal{\mathscr{F}}_{\tau}\}$-adapted.\end{leftbar} 
\end{thm}

Finally, the construction of our Brownian and relativistic kinematics
is mathematically completed by defining $\lambda$. The following
conjecture is demonstrated by Eqs.(\ref{eq: EOM_St}) of Issue-$\langle\mathrm{A}\rangle$
in the main body.
\begin{conjecture}
\begin{leftbar}\label{conj_kinematics}A D-prog.\ $\hat{x}(\circ,\bullet)$
is a trajectory of a scalar electron satisfying a Klein-Gordon equation
when $\lambda=\sqrt{\hbar/m_{0}}$.\end{leftbar} 
\end{conjecture}

$\hat{x}(\circ,\bullet)$ of a D-prog.\ imposes the following It\^{o}
formula.
\begin{thm}[It\^{o} formula]
\begin{leftbar}\label{Ito formula}Consider a $C^{2}$-function
$f:\mathbb{A}^{4}(\mathbb{V_{\mathrm{M}}^{\mathrm{4}}},g)\rightarrow\mathbb{C}$,
the following It\^{o} formula w.r.t. a D-prog.\ $\hat{x}(\circ,\bullet)$
is found;
\begin{align}
d_{\pm}f(\hat{x}(\tau,\omega)) & =\partial_{\mu}f(\hat{x}(\tau,\omega))d_{\pm}\hat{x}^{\mu}(\tau,\omega)\mp\frac{\lambda^{2}}{2}\partial_{\mu}\partial^{\mu}f(\hat{x}(\tau,\omega))d\tau\,\,\,\,\mathrm{a.s.}
\end{align}
 This is given by the following stochastic integral:
\begin{align}
f(\hat{x}(\tau_{b},\omega)) & -f(\hat{x}(\tau_{a},\omega))\nonumber \\
 & =\int_{\tau_{a}}^{\tau_{b}}d_{\pm}f(\hat{x}(\tau,\omega))\\
 & =\int_{\tau_{a}}^{\tau_{b}}\partial_{\mu}f(\hat{x}(\tau,\omega))d_{\pm}\hat{x}^{\mu}(\tau,\omega)\mp\frac{\lambda^{2}}{2}\int_{\tau_{a}}^{\tau_{b}}\partial_{\mu}\partial^{\mu}f(\hat{x}(\tau,\omega))d\tau\,\,\,\,\mathrm{a.s.}
\end{align}
\end{leftbar}
\end{thm}

\begin{defn}[Mean derivatives $\mathfrak{D}_{\tau}^{\pm}$]
\begin{leftbar}$\mathfrak{D}_{\tau}^{+}$ and $\mathfrak{D}_{\tau}^{-}$
the operators of the mean derivatives are defined by the following:
\begin{align}
\mathfrak{D}_{\tau}^{\pm}f(\hat{x}(\tau,\omega)) & \coloneqq\left[\mathcal{V}_{\pm}^{\mu}(\hat{x}(\tau,\omega))\mp\frac{\lambda^{2}}{2}\partial^{\mu}\right]\partial_{\mu}f(\hat{x}(\tau,\omega))
\end{align}
By using this, 
\begin{align}
\int_{\tau_{a}}^{\tau_{b}}d_{\pm}f(\hat{x}(\tau,\omega)) & =\int_{\tau_{a}}^{\tau_{b}}\mathfrak{D}_{\tau}^{\pm}f(\hat{x}(\tau,\omega))d\tau+\lambda\times\int_{\tau_{a}}^{\tau_{b}}\partial_{\mu}f(\hat{x}(\tau,\omega))dW_{\pm}^{\mu}(\tau,\omega).
\end{align}
\end{leftbar}
\end{defn}

\subsection{Complexification of the evolution of a D-prog.}

In quantum dynamics, complex-valued $\mathcal{V}(\hat{x}(\tau,\omega))$
given by Eq.(\ref{eq: Comp-V}) is required to describe a wave function.
Since $f(\hat{x}(\tau_{b},\omega))-f(\hat{x}(\tau_{a},\omega))=\int_{\tau_{a}}^{\tau_{b}}d_{\pm}f(\hat{x}(\tau,\omega))$,
\begin{align}
f(\hat{x}(\tau_{b},\omega))-f(\hat{x}(\tau_{a},\omega)) & =\frac{1}{2}\int_{\tau_{a}}^{\tau_{b}}[d_{+}f(\hat{x}(\tau,\omega))+d_{-}f(\hat{x}(\tau,\omega))]\\
 & \quad\quad-\frac{i}{2}\int_{\tau_{a}}^{\tau_{b}}[d_{+}f(\hat{x}(\tau,\omega))-d_{-}f(\hat{x}(\tau,\omega))].
\end{align}
is expected. Let us argue this idea for $\mathcal{V}(\hat{x}(\tau,\omega))$.
\begin{defn}[Complex differential $\hat{d}$]
\begin{leftbar}\label{C-Ito formula}For a D-prog.\ $\hat{x}(\circ,\bullet)$
and a $C^{2}$-function $f:\mathbb{A}^{4}(\mathbb{V_{\mathrm{M}}^{\mathrm{4}}},g)\rightarrow\mathbb{R}$,
let $\hat{d}$ be a differential mapping as follows:
\begin{align}
\hat{d}f(\hat{x}(\tau,\omega)) & =\partial_{\mu}f(\hat{x}(\tau,\omega))\hat{d}\hat{x}^{\mu}(\tau,\omega)+\frac{i\lambda^{2}}{2}\partial^{\mu}\partial_{\mu}f(\hat{x}(\tau,\omega))d\tau
\end{align}
When $\int_{\tau_{a}}^{\tau_{b}}\hat{d}f(\hat{x}(\tau,\omega))$ satisfies
\begin{align}
\mathrm{Re}\left\{ \int_{\tau_{a}}^{\tau_{b}}\hat{d}f(\hat{x}(\tau,\omega))\right\} \mp\mathrm{Im}\left\{ \int_{\tau_{a}}^{\tau_{b}}\hat{d}f(\hat{x}(\tau,\omega))\right\}  & =\int_{\tau_{a}}^{\tau_{b}}d_{\pm}f(\hat{x}(\tau,\omega)),
\end{align}
$\hat{d}$ indicates $\hat{d}=(d_{+}+d_{-})/2-i(d_{+}-d_{-})/2$ symbolically.
Therefore, 
\begin{align}
f(\hat{x}(\tau_{b},\omega)) & =f(\hat{x}(\tau_{a},\omega))+\int_{\tau_{a}}^{\tau_{b}}\hat{d}f(\hat{x}(\tau,\omega))\,\,\,\,\mathrm{a.s.}
\end{align}
\end{leftbar}
\end{defn}

\
\begin{defn}[Mean derivative $\mathfrak{D}_{\tau}$]
\begin{leftbar} The operator $\mathfrak{D}_{\tau}$ is defined by
the following:
\begin{align}
\mathfrak{D}_{\tau}f(\hat{x}(\tau,\omega)) & \coloneqq\left[\mathcal{V}^{\mu}(\hat{x}(\tau,\omega))+\frac{i\lambda^{2}}{2}\partial^{\mu}\right]\partial_{\mu}f(\hat{x}(\tau,\omega))
\end{align}
By employing this operator
\begin{align}
\int_{\tau_{a}}^{\tau_{b}}\hat{d}f(\hat{x}(\tau,\omega)) & =\int_{\tau_{a}}^{\tau_{b}}\mathfrak{D}_{\tau}f(\hat{x}(\tau,\omega))d\tau+\lambda\times\int_{\tau_{a}}^{\tau_{b}}d_{\pm}f(\hat{x}(\tau,\omega))\hat{d}W^{\mu}(\tau,\omega).
\end{align}
\end{leftbar}
\end{defn}

Hence, $\mathfrak{D}_{\tau}\hat{x}(\tau,\omega)=\mathcal{V}(\hat{x}(\tau,\omega))$.
For $f(\hat{x}(\tau_{b},\omega))=f(\hat{x}(\tau_{a},\omega))+\int_{\tau_{a}}^{\tau_{b}}\hat{d}f(\hat{x}(\tau,\omega))$,
its iteration is 
\begin{multline}
f(\hat{x}(\tau_{b},\omega))=\sum_{n=0}^{\infty}\frac{(\tau_{b}-\tau_{a})^{n}}{n!}\times\mathfrak{D}_{\tau}^{n}f(\hat{x}(\tau_{a},\omega))\\
+\lambda\times\sum_{n=0}^{\infty}\int_{\tau_{a}}^{\tau_{b}}d\tau_{1}\int_{\tau_{a}}^{\tau_{1}}d\tau_{2}\cdots\int_{\tau_{a}}^{\tau_{n-2}}d\tau_{n-1}\int_{\tau_{a}}^{\tau_{n-1}}\hat{d}W^{\mu}(\tau_{n},\omega)\cdot\partial_{\mu}\mathfrak{D}_{\tau}^{n}f(\hat{x}(\tau_{n},\omega)).
\end{multline}
We named it the stochastic-Taylor expansion, however strictly speaking,
the second term in the RHS has to be expanded more at $\tau_{a}$,
too. 

\section{Nelson's partial integral formula\label{APP-B}}

Let us demonstrate Nelson's partial integral formula of Eqs.(\ref{eq: Nelson-partial integral1},\ref{eq: Nelson-partial integral2})
\begin{lem}[Nelson's partial integral formula]
\begin{leftbar}\label{Nelson_partial}For $(\mathbb{A}^{4}(\mathbb{V_{\mathrm{M}}^{\mathrm{4}}},g),\mathscr{B}(\mathbb{A}^{4}(\mathbb{V_{\mathrm{M}}^{\mathrm{4}}},g)),\mu)$,
let $\alpha$ and $\beta$ be the complex-valued and $C^{2}$-local
square integrable functions on $\mathbb{A}^{4}(\mathbb{V_{\mathrm{M}}^{\mathrm{4}}},g)$.
Then, the following partial integral formula is fulfilled:
\begin{multline}
\int_{\tau_{1}}^{\tau_{2}}d\tau\,\mathbb{E}\left\llbracket \mathfrak{D}_{\tau}^{\pm}\alpha_{\mu}(\hat{x}(\tau,\bullet))\cdot\beta^{\mu}(\hat{x}(\tau,\bullet))+\alpha_{\mu}(\hat{x}(\tau,\bullet))\cdot\mathfrak{D}_{\tau}^{\mp}\beta^{\mu}(\hat{x}(\tau,\bullet))\right\rrbracket \\
=\mathbb{E}\left\llbracket \alpha_{\mu}(\hat{x}(\tau_{2},\bullet))\beta^{\mu}(\hat{x}(\tau_{2},\bullet))-\alpha_{\mu}(\hat{x}(\tau_{1},\bullet))\beta^{\mu}(\hat{x}(\tau_{1},\bullet))\right\rrbracket \label{eq: partial int formula original1}
\end{multline}
 Its differential form is
\begin{align}
\frac{d}{d\tau}\mathbb{E}\left\llbracket \alpha_{\mu}(\hat{x}(\tau,\bullet))\beta^{\mu}(\hat{x}(\tau,\bullet))\right\rrbracket  & =\mathbb{E}\left\llbracket \begin{gathered}\mathfrak{D}_{\tau}^{\pm}\alpha_{\mu}(\hat{x}(\tau,\bullet))\cdot\beta^{\mu}(\hat{x}(\tau,\bullet))\\
+\alpha_{\mu}(\hat{x}(\tau,\bullet))\cdot\mathfrak{D}_{\tau}^{\mp}\beta^{\mu}(\hat{x}(\tau,\bullet))
\end{gathered}
\right\rrbracket .\label{eq: partial int formula original2}
\end{align}
\end{leftbar}
\end{lem}

\begin{proof}
Confirm the following relation at first.
\begin{multline}
\mathbb{E}\left\llbracket \mathfrak{D}_{\tau}^{+}\alpha_{\mu}(\hat{x}(\tau,\bullet))\cdot\beta^{\mu}(\hat{x}(\tau,\bullet))+\alpha_{\mu}(\hat{x}(\tau,\bullet))\cdot\mathfrak{D}_{\tau}^{-}\beta^{\mu}(\hat{x}(\tau,\bullet))\right\rrbracket \\
=\mathbb{E}\left\llbracket \mathfrak{D}_{\tau}^{-}\alpha_{\mu}(\hat{x}(\tau,\bullet))\cdot\beta^{\mu}(\hat{x}(\tau,\bullet))+\alpha_{\mu}(\hat{x}(\tau,\bullet))\cdot\mathfrak{D}_{\tau}^{+}\beta^{\mu}(\hat{x}(\tau,\bullet))\right\rrbracket \label{eq: D+D-=00003DD-D+}
\end{multline}
since
\begin{align*}
\mathbb{E} & \left\llbracket \mathfrak{D}_{\tau}^{+}\alpha_{\mu}(\hat{x}(\tau,\bullet))\cdot\beta^{\mu}(\hat{x}(\tau,\bullet))+\alpha_{\mu}(\hat{x}(\tau,\bullet))\cdot\mathfrak{D}_{\tau}^{-}\beta^{\mu}(\hat{x}(\tau,\bullet))\right\rrbracket \\
 & \quad-\mathbb{E}\left\llbracket \mathfrak{D}_{\tau}^{-}\alpha_{\mu}(\hat{x}(\tau,\bullet))\cdot\beta^{\mu}(\hat{x}(\tau,\bullet))+\alpha_{\mu}(\hat{x}(\tau,\bullet))\cdot\mathfrak{D}_{\tau}^{+}\beta^{\mu}(\hat{x}(\tau,\bullet))\right\rrbracket \\
 & =-\lambda^{4}\times\int_{\mathbb{A}^{4}(\mathbb{V_{\mathrm{M}}^{\mathrm{4}}},g)}d\mu(x)\partial^{\nu}\left\{ p(x,\tau)\left[\partial_{\nu}\alpha_{\mu}(x)\cdot\beta^{\mu}(x)-\alpha_{\mu}(x)\cdot\partial_{\nu}\beta^{\mu}(x)\right]\right\} \\
 & =0.
\end{align*}
Then by using the Fokker-Planck equation (\ref{eq: Fokker-Planck}),
\begin{align*}
\frac{d}{d\tau}\mathbb{E}\left\llbracket \alpha_{\mu}(\hat{x}(\tau,\bullet))\beta^{\mu}(\hat{x}(\tau,\bullet))\right\rrbracket  & =\frac{1}{2}\times\mathbb{E}\left\llbracket \begin{gathered}(\mathfrak{D}_{\tau}^{+}+\mathfrak{D}_{\tau}^{-})\alpha_{\mu}(\hat{x}(\tau,\bullet))\cdot\beta^{\mu}(\hat{x}(\tau,\bullet))\\
+\alpha_{\mu}(\hat{x}(\tau,\bullet))\cdot(\mathfrak{D}_{\tau}^{+}+\mathfrak{D}_{\tau}^{-})\beta^{\mu}(\hat{x}(\tau,\bullet))
\end{gathered}
\right\rrbracket .
\end{align*}
Where, we employ $\int_{\varOmega}d\mathscr{P}(\omega)=\int_{\mathbb{A}^{4}(\mathbb{V_{\mathrm{M}}^{\mathrm{4}}},g)}p(x,\tau)d\mu(x)$.
By combining it with Eq.(\ref{eq: D+D-=00003DD-D+}), Eq.(\ref{eq: partial int formula original2})
is demonstrated.
\end{proof}
By using the superposition of the above ``$\pm$''-formulas, it
can be switched to Eq.(\ref{eq: Partial int formula1}-\ref{eq: Partial int formula2})
the formula for the complex derivatives of $\mathfrak{D}_{\tau}$
and $\mathfrak{D}_{\tau}^{*}$.
\begin{align}
\frac{d}{d\tau}\mathbb{E}\left\llbracket \alpha_{\mu}(\hat{x}(\tau,\bullet))\beta^{\mu}(\hat{x}(\tau,\bullet))\right\rrbracket  & =\mathbb{E}\left\llbracket \mathfrak{D}_{\tau}\alpha_{\mu}(\hat{x}(\tau,\bullet))\cdot\beta^{\mu}(\hat{x}(\tau,\bullet))+\alpha_{\mu}(\hat{x}(\tau,\bullet))\cdot\mathfrak{D}_{\tau}^{*}\beta^{\mu}(\hat{x}(\tau,\bullet))\right\rrbracket \label{eq: Partial int formula1}\\
 & =\mathbb{E}\left\llbracket \mathfrak{D}_{\tau}^{*}\alpha_{\mu}(\hat{x}(\tau,\bullet))\cdot\beta^{\mu}(\hat{x}(\tau,\bullet))+\alpha_{\mu}(\hat{x}(\tau,\bullet))\cdot\mathfrak{D}_{\tau}\beta^{\mu}(\hat{x}(\tau,\bullet))\right\rrbracket \label{eq: Partial int formula2}
\end{align}

\end{document}